\RequirePackage{xcolor}
\documentclass[journal,twoside,web]{ieeecolor}
\usepackage{cite}
\usepackage{generic}
\usepackage{algorithm}
\usepackage[noend]{algpseudocode}

\usepackage{graphicx}
\usepackage{textcomp}
\usepackage{booktabs}
\usepackage{arydshln}
\def\BibTeX{{\rm B\kern-.05em{\sc i\kern-.025em b}\kern-.08em T\kern-.1667em\lower.7ex\hbox{E}\kern-.125emX}}

\usepackage[cmex10]{amsmath}
\usepackage{amssymb, mathrsfs}
\usepackage{mathtools}			
\usepackage{microtype}
\usepackage{enumerate, paralist}
\usepackage{refcount}

\usepackage{afterpage}

\usepackage{amsthm}
\newtheorem{proposition}{Proposition}
\newtheorem{lemma}[proposition]{Lemma}
\newtheorem{theorem}[proposition]{Theorem}
\newtheorem{corollary}{Corollary}[proposition]
\newtheorem{claim}{Claim}[proposition]
\theoremstyle{definition}

\newtheorem{assumption}[proposition]{Assumption}
\theoremstyle{remark}
\newtheorem{remark}{Remark}

\renewcommand{\qedsymbol}{\rule{.7em}{.7em}}

\newcommand{\oprocend}{\relax\ifmmode\else\unskip\hfill\fi\qedsymbol}

\newcommand{\col}{\mathrm{col}}
\newcommand{\Diag}{\mathrm{Diag}}
\newcommand{\mat}[1]{\begin{bmatrix}#1\end{bmatrix}}
\newcommand{\neatmat}[1]{\text{\small$\mat{#1}$}}
\newcommand{\compactmat}[1]{\text{\footnotesize$\mat{#1}$}}
\newcommand{\smallmat}[1]{\begin{bsmallmatrix}#1\end{bsmallmatrix}}
\newcommand{\compact}[1]{\text{\footnotesize$#1$}}
\newcommand{\Hank}{\mathcal{H}}

\newcommand{\real}{\mathbb{R}}
\newcommand{\integer}{\mathbb{Z}}
\renewcommand{\natural}{\mathbb{N}}
\newcommand{\complex}{\mathbb{C}}
\newcommand{\symmetric}{\mathbb{S}}
\newcommand{\normal}{\mathcal{N}}
\newcommand{\simiid}{\overset{\rm i.i.d.}{\sim}}
\newcommand{\minus}{{\scalebox{1.2}[1] -}}
\newcommand{\plus}{{\scalebox{.7} +}}
\newcommand{\transpose}{{\sf T}}
\newcommand{\initial}{{\sf ini}}

\newcommand{\auxiliary}{{\sf aux}}
\newcommand{\original}{{\sf orig}}
\newcommand{\orthogonal}{{\sf orth}}
\newcommand{\data}{{\rm d}}
\usepackage{accents}
\newcommand\thickbar[1]{\accentset{\rule{.4em}{.8pt}}{#1}}

\DeclareMathOperator*{\minimize}{minimize}
\DeclareMathOperator*{\argmin}{argmin}

\newcommand{\tb}{\color{blue}}

\usepackage{soul}

\renewcommand{\tb}{}

\newif\ifversionA
\newif\ifversionB
\versionBtrue
\ifversionA\newcommand{\Versions}[2]{#1}\fi
\ifversionB\newcommand{\Versions}[2]{#2}\fi

\begin{document}

\title{Stochastic Data-Driven Predictive Control \\ with Equivalence to Stochastic MPC}

\author{Ruiqi Li,  \IEEEmembership{Student Member, IEEE}, John W. Simpson-Porco, \IEEEmembership{Senior Member, IEEE}, and Stephen L.\ Smith,  \IEEEmembership{Senior Member, IEEE}
\thanks{This research is supported in part by the Natural Sciences and Engineering Research Council of Canada (NSERC).}%
\thanks{Ruiqi Li and Stephen L. Smith are with the Electrical and Computer Engineering at the University of Waterloo, Waterloo, ON, Canada
{\tt\small \{r298li,stephen.smith\}@uwaterloo.ca}}%
\thanks{John W. Simpson-Porco is with the Department of Electrical and Computer Engineering at the University of Toronto, Toronto, ON, Canada
{\tt\small jwsimpson@ece.utoronto.ca}}
}

\maketitle
\thispagestyle{empty}
\pagestyle{empty}

\begin{abstract}
    We propose a data-driven receding-horizon control method dealing with the chance-constrained output-tracking problem of unknown stochastic linear time-invariant (LTI) systems.
    The proposed method takes into account the statistics of the process noise, the measurement noise and the uncertain initial condition, following an analogous framework to Stochastic Model Predictive Control (SMPC), but does not rely on the use of a parametric system model. 
    As such, our receding-horizon algorithm produces a sequence of closed-loop control policies for predicted time steps, as opposed to a sequence of open-loop control actions.
    Under certain conditions, we establish that our proposed data-driven control method produces identical control inputs as that produced by the associated model-based SMPC. 
    Simulation results on a grid-connected power converter are provided to illustrate the performance benefits of our methodology.
\end{abstract}

\section{Introduction}

Model predictive control (MPC) is a widely used multi-variable control technique \cite{MPC:mayne2014}, capable of handling hard constraints on inputs, states, and outputs, along with complex performance criteria. Constraints can model actuator saturations or encode safety constraints in safety-critical applications. As the name suggests, MPC uses a system model, obtained either from first-principles modelling or from identification, to predict how inputs will influence the system evolution. MPC is therefore an \emph{indirect} design method, since one goes from data to a controller through an intermediate modelling step \cite{DDC:markovsky2021, DDC:markovsky2022}. In contrast, \emph{direct} methods, or data-driven methods, seek to compute controllers directly from input-output data. Data-driven methods show promise for systems that are complex or difficult to model \cite{DDC:dorfler2023, DDC:Hou2013}.

For stochastic systems, work on \emph{Stochastic MPC (SMPC)} \cite{SMPC:mesbah2016, SMPC:heirung2018, SMPC:farina2016} has focused on modelling the uncertainty in systems probabilistically.
SMPC methods optimize over feedback control policies rather than control actions, resulting in performance benefits when compared to the naive use of deterministic MPC \cite{SMPC_MPC:kumar2019}. Additionally, SMPC allows the use of probabilistic constraints, useful for computing risk-aware controllers. Another MPC method dealing with uncertainty is \emph{Robust MPC (RMPC)} \cite{RMPC:bemporad2007}, which attempts to conservatively guard against the worst-case deterministic uncertainty; our focus here is on the stochastic case.

The best achievable performance obtainable via data-driven control is clearly that of model-based control with a perfect model.
For deterministic linear time-invariant (LTI) systems, recent work has demonstrated that the data-driven control methods can indeed produce controls that are equivalent to their model-based counterparts \cite{SPC:huang2008, DeePC:coulson2019a}. 
However, for stochastic systems, equivalence between a data-based and model-based method have not been established, except in a few special cases which will be discussed shortly. 
Thus, the focus of this work is to develop a stochastic data-driven control framework with provable equivalence to its model-based SMPC counterpart.

\emph{Related Work:} Although data-driven control has been developed for decades, early work on data-driven methods did not adequately account for constraints on input and output; see examples in \cite{DDC:Hou2013}.  This observation led to the development of \emph{Data-Driven Predictive Control (DDPC)} as data-driven control methods incorporating input and output constraints.
Two of the best known DDPC methods are Data-enabled Predictive Control (DeePC) \cite{DeePC:coulson2019a, DeePC:coulson2019b, DeePC:coulson2021} and Subspace Predictive Control (SPC) \cite{SPC:huang2008}, both of which have been applied in multiple experiments with reliable results \cite{DeePCApp:elokda2021quadcopters, DeePCApp:carlet2020motorDrives, 
DeePCApp:mahdavipour2022combinedCycle, DeePCApp:huang2019gridConnected, DeePCApp:huang2021oscillationDamping, DeePCApp:zhao2021frequencyRegulation}.
On the theoretical side, for \emph{deterministic} LTI systems, both DeePC and SPC yield equivalent control actions to MPC, which is their model-based counterpart \cite{DeePC:coulson2019a, SPC:huang2008}.

Beyond the idealized case with deterministic linear systems, real-world systems are often stochastic and non-linear, and real-life data typically are perturbed by noise.
Hence, data-driven methods in practice need to adapt to data that is subject to these perturbations.
Most classical data-driven control methods are designed in robust ways \cite{DDC:Hou2013}, so their control performances are not sensitive to noisy data.
In application of SPC with noisy data, a predictor matrix is often computed with denoising methods, such as prediction error methods \cite{DeePCApp:huang2019gridConnected, DeePCApp:huang2021oscillationDamping} and truncated singular value decomposition \cite{DeePCApp:carlet2020motorDrives}.

Robust versions of DeePC have also been developed with stochastic systems in mind, such as norm-based regularized DeePC \cite{DeePC:coulson2019a, DeePC:coulson2019b} in which the regularization can be interpreted as a result of worst-case robust optimization \cite{DeePC:huang2021, DeePC:huang2023}, as well as distributionally robust DeePC \cite{DeePC:coulson2019b, DeePC:coulson2021}. Some other variations of DeePC were designed in purpose of ensuring closed-loop stability \cite{DDMPC:berberich2020a, DDMPC:berberich2020b, DDMPC:berberich2020c, DDMPC:berberich2021}, robustness to nonlinear systems \cite{DeePC:huang2023RoK} etc.
Although the stochastic adaptations of DeePC and SPC were validated through experiments, these stochastic data-driven methods do not possess an analogous theoretical equivalence to any Stochastic MPC or model-based method. 
Other related works include a tube-based \cite{kerz2023}, a sampling-based \cite{teutsch2024}, an innovation-based \cite{wang2025} and a constraint-tightening \cite{yin2024} stochastic DDPC scheme.
Again, however, no equivalence in performance was established between these methods and model-based MPC methods.

This disconnect between data-driven and model-based methods in the stochastic case has been noticed by some researchers, and some recent DDPC methods were developed for stochastic systems that have provable equivalence to model-based MPC methods. 
The works in \cite{PCE:pan2022a, PCE:pan2022b, PCE:pan2023a, PCE:pan2023b} proposed data-driven control frameworks for stochastic systems applying Polynomial Chaos Expansion (PCE); the use of PCE enables modeling of arbitrary random variables of finite mean and variance. Their methods have equivalent performance to SMPC when disturbances are known and when stochastic signals are exactly represented by finite PCE terms \cite[Thm. 1]{PCE:pan2022a} \cite[Cor. 1]{PCE:pan2022b}. 
In practice, disturbances should be estimated using input-output data, which requires heavier computation with larger amounts of data.
Their frameworks have considered systems without sensor noise and systems in the Auto-Regressive form with eXogenous input (ARX), which are special cases of systems in the state-space representation.
Thus, the gap addressed in this paper is to develop an alternative data-driven stochastic control method that has provably equivalent performance to the model-based SMPC, where we only estimate a fixed number of parameters regardless of data amount, and we consider general systems in the state-space form with separate process and sensor noise.

\emph{Contributions:} We develop a DDPC control method for stochastic LTI systems. Our technical approach is based on the construction of an auxiliary state model directly parameterized by input-output data. Mirroring SMPC, we formulate a stochastic control problem using this data-based auxiliary model, and establish equivalence between the proposed data-driven approach and its model-based SMPC counterpart. Our approach preserves three key features and benefits of SMPC. 
First, our formulation includes both process noise and measurement noise, so one can study the effect of different noise magnitudes on the control performance.
Second, we produce a feedback control policy at each time step, so that the control inputs are decided after real-time measurements in a closed-loop manner. 
Third, our control method incorporates safety chance constraints, which are consistent with the SMPC framework that we investigate.
Our data-driven method is established with symbolic analogy to SMPC, which enables us to adapt to data-driven counterparts of other SMPC settings, such as distributionally robust SMPC and correlated-noise SMPC.

\emph{Organization:} The rest of the paper is organized as follows.
Section \ref{SECTION:problem} shows the formal problem statement, with a brief overview of SMPC in Section \ref{SECTION:problem:SMPC}.
Our control method is introduced in Section \ref{SECTION:DDSMPC}, where we show the formulation and the theoretical performance guarantee, i.e., equivalence to SMPC.
Simulation results are displayed in Section \ref{SECTION:simulations}, comparing our proposed method and some benchmark control methods, and Section \ref{SECTION:conclusion} is the conclusion.

\emph{Notation:} 
Let $M^\dagger$ be the pseudo-inverse of a matrix $M$.
Let $\otimes$ denote the Kronecker product.
Let $\symmetric^q_+$ and $\symmetric^q_{++}$ be the sets of $q\times q$ positive semi-definite and positive definite matrices respectively.
Let $\col(M_1, \ldots, M_k)$ denote the column concatenation, and $\Diag(M_1, \ldots, M_k)$ the block-diagonal concatenation, of matrices/vectors $M_1, \ldots, M_k$.
Let $\integer_{[a,b]} := [a,b] \cap \integer$ denote a set of consecutive integers from $a$ to $b$. Let $\integer_{[a,b)} := \integer_{[a,b-1]}$.
For a $\real^q$-valued discrete-time signal $z_t$ with integer index $t$, let $z_{[t_1, t_2]}$ denote either a sequence $\{z_t\}_{t=t_1}^{t_2}$ or a concatenated vector $\col(z_{t_1}, \ldots, z_{t_2}) \in \real^{q(t_2-t_1+1)}$ where the usage is clear from the context. Similarly, let $z_{[t_1,t_2)} := z_{[t_1,t_2-1]}$.
A matrix sequence $\{M_t\}_{t=t_1}^{t_2}$ and a function sequence $\{\pi_t(\cdot)\}_{t=t_1}^{t_2}$ are denoted by $M_{[t_1,t_2]}$ and $\pi_{[t_1,t_2]}$ respectively.

\section{Problem Statement} \label{SECTION:problem}

We consider a stochastic linear time-invariant (LTI) system
\begin{subequations} \label{Eq:LTI} \begin{align}
    \label{Eq:LTI:state}
    x_{t+1} =&\; A x_t + B u_t + w_t, \\
    \label{Eq:LTI:output}
    y_t =&\; C x_t + D u_t + v_t,
\end{align} \end{subequations}
with input $u_t \in \real^m$, state $x_t \in \real^n$, output $y_t \in \real^p$, process noise $w_t \in \real^n$, and measurement noise $v_t \in \real^p$, all of which are random variables.
The initial state $x_0$ is uncertain with given mean $\mu^{\sf x}_\initial$ and with variance to be specified by a steady-state Kalman filter.
The system matrices $A,B,C,D$ are \emph{unknown} and the state $x_t$ is \emph{unmeasured}; we have access only to the input $u_t$ and output $y_t$ in \eqref{Eq:LTI}.
The disturbances $w_t$ and $v_t$ in \eqref{Eq:LTI} are independent of each other and of $x_0$, and are independently and identically distributed (i.i.d.) normally with zero mean and with variances $\Sigma^{\sf w} \in \symmetric^n_+$ and $\Sigma^{\sf v} \in \symmetric^p_{++}$ respectively, i.e.,
\begin{align} \label{Eq:noise_distribution}
    w_t \simiid \normal(0_{n\times 1}, \Sigma^{\sf w}), \quad
    v_t \simiid \normal(0_{p\times 1}, \Sigma^{\sf v}).
\end{align}
We assume the system $(A,B,C,D)$ is controllable and observable (i.e., a minimal realization), where observability is assumed without loss of generality for an unknown system \cite[Sec. 2.4]{DDC:markovsky2021}. 
Let $L \in \natural$ be such that the extended observability matrix $\mathcal{O} := \col(C, CA, \ldots, CA^{L-1})$ has full column rank; such smallest $L$ is the \emph{lag} of the system \cite{DDC:markovsky2021, DDC:markovsky2022}. Finally, we assume the pair $(A, \Sigma^{\sf w})$ is stabilizable (or equivalently, $(A, (\Sigma^{\sf w})^{1/2})$ is stabilizable), which will subsequently ensure uniqueness of the state variance by the Kalman filter \cite{kuvcera1972}.

In a reference tracking problem, the objective is for the output $y_t$ to follow a specified reference signal $r_t \in \real^p$. The trade-off between tracking error and control effort may be encoded in the cost
\begin{align} \label{Eq:stage_cost}
    J_t (u_t, y_t) := \Vert y_t-r_t \Vert_Q^2 + \Vert u_t \Vert_R^2 
\end{align}
to be minimized over a horizon, where $Q \in \symmetric^p_{++}$ and $R \in \symmetric^m_{++}$ are user-selected parameters. This tracking should be achieved subject to constraints on the inputs and outputs. We consider a polytopic constraint in the form $E \, \col(u_t, y_t) \leq f$, modeled in the stochastic setting as a probabilistic \emph{chance constraint}
\begin{align} \label{Eq:safety_constraint}
    \mathbb{P} \big\{ E \, \col(u_t, y_t) \leq f \big\} \geq 1 - p
\end{align}
for $t \in \natural_{\geq 0}$, where $E \in \real^{q \times (m+p)}$ is a fixed matrix, $f \in \real^q$ is a fixed vector, with some $q \in \natural$, and $p \in (0,1)$ is a probability bound of constraint violation.
One can similarly impose multiple chance constraints, e.g., separate input and output chance constraints, in the form of \eqref{Eq:safety_constraint}.

In a model-based setting where $A,B,C,D$ are known, the general control problem above can be addressed by SMPC, as will be reviewed in Section \ref{SECTION:problem:SMPC}.
Our broad objective is to construct a direct data-driven method that addresses the same stochastic control problem and is equivalent, under certain tuning conditions, to SMPC.

\begin{remark}[\bf Output Constraints and Output Tracking] \label{Remark:output_objectives}
    State constraints and costs are commonly considered in MPC and SMPC methods \cite{MPC:mayne2014, SMPC:mesbah2016, SMPC:heirung2018, SMPC:farina2016}, being used to enforce safety conditions and quantify control performance, respectively. Our problem setup focuses on output control, with the internal state being unknown and unrealized. For this reason, we instead considered input-output constraint \eqref{Eq:safety_constraint} for safety conditions and output-tracking cost \eqref{Eq:stage_cost} for performance evaluation, which are both common in DDPC methods such as \cite{DeePC:coulson2019a}. 
    \hfill \oprocend
\end{remark}
\subsection{Stochastic MPC: A Benchmark Model-Based Design} \label{SECTION:problem:SMPC}

Several formulations of SMPC methods have been developed in the literature \cite[Table 2]{SMPC:mesbah2016}.
Our focus is on output-feedback SMPC \cite{OFSMPC:cannon2012, OFSMPC:farina2015, OFSMPC:joa2023, OFSMPC:ridderhof2020, OFSMPC:hokayem2012}, which is typically approached by enforcing a separation principle within the design, augmenting full-state-feedback SMPC with state estimation. 
Our formulation here is based on an affine feedback-policy parameterization, following e.g., \cite{OFSMPC:farina2015, OFSMPC:joa2023}, with the modifications that we consider output tracking and output constraints, as opposed to state objectives.
The SMPC method under consideration here also integrates interpolation of initial condition \cite{Kohler2022, Schluter2022}, which is required for recursive feasibility with unbounded noise, and approximation of chance constraints \cite{Ono2008}, which leads to a tractable optimization problem.

\subsubsection{Initial Condition and State Estimation}
SMPC follows a receding-horizon strategy and makes decisions for $N$ upcoming steps at each \emph{control step}.
At control step $t=k$, the initial condition of the state $x_k$ is modelled as
\begin{align} \label{Eq:initial_condition}
    x_k \;\sim\; \normal (\mu^{\sf x}_k, \Sigma^{\sf x}),
\end{align}
where the mean $\mu^{\sf x}_k \in \real^n$ depends on a decision variable $\theta \in [0,1]$, according to an interpolation technique to be introduced in Section \ref{SECTION:problem:SMPC}-2.
The state variance $\Sigma^{\sf x} \in \symmetric^n_+$ in \eqref{Eq:initial_condition} is fixed and induced by the steady-state Kalman filter. Specifically, $\Sigma^{\sf x}$ is the \emph{unique} positive semidefinite solution to the associated discrete-time algebraic Riccatti equation (DARE) \cite{kuvcera1972}
\begin{subequations} \label{Eq:state_variance} \begin{align} \label{Eq:state_variance:DARE}
    & \Sigma^{\sf x} = (A - L_{\sf L} C) \Sigma^{\sf x} A^\transpose + \Sigma^{\sf w} \\
\label{Eq:state_variance:gain}
    & L_{\sf L} := A L_{\sf K}, \quad 
    L_{\sf K} := \Sigma^{\sf x} C^\transpose (C \Sigma^{\sf x}C^\transpose + \Sigma^{\sf v})^{-1}
\end{align} \end{subequations}
given detectable $(A, C)$ and stabilizable $(A, \Sigma^{\sf w})$, where we let $L_{\sf K} \in \real^{n \times p}$ denote the steady-state Kalman gain and $L_{\sf L} \in \real^{n \times p}$ the associated Luenberger observer gain.

With the initial condition \eqref{Eq:initial_condition}, we simulate the noise-free model for future $N$ time steps,\!
\begin{subequations} \label{Eq:nominal_model} \begin{align}
    \label{Eq:nominal_model:state}
    \thickbar x_{t+1} &:= A \thickbar x_t + B \thickbar u_t, && t \in \integer_{[k,k+N)} \\
    \label{Eq:nominal_model:output}
    \thickbar y_t &:= C \thickbar x_t + D \thickbar u_t, && t \in \integer_{[k,k+N)} \\
    \label{Eq:nominal_model:initial}
    \thickbar x_k &:= \mu^{\sf x}_k
\end{align} \end{subequations}
where the \emph{nominal inputs} $\thickbar u_t \in \real^m$ for $t \in \integer_{[k, k+N)}$ will be decision variables in optimization, with resulting \emph{nominal states} $\thickbar x_t \in \real^n$ and \emph{nominal outputs} $\thickbar y_t \in \real^p$. 

After the reveal of future measurements, estimates of the future states over the desired horizon will be computed through the steady-state Kalman filter, with $L_{\sf K}$ in \eqref{Eq:state_variance:gain},
\begin{subequations} \label{Eq:Kalman_filter} \begin{align}
    \label{Eq:Kalman_filter:postetior}
    \hat x^\plus_t &:= \hat x^\minus_t + L_{\sf K} (y_t - C \hat x^\minus_t - D u_t), && t \in \integer_{[k,k+N)}
    \\
    \label{Eq:Kalman_filter:prior}
    \hat x^\minus_{t+1} &:= A \hat x^\plus_t + B u_t, && t \in \integer_{[k,k+N)}
    \\
    \label{Eq:Kalman_filter:initial}
    \hat x^\minus_k &:= \mu^{\sf x}_k
\end{align} \end{subequations}
where $\hat x^\plus_t$ and $\hat x^\minus_t$ denote the posterior and prior estimates of $x_t$, respectively.
The steady-state Kalman filter \eqref{Eq:Kalman_filter} is equivalent to a Luenberger observer as in \cite{OFSMPC:farina2015, OFSMPC:cannon2012} with observer gain $L_{\sf L}$ in \eqref{Eq:state_variance:gain}, 
and is the stationary case of time-varying Kalman filters used in \cite{OFSMPC:joa2023, OFSMPC:ridderhof2020, OFSMPC:hokayem2012}.

\subsubsection{Interpolation of Initial Condition}
A common choice of $\mu^{\sf x}_k$ in \eqref{Eq:initial_condition} is the prior state estimate $\hat x^\minus_k$ produced by the estimator \eqref{Eq:Kalman_filter} in the previous control step \cite{OFSMPC:ridderhof2020, OFSMPC:joa2023, OFSMPC:hokayem2012}; we denote this choice by $\mu^{\sf \hat x}_k$.
However, in our setting the state estimates are normally distributed and thus unbounded. The choice $\mu^{\sf x}_k = \mu^{\sf \hat x}_k$ may lead to an extreme value of $\mu^{\sf x}_k$, which in turn could render the constraint \eqref{Eq:safety_constraint} infeasible. A different choice of $\mu^{\sf x}_k$ is the deterministic prediction $\thickbar x_k$ of state the $x_k$, obtained via \eqref{Eq:nominal_model} at last control step \cite{OFSMPC:farina2015}; we denote this choice by $\mu^{\sf \bar x}_k$.
Choosing $\mu^{\sf x}_k = \mu^{\sf \bar x}_k$ can guarantee feasibility, with proper design of the control optimization problem; however, the value $\mu^{\sf \bar x}_k$ does not incorporate feedback from past measurements.

Trading off the two options, we let the initial condition $\mu^{\sf x}_k$ in \eqref{Eq:initial_condition} \emph{interpolate} between $\mu^{\sf \hat x}_k$ and $\mu^{\sf \bar x}_k$ \cite{Kohler2022, Schluter2022} as
\begin{align} \label{Eq:interpolating_initial_condition}
    \mu^{\sf x}_k := (1-\theta)\, \mu^{\sf \hat x}_k + \theta\, \mu^{\sf \bar x}_k,
\end{align}
where $\theta \in [0,1]$ is a decision variable, and both $\mu^{\sf \hat x}_k, \mu^{\sf \bar x}_k \in \real^n$ are fixed and known at time $t=k$.
At initial control step $k=0$, $\mu^{\sf x}_0$ is equal to a given parameter $\mu^{\sf x}_\initial$, i.e., we let $\mu^{\sf \hat x}_0 := \mu^{\sf x}_\initial$ and $\mu^{\sf \bar x}_0 := \mu^{\sf x}_\initial$.

\subsubsection{Feedback Control Policies}
Stochastic state-feedback control requires the determination of (causal) feedback policies $\pi_t$ which map the observation history into control actions. 
As the space of policies is an infinite-dimensional function space, a simple affine feedback parameterization is typically used in SMPC to obtain a tractable finite-dimensional optimization problem, written as (cf. \cite{OFSMPC:cannon2012, OFSMPC:farina2015, OFSMPC:joa2023})
\begin{align} \label{Eq:feedback_policy}
    u_t = \pi_t(\hat x^\minus_t) := \thickbar u_t - K (\hat x^\minus_t - \thickbar x_t),
\end{align}
where $K \in \real^{m \times n}$ is a \emph{fixed} feedback gain such that $A-BK$ is Schur stable.
Through the policy \eqref{Eq:feedback_policy}, the control action $u_t$ depends both on the decision $\thickbar u_t$ optimized at the control step, 
and on the state estimate $\hat x^\minus_t$ via \eqref{Eq:Kalman_filter} which is decided after the measurement of $y_{[k,t)}$ and embodies feedback from the measurements.
Based on the cost \eqref{Eq:stage_cost}, we select the gain matrix $K$ in \eqref{Eq:feedback_policy} as the infinite-horizon LQR gain of system \eqref{Eq:LTI} with LQR stage cost $\Vert C x_t + D u_t \Vert_Q^2 + \Vert u_t \Vert_R^2$ 
\begin{align} \label{Eq:feedback_gain}
    K := (R + B^\transpose P B + D^\transpose Q D)^{-1} (B^\transpose P A + D^\transpose Q C)
\end{align}
where $P \in \symmetric^n_+$ is the \emph{unique} positive semidefinite solution to the discrete-time algebraic Riccati equation (DARE) \cite{kuvcera1972}
\begin{align} \label{Eq:feedback_DARE} 
    P = A^\transpose P (A - B K) + C^\transpose Q (C - D K),
\end{align}
given stabilizable $(A, B)$, detectable $(A, C)$ and $Q \succ 0$. 
An equivalent form $\pi_t(\hat x^\minus_t) := c_t - K \hat x^\minus_t$ of \eqref{Eq:feedback_policy} with decision variable $c_t$ has been used in \cite{OFSMPC:cannon2012} and in many SMPC examples surveyed in \cite{SMPC:mesbah2016}. A time-varying-gain version of \eqref{Eq:feedback_policy} is adopted in \cite{OFSMPC:farina2015}, and \cite{OFSMPC:joa2023} uses $\hat x^\plus_t$ in place of $\hat x^\minus_t$ in the control policy.
\emph{Affine disturbance feedback} is sometimes considered in SMPC methods, e.g. \cite{OFSMPC:ridderhof2020}, and it is shown that affine disturbance feedback control policies and affine state feedback control policies lead to equivalent control inputs \cite{goulart2006}; here we focus on the state feedback parameterization.

\begin{remark}[\bf Input Chance Constraints] \label{REMARK:conflict}
Hard input constraints are difficult to integrate with the affine policy \eqref{Eq:feedback_policy}, as under our previous assumptions the resulting control input is normally distributed and unbounded. Chance constraint \eqref{Eq:safety_constraint} on input is thus used in its place, as in \cite{OFSMPC:farina2015}.
Another option as in \cite{OFSMPC:hokayem2012} is to use (nonlinear) saturated policies in place of  \eqref{Eq:feedback_policy}, but then the resulting inputs and outputs are no longer linear in decision variables and further analysis would be much more complicated. 
Ultimately in implementation of course, one can saturate input actions to satisfy hard input constraints.
\hfill \oprocend
\end{remark}

\subsubsection{Propagation of Input-Output Distribution}

With \eqref{Eq:LTI}, \eqref{Eq:noise_distribution}, \eqref{Eq:initial_condition}, \eqref{Eq:nominal_model}, \eqref{Eq:Kalman_filter} and \eqref{Eq:feedback_policy}, at control step $t=k$, the resulting future inputs $u_t$ and outputs $y_t$ for $t \in \integer_{[k,k+N)}$ are distributed as
\begin{align} \label{Eq:input_output_distribution} 
    \neatmat{u_t \\[-.25em] y_t} \sim \normal \big( \neatmat{\thickbar u_t \\[-.25em] \thickbar y_t},\, \Delta_{t-k} \big),
\end{align}
where the covariance matrix $\Delta_s \in \symmetric^{m+p}_+$ for $s \in \integer_{[0,N)}$ can be computed as \eqref{Eq:input_output_variance:Delta} using $\Lambda_s \in \symmetric^n_+$ defined by \eqref{Eq:input_output_variance:Lambda},
\begin{subequations} \label{Eq:input_output_variance} \begin{align}
\label{Eq:input_output_variance:Delta}
    & \Delta_s := \compactmat{-K \\ C \!-\! D K} \Lambda_s  \compactmat{-K \\ C \!-\! D K}^\transpose + \compactmat{0_{m \times m}\!\!\! \\& C \Sigma^{\sf x} C^\transpose \!+\! \Sigma^{\sf v}} \\
\label{Eq:input_output_variance:Lambda}
    & \Lambda_s := \textstyle{\sum_{r=0}^{s-1}} (A\!-\!BK)^r L_{\sf L} \, (C \Sigma^{\sf x} C^\transpose \!+\! \Sigma^{\sf v}) \, L_{\sf L}\!{}^\transpose (A\!-\!BK)^{r \transpose}
\end{align} \end{subequations}
with $L_{\sf L}$ in \eqref{Eq:state_variance:gain} and $K$ in \eqref{Eq:feedback_gain}.
The distribution \eqref{Eq:input_output_distribution} is derived analogously to \cite{OFSMPC:farina2015} and a complete derivation can be found in \Versions{the extended version \cite[Appendix A]{EXTENDED}}{\ref{APPENDIX:PROOF:input_output_distribution}}.
Note that the matrices $\Delta_0, \Delta_1, \ldots, \Delta_{N-1}$ are fixed and can be computed offline.

SMPC problems typically consider the expectation of cost \eqref{Eq:stage_cost} to be summed over the desired horizon. Given distribution \eqref{Eq:input_output_distribution}, the expected cost is known as a deterministic value
\begin{align} \label{Eq:cost_reduced}
    \textstyle{\sum_{t=k}^{k+N-1}} \mathbb{E} [J_t (u_t, y_t)] = \textstyle{\sum_{t=k}^{k+N-1}} \big[ J_t (\thickbar u_t, \thickbar y_t) + J^{\sf var}_{t-k} \big],
\end{align}
where $J^{\sf var}_s := \mathrm{Trace} ( \Delta_s \, \Diag(R, Q))$ is a constant independent of the decision variables $\thickbar u$ and $\theta$.

\subsubsection{Chance-Constraint Approximation}
Despite known input-output distribution \eqref{Eq:input_output_distribution}, an exact deterministic representation of the joint chance constraint \eqref{Eq:safety_constraint} is difficult, as it requires integration of a multivariate probability density function over a polytope and generally no analytic representation is available \cite[Sec. 2.2]{SMPC:farina2016}.
For this reason, the joint constraint \eqref{Eq:safety_constraint} is commonly approximated by, e.g., being split into individual chance constraints \cite{Ono2008}, for each time $t \in \integer_{[k,k+N)}$,
\begin{align} \label{Eq:safety_constraint_relaxed}
    \mathbb{P} \big\{ e_i^\transpose \neatmat{u_t \\[-.25em] y_t} \leq f_i \big\} \geq 1-p_{i,t}, \quad i \in \integer_{[1,q]}
\end{align}
where $e_i \in \real^{m+p}$ is the transposed $i$-th row of $E$, and $f_i \in \real$ is the $i$-th entry of $f$.
The allocated risk probabilities $p_{i,t} > 0$ in \eqref{Eq:safety_constraint_relaxed} are introduced as additional decision variables, such that $p_{1,t}, p_{2,t}, \ldots, p_{q,t}$ sum up to the total risk $p$ for each time $t$.
Note that \eqref{Eq:safety_constraint_relaxed} is a conservative approximation (or a sufficient condition) of \eqref{Eq:safety_constraint}, due to subadditivity of probabilities.
Given distribution \eqref{Eq:input_output_distribution}, the chance constraints \eqref{Eq:safety_constraint_relaxed} are converted into an equivalent deterministic form, cf. \cite{OFSMPC:farina2015, Ono2008},
\begin{subequations} \label{Eq:safety_constraint_reduced} \begin{align} &
    \label{Eq:safety_constraint_reduced:constraint}
    e_i^\transpose \neatmat{\thickbar u_t \\[-.25em] \thickbar y_t} \leq f_i \!-\! \sqrt{ e_i\!{}^\transpose \Delta_{t-k} \, e_i } \, \mathrm{icdfn} (1 \!-\! p_{i,t}), \; i \in \integer_{[1,q]} \\&
    \label{Eq:safety_constraint_reduced:risk_allocation}
    \textstyle{\sum_{i=1}^q} \, p_{i,t} = p, \qquad\quad
    p_{i,t} > 0, \;\; i \in \integer_{[1,q]}
\end{align} \end{subequations}
for $t \in \integer_{[k,k+N)}$, where $\mathrm{icdfn}(z) := \sqrt2 \, \mathrm{erf}^{-1}(2z-1)$ is the inverse cumulative distribution function (inverse c.d.f.) or the $z$-quantile of the standard normal distribution, with $\mathrm{erf}^{-1}$ the inverse error function.
The constraints \eqref{Eq:safety_constraint_reduced} are convex when we require $p \in (0, \frac12]$ \cite[Thm. 1]{Ono2008}.

\begin{remark}[\bf Gaussian Signals]
We have assumed through \eqref{Eq:noise_distribution} and \eqref{Eq:initial_condition} that random variables are normally distributed. 
In the case where random signals are non-Gaussian but with the same means and variances in \eqref{Eq:noise_distribution} and \eqref{Eq:initial_condition}, the resulting inputs $u_t$ and outputs $y_t$ still possess the mean and variance in \eqref{Eq:input_output_distribution}, and thus the expected cost is still \eqref{Eq:cost_reduced}.
However, the inverse c.d.f. in \eqref{Eq:safety_constraint_reduced:constraint} should change correspondingly to the actual distribution (if known), or be replaced into an upper bound $\sqrt{(1-p_{i,t}) / p_{i,t}}$ via Chebyshev–Cantelli inequality that guarantees the worst case over all distributions \cite{teutsch2024, yin2024}.
\hfill \oprocend
\end{remark}

\subsubsection{Terminal Condition}
Terminal constraints are considered in (S)MPC frameworks to ensure recursive feasibility and closed-loop stability. Assume $N \geq L$ going forward. Here, we impose a \emph{terminal equality constraint} \cite{DDMPC:berberich2020a, DDMPC:berberich2020b, DDMPC:berberich2020c, DDMPC:berberich2021},
\begin{align} \label{Eq:terminal_condition} \begin{aligned}
    & \thickbar u_{k+N-L} = \thickbar u_{k+N-L+1} = \cdots = \thickbar u_{k+N-1} \\
    & \thickbar y_{k+N-L} = \thickbar y_{k+N-L+1} = \cdots = \thickbar y_{k+N-1} 
\end{aligned} \end{align} 
that requires the nominal input-output trajectory to stay at some setpoint for final $L$ steps in the prediction horizon.
\emph{Terminal set constraints} are also leveraged in (S)MPC methods, bounding the final nominal state in a positively invariant set \cite{OFSMPC:cannon2012, OFSMPC:farina2015, OFSMPC:joa2023, Kohler2022}; here we find the input-output terminal constraint \eqref{Eq:terminal_condition} more straightforward to adapt to the data-driven case.

\subsubsection{SMPC Optimization Problem and Implementation}
With the expected cost \eqref{Eq:cost_reduced}, the approximation \eqref{Eq:safety_constraint_reduced} of the constraint \eqref{Eq:safety_constraint}, the interpolation \eqref{Eq:interpolating_initial_condition} and the terminal constraint \eqref{Eq:terminal_condition}, the SMPC problem is formulated as
\begin{align} \label{Eq:SMPC_reduced} \begin{aligned}
    \minimize_{\thickbar u,\, \theta,\, p_{i,t}} \quad& \textstyle{\sum_{t=k}^{k+N-1}} J_t (\thickbar u_t, \thickbar y_t) + \lambda_\theta \,\theta
    \\
    \mathrm{subject\;to} \quad&
    \text{\eqref{Eq:safety_constraint_reduced} for $t \in \integer_{[k,k+N)}$}, \eqref{Eq:nominal_model}, \eqref{Eq:interpolating_initial_condition}, \eqref{Eq:terminal_condition},
\end{aligned} \end{align}
with an interpolation penalty term of parameter $\lambda_\theta > 0$ \cite{Schluter2022}.
With $R \succ 0$ and $\lambda_\theta > 0$, the cost in \eqref{Eq:SMPC_reduced} is jointly strongly convex in $\thickbar u$ and $\theta$, and thus problem \eqref{Eq:SMPC_reduced} possesses a unique optimal $(\thickbar u, \theta)$ if feasible, although optimal $p_{i,t}$ may not be unique. 
Problem \eqref{Eq:SMPC_reduced} can be efficiently solved by the Iterative Risk Allocation method \cite{Ono2008}; see \Versions{\cite[Appendix B]{EXTENDED}}{\ref{APPENDIX:iterative_risk_allocation}} for more details of our implementation.

The nominal inputs $\thickbar u_{[k,k+N)}$ and interpolation variable $\theta$ determined from problem \eqref{Eq:SMPC_reduced} complete the parameterization of the control policies $\pi_{[k, k+N)}$ in \eqref{Eq:feedback_policy}.
The upcoming $N_{\rm c}$ control inputs $u_{[k, k+N_{\rm c})}$ are decided by the first $N_{\rm c}$ policies $\pi_{[k, k+N_{\rm c})}$ respectively, with a parameter $N_{\rm c} \in \integer_{[1,N]}$.
Then, the next control step is set as $t = k+N_{\rm c}$.
At the new control step, the initial condition $\mu^{\sf x}_{k+N_{\rm c}}$ interpolates between two fixed options $\mu^{\sf \hat x}_{k+N_{\rm c}}$ and $\mu^{\sf \bar x}_{k+N_{\rm c}}$ which are decided by
\begin{align} \label{Eq:initial_condition_iteration}
    \mu^{\sf \hat x}_{k+N_{\rm c}} := \hat x^\minus_{k+N_{\rm c}}, \;\;\quad
    \mu^{\sf \bar x}_{k+N_{\rm c}} := \thickbar x_{k+N_{\rm c}},
\end{align}
as described in Section \ref{SECTION:problem:SMPC}-2.
The entire SMPC control process is shown in Algorithm \ref{ALGO:SMPC}.

\begin{algorithm}
\caption{a Framework of Stochastic MPC (SMPC)} \label{ALGO:SMPC}
\begin{algorithmic}[1]
    \Require horizon lengths $L, N, N_{\rm c}$, system matrices $A,B,C,D$, noise variances $\Sigma^{\sf w}, \Sigma^{\sf v}$, initial state mean $\mu^{\sf x}_\initial$, cost matrices $Q, R$, constraint coefficients $E, f$, probability bound $p$, interpolation penalty coefficient $\lambda_\theta$.
    \State Compute Kalman gain $L_{\sf K}$ via \eqref{Eq:state_variance:gain}, feedback gain $K$ via \eqref{Eq:feedback_gain}, and covariance matrices $\Delta_{[0,N)}$ via \eqref{Eq:input_output_variance}.
    \State Initialize the control step $k \gets 0$ and set the initial condition $\mu^{\sf \hat x}_0 \gets \mu^{\sf x}_\initial$ and $\mu^{\sf \bar x}_0 \gets \mu^{\sf x}_\initial$.
    \While{\texttt{true}}
    \State Solve $\thickbar u_{[k,k+N)}$ and $\theta$ from problem \eqref{Eq:SMPC_reduced}.
    \State Obtain $\mu^{\sf x}_k$ via \eqref{Eq:interpolating_initial_condition} and obtain $\thickbar x_{[k,k+N]}$ via \eqref{Eq:nominal_model}.
    \State Obtain policies $\pi_{[k,k+N)}$ from \eqref{Eq:feedback_policy}.
    \For{{\bf $t$ from $k$ to $k+N_{\rm c}-1$}}
        \State Compute $\hat x^\minus_t$ via \eqref{Eq:Kalman_filter}.
        \State Input $u_t \gets \pi_t(\hat x^\minus_t)$ to the system \eqref{Eq:LTI}.
        \State Measure $y_t$ from the system \eqref{Eq:LTI}.
    \EndFor
    \State Set $\mu^{\sf \hat x}_{k+N_{\rm c}} \gets \hat x^\minus_{k+N_{\rm c}}$ and $\mu^{\sf \bar x}_{k+N_{\rm c}}  \gets \thickbar x_{k+N_{\rm c}}$ as \eqref{Eq:initial_condition_iteration}.
    \State Set $k \gets k + N_{\rm c}$.
    \EndWhile
\end{algorithmic}
\end{algorithm}

\subsubsection{Closed-loop Properties}
The investigated SMPC framework possesses recursive feasibility and closed-loop stability. 

\begin{lemma}[SMPC Recursive Feasibility] \label{LEMMA:recursive_feasibility}
    Assume $p \in (0,\frac12]$.
    In Algorithm \ref{ALGO:SMPC}, if the problem \eqref{Eq:SMPC_reduced} is feasible at control step $k=\kappa$, then it is feasible at next control step $k = \kappa + N_{\rm c}$. 
\end{lemma}

With Lemma \ref{LEMMA:recursive_feasibility}, problem \eqref{Eq:SMPC_reduced} is feasible at all control steps if it is feasible at the initial control step, where initial feasibility can be achieved by a proper choice of parameters $\mu^{\sf x}_\initial, E, f, p$.
Closed-loop stability is reflected in the decrease of the optimal cost value and the finiteness of the asymptotic cost \cite{Cannon2009, Hewing2020, PCE:pan2022b, PCE:pan2023a}. Let $\mathbb{E}_k[\cdot]$ denote the expectation given the initial condition \eqref{Eq:initial_condition} at control step $k$. Define $V^*_k := \Sigma_{t=k}^{k+N-1} J_t (u_t, y_t)$ the stochastic cost with the optimal nominal trajectory solved from problem \eqref{Eq:SMPC_reduced}.

\begin{lemma}[SMPC Closed-loop Stability] \label{LEMMA:closed_loop_stability}
    Assume $\{z : E z \leq f\}$ is a bounded set.
    Let system \eqref{Eq:LTI} be controlled by Algorithm \ref{ALGO:SMPC}, where problem \eqref{Eq:SMPC_reduced} is assumed feasible at all control steps, and the reference signal $r_t = r$ is fixed. Then, the expectation of optimal cost values at consecutive control steps differ as
\begin{align} \label{Eq:LEMMA:closed_loop_stability:value_difference}
    \mathbb{E}_\kappa [V^*_{\kappa + N_{\rm c}} - V^*_\kappa] \leq - \textstyle{\sum_{t=\kappa}^{\kappa + N_{\rm c} -1}} \mathbb{E}_\kappa [J_t (u_t, y_t)] + N_{\rm c} c,
\end{align}
    and the asymptotic expected cost is upper bounded as
\begin{align} \label{Eq:LEMMA:closed_loop_stability:average_cost}
    \displaystyle \lim_{T \to \infty} \textstyle{\frac1T \sum_{t=0}^{T-1}} \mathbb{E}_0 [J_t (u_t, y_t)] \leq c
\end{align}
    for some $c > 0$.
\end{lemma}

{\tb The proofs of the above lemmas are analogous to the proofs in \cite{OFSMPC:farina2015} and \Versions{available in the extended version \cite[Appendices C and D]{EXTENDED}}{can be found in \ref{APPENDIX:PROOF:recursive_feasibility} and \ref{APPENDIX:PROOF:closed_loop_stability}, respectively}.}

\subsection{Our Objective: An Equivalent Data-Driven Method} \label{SECTION:problem:objective}

In direct data-driven control methods such as DeePC and SPC for deterministic systems, a sufficiently long and sufficiently rich set of noise-free input-output data is collected. Under technical conditions, this data provides an equivalent representation of the underlying system dynamics, and is used to replace the parametric model in predictive control schemes, yielding control algorithms which are \emph{equivalent} to model-based predictive control \cite{DeePC:coulson2019a, SPC:huang2008}. Motivated by this equivalence, our goal here is to develop a direct data-driven control method that produces the same input-state-output sequences as produced by Algorithm \ref{ALGO:SMPC} when applied to the same system \eqref{Eq:LTI} with the same initial condition $x_0$ and the same realizations of noise $w_t, v_t$. Put simply, we seek a direct data-driven counterpart to SMPC.

As in the described cases of equivalence for DeePC and SPC, we will subsequently show equivalence of our data-driven method to SMPC \emph{in the idealized case where we assume access to noise-free offline data}. This assumption solely facilitates the proof of equivalence, and is not a fundamental requirement of the method itself. While the result under this assumption may initially seem peculiar in an explicitly stochastic control setting, we view it as the most reasonable theoretical result to aim for, given that the prediction model must be replaced using only a finite amount of recorded data. Noisy offline data can be accommodated in a robust fashion through the use of regularized least-squares (Section \ref{SECTION:DDSMPC:offline}), as supported by simulation results in  Section \ref{SECTION:simulations}, and our stochastic control approach will fully take into account process and sensor noise during the online execution of the control process.

\section{Stochastic Data-Driven Predictive Control} \label{SECTION:DDSMPC}

This section develops a data-driven control method whose performance will be shown to be equivalent to SMPC under certain tuning conditions. In the spirit of DeePC and SPC, our proposed control method consists of an offline process, where data is collected and used for system representation, and an online process which controls the system. 

At a high level, our technical approach has three key steps. First, we collect offline input-output data (Section \ref{SECTION:DDSMPC:offline}), and use this offline data to parameterize an auxiliary model (Section \ref{SECTION:DDSMPC:auxiliary}-1). This auxiliary model will take the place of the original parametric system model \eqref{Eq:LTI} in the design procedure. Second, we will formulate a  stochastic predictive control method using the auxiliary model (Section \ref{SECTION:DDSMPC:auxiliary}, Section \ref{SECTION:DDSMPC:optimization}-1, Section \ref{SECTION:DDSMPC:control}-1). Third and finally, we will establish theoretical equivalences between the model-based and data-based control methods (Section \ref{SECTION:DDSMPC:optimization}-2, Section \ref{SECTION:DDSMPC:control}-2).

\subsection{Use of Offline Data} \label{SECTION:DDSMPC:offline}

In data-driven control, sufficiently rich offline data must be collected to capture the internal dynamics of the system. 
In this subsection, we demonstrate how offline data is collected, and use the data to compute some quantities that are useful to formulate our control method in the rest of the section. We first develop results with data from deterministic LTI systems, and then address the case of noisy data.

\subsubsection{Deterministic Offline Data}
Consider the deterministic version of system \eqref{Eq:LTI}, reproduced for convenience as
\begin{align} \label{Eq:LTI_deterministic}
    x_{t+1} = A x_t + B u_t, \qquad y_t = C x_t + D u_t,
\end{align}
where with a slight abuse of notation,  we temporarily in this section let $x_t$ and $y_t$ denote the state and output of system \eqref{Eq:LTI_deterministic}. 
By assumption, \eqref{Eq:LTI_deterministic} is minimal; recall $L \in \natural$ in Section \ref{SECTION:problem} such that $\mathcal{O} := \col(C, CA, \ldots, CA^{L-1})$ has full column rank.
Let $u^\data_{[1,T_\data]}, y^\data_{[1,T_\data]}$ be a $T_\data$-length trajectory of input-output data collected from \eqref{Eq:LTI_deterministic}. The input sequence $u^\data_{[1,T_\data]}$ is assumed to be \emph{persistently exciting} of order $K_\data := L + n + 1$, i.e., its associated $K_\data$-depth block-Hankel matrix $\Hank_{K_\data} ( u^\data_{[1,T_\data]} )$, defined as
\begin{align*}
    \Hank_{K_\data}(u^\data_{[1,T_\data]}) := \compactmat{
        u^\data_1 & u^\data_2 & \cdots & u^\data_{T_\data-K_\data+1} \\
        u^\data_2 & u^\data_3 & \cdots & u^\data_{T_\data-K_\data+2} \\[-.5em]
        \vdots & \vdots & \ddots & \vdots \\
        u^\data_{K_\data} & u^\data_{K_\data+1} & \cdots & u^\data_{T_\data} },
\end{align*}
has full row rank. 
To achieve persistent excitation, one must collect at least $T_\data \geq (m+1)K_\data-1$ data samples \cite{DeePC:coulson2019a}.
We formulate data matrices $U_1 \in \real^{mL \times h}$, $U_2 \in \real^{m \times h}$, $Y_1 \in \real^{pL \times h}$ and $Y_2 \in \real^{p \times h}$ of a common width $h := T_\data - L$ by partitioning associated Hankel matrices as
\begin{align} \label{Eq:data_matrices} \begin{aligned}
    \col(U_1, U_2) &:= \Hank_{L+1} \big( u^\data_{[1,T_\data]} \big),
    \\ 
    \col(Y_1, Y_2) &:= \Hank_{L+1} \big( y^\data_{[1,T_\data]} \big).
\end{aligned} \end{align}
The data matrices in \eqref{Eq:data_matrices} will now be used to represent some quantities related to the system \eqref{Eq:LTI_deterministic}. Before stating the result, we introduce some additional notation. 
Define a system-related matrix $\mathbf{\Gamma} \in \real^{p \times (m+p)L}$ as
\begin{align} \label{Eq:Gamma_definition}
    \mathbf{\Gamma} = \mat{\mathbf{\Gamma}_{\sf U} & \mathbf{\Gamma}_{\sf Y}} 
    := \mat{C \mathcal{C} & C A^L} \mat{I_{mL} \\ \mathcal{G} & \mathcal{O}}^\dagger.
\end{align}
with sub-blocks $\mathbf{\Gamma}_{\sf U} \in \real^{p \times mL}$ and $\mathbf{\Gamma}_{\sf Y} \in \real^{p \times pL}$,
where $\mathcal{C} := [A^{L-1}B, \ldots, AB, B]$ is the extended (reversed) controllability matrix, and $\mathcal{G} \in \real^{pL \times mL}$ is an impulse-response matrix
\begin{align} \label{Eq:GH_definition}
    \mathcal{G} := \neatmat{
        D \\
        CB & D \\[-.5em]
        \vdots & \ddots & \ddots \\
        CA^{L-2}B & \cdots & CB & D }.
\end{align}
The following result provides expressions for the quantity $\mathbf{\Gamma}$ and the system matrix $D$ in terms of raw data. 

\begin{lemma}[\bf Data Representation of Model Quantities] \label{LEMMA:system_quantity_data_representation}
    Given the data matrices in \eqref{Eq:data_matrices}, if system \eqref{Eq:LTI_deterministic} is controllable and the input data $u^\data_{[1,T_\data]}$ is persistently exciting of order $L+n+1$, then the matrix $\mathbf{\Gamma}$ defined in \eqref{Eq:Gamma_definition} and the matrix $D$ in the model \eqref{Eq:LTI_deterministic} can be expressed as 
\begin{align*}
    [\mathbf{\Gamma}_{\sf U}, \mathbf{\Gamma}_{\sf Y}, D] = Y_2\, \col(U_1, Y_1, U_2)^\dagger.
\end{align*}
\end{lemma}
\begin{proof}
    See \ref{APPENDIX:PROOF:LEMMA:system_quantity_data_representation}.
\end{proof}

The data-expression of impulse response, e.g., $D$ and $\mathcal{G}$, is present in SPC literature \cite{SPC:huang2008}. The novelty of Lemma \ref{LEMMA:system_quantity_data_representation} is the data-based representation of $\mathbf{\Gamma}$, which will be used as part of the construction for our data-driven control method.

\subsubsection{The Case of Stochastic Offline Data}
Lemma \ref{LEMMA:system_quantity_data_representation} holds for the case of noise-free data. When the measured data is corrupted by noise, as will usually be the case, the pseudoinverse computations in Lemma \ref{LEMMA:system_quantity_data_representation} become fragile and do not recover the desired matrices $\mathbf{\Gamma}$ and $D$. A standard technique to robustify these computations is to replace the pseudoinverse $W^\dagger$ of $W := \col(U_1, Y_1, U_2)$ in Lemma \ref{LEMMA:system_quantity_data_representation} with its Tikhonov regularization $W^{\sf tik} := (W^\transpose W + \lambda I_h)^{-1} W^\transpose$ where $\lambda > 0$ is the regularization parameter. 
To interpret this, recall that $\mathcal{P} := Y_2 W^\dagger$ is a least-square solution to $\argmin_\mathcal{P} \Vert Y_2 - \mathcal{P} W \Vert_{\sf F}^2$. Correspondingly, the regularization $Y_2 W^{\sf tik}$ is the solution to a ridge-regression problem $\argmin_\mathcal{P} \Vert Y_2 - \mathcal{P} W \Vert_{\sf F}^2 + \lambda \Vert \mathcal{P} \Vert_{\sf F}^2$, which gives a maximum-likelihood or worst-case robust solution to the original least-square problem $\argmin_\mathcal{P} \Vert Y_2 - \mathcal{P} W \Vert_{\sf F}^2$ whose multiplicative parameter $W$ has uncertain entries; see \cite{DDC:markovsky2022} sidebar ``Roles of Regularization'' and references therein for more details. Following this standard technique, in the stochastic case, we estimate matrices $\mathbf{\Gamma}$ and $D$ by applying Lemma \ref{LEMMA:system_quantity_data_representation} with $\mathcal{P} = Y_2 W^\dagger$ replaced by $\widehat{\mathcal{P}} := Y_2 W^{\sf tik}$.

\subsection{Data-Driven State Estimation and Output Feedback} \label{SECTION:DDSMPC:auxiliary}

The SMPC approach of Section \ref{SECTION:problem:SMPC} uses as sub-components a state estimator and an affine feedback law. We now leverage the offline data as described in Section \ref{SECTION:DDSMPC:offline} to directly design analogs of these components based on data, and without knowledge of the system matrices. 

\subsubsection{Auxiliary State-Space Model}

We begin by constructing an auxiliary state-space model which has equivalent input-output behavior to \eqref{Eq:LTI}, but is parameterized only by the recorded data sequences of Section \ref{SECTION:DDSMPC:offline}. Define auxiliary signals $\mathbf{x}_t, \mathbf{w}_t \in \real^{n_\auxiliary}$ of dimension $n_\auxiliary := mL+pL+pL^2$ for system \eqref{Eq:LTI} by
\begin{align} \label{Eq:DDModel_state_definition}
    \mathbf{x}_t := \left[ \begin{array}{c} u_{[t-L,t)} \\ \hline y^\circ_{[t-L,t)} \\ \hline \rho_{[t-L,t)} \end{array} \right], \quad
    \mathbf{w}_t := \left[ \compact{\begin{array}{c} 0_{mL\times 1} \\ \hline 0_{pL\times 1} \\ \hline 0_{pL(L-1)\times 1} \\ \rho_t \end{array}} \right]
\end{align}
where $y^\circ_t := y_t - v_t \in \real^p$ is the output excluding measurement noise, and $\rho_t := \mathcal{O} w_t \in \real^{pL}$ stacks the system's response to process noise $w_t$ on time interval $[t+1, t+L]$.
The construction of the auxiliary state $\mathbf{x}_t$ was inspired by \cite{EKFDeePC:alpago2020}.
The auxiliary signals $\mathbf{x}_t, \mathbf{w}_t$ together with $u_t, y_t, v_t$ then satisfy the relations given by Lemma \ref{LEMMA:AuxModel}. 

\begin{lemma}[Auxiliary Model] \label{LEMMA:AuxModel}
    For system \eqref{Eq:LTI}, the signals $u_t, y_t, v_t$ and the auxiliary signals $\mathbf{x}_t, \mathbf{w}_t$ in \eqref{Eq:DDModel_state_definition} satisfy
\begin{subequations} \label{Eq:DDModel} \begin{align}
    \label{Eq:DDModel:state}
    \mathbf{x}_{t+1} &= \mathbf{A} \mathbf{x}_t + \mathbf{B} u_t + \mathbf{w}_t, \\
    \label{Eq:DDModel:output}
    y_t &= \mathbf{C} \mathbf{x}_t + D u_t + v_t,
\end{align} \end{subequations}
    with $\mathbf{A} \in \real^{n_\auxiliary \times n_\auxiliary}$, $\mathbf{B} \in \real^{n_\auxiliary \times m}$, $\mathbf{C} \in \real^{p \times n_\auxiliary}$ given by
\begin{align*}
    \mathbf{A} &:= \left[ \compact{\begin{array}{c|c|c}
        \begin{matrix} & I_{m(L-1)} \\ 0_{m \times m} \end{matrix} & 0 & 0 \\ \hline
        \begin{matrix} 0 \\ \mathbf{\Gamma}_{\sf U} \end{matrix} & \begin{matrix} \begin{matrix} 0 & I_{p(L-1)} \end{matrix} \\ \hline \mathbf{\Gamma}_{\sf Y} \end{matrix} & \begin{matrix} 0 \\ \mathbf{F} - \mathbf{\Gamma}_{\sf Y} \mathbf{E} \end{matrix} \\ \hline
        0 & 0 & \begin{matrix} & I_{pL(L-1)} \\ 0_{pL \times pL} \end{matrix}
    \end{array}} \right] \\
    \mathbf{B} &:= \left[ \compact{\begin{array}{c}
        0_{m(L-1) \times m} \\ I_m \\ \hline 0_{p(L-1)\times m} \\ D \\ \hline  0_{pL^2 \times m}
    \end{array}} \right], \;\;
    \mathbf{C} := \left[ \begin{array}{c|c|c} \mathbf{\Gamma}_{\sf U} & \mathbf{\Gamma}_{\sf Y} & \mathbf{F} - \mathbf{\Gamma}_{\sf Y} \mathbf{E} \end{array} \right]
\end{align*}
    with matrices $\mathbf{\Gamma}_{\sf U}, \mathbf{\Gamma}_{\sf Y}$ in \eqref{Eq:Gamma_definition}, the same matrix $D$ in \eqref{Eq:LTI}, 
    and zero-one matrices $\mathbf{E} \in \real^{pL\times pL^2}$ and $\mathbf{F} \in \real^{p\times pL^2}$ composed by selection matrices $S_j := [ 0_{p\times (j-1)p} , I_p , 0_{p\times (L-j)p} ]$ $\in \real^{p\times pL}$ for $j \in \{1, \ldots, L\}$ as
    \begin{align*} \left[ \begin{array}{c} \mathbf{E} \\ \hline \mathbf{F} \end{array} \right] := \left[ \compact{\begin{array}{cccc} 0_{p\times pL} \\ S_1 & 0_{p\times pL} \\[-.5em] \vdots & \ddots & \ddots \\ S_{L-1} & \cdots & S_1 & 0_{p\times pL} \\ \hline S_L & \cdots & S_2 & S_L \end{array}} \right]. \end{align*}
\end{lemma}
\begin{proof}
    See \ref{APPENDIX:PROOF:DDModel}.
\end{proof}

The output noise signal $v_t$ in \eqref{Eq:DDModel} is precisely the same as in \eqref{Eq:LTI}; the signal $\mathbf{w}_t$ appears now as a new disturbance; $\mathbf{w}_t$ and $v_t$ are independent and follow the i.i.d. zero-mean normal distributions
\begin{align} \label{Eq:DDModel_noise_distribution}
    \mathbf{w}_t \simiid \normal(0_{n_\auxiliary \times 1}, \mathbf{\Sigma}^{\sf w}), \quad
    v_t \simiid \normal(0_{p\times 1}, \Sigma^{\sf v})
\end{align}
with variances $\mathbf{\Sigma}^{\sf w} \in \symmetric^{n_\auxiliary}_+$ and $\Sigma^{\sf v} \in \symmetric^p_{++}$,
\begin{align} \label{Eq:DDModel_noise_variance}
    \mathbf{\Sigma}^{\sf w} := \Diag ( 0_{(n_\auxiliary - pL) \times (n_\auxiliary - pL)} ,\, \Sigma^\rho )
\end{align}
where $\Sigma^\rho := \mathcal{O} \Sigma^{\sf w} \mathcal{O}^\transpose \in \symmetric^{pL}_+$ is the variance of $\rho_t$.
The matrices $\mathbf{A}, \mathbf{B}, \mathbf{C}, D$ are known given offline data described in Section \ref{SECTION:DDSMPC:offline}, since they by definition only depend on matrices $\mathbf{\Gamma}$ and $D$ which are data-representable via Lemma \ref{LEMMA:system_quantity_data_representation}.
Hence, the auxiliary model \eqref{Eq:DDModel} can be interpreted as a non-minimal data-representable realization of system \eqref{Eq:LTI}. 
Nonetheless, the model is indeed stabilizable and detectable. 
\begin{lemma} \label{LEMMA:AuxModel_Stabilizability_Detectability}
    For the auxiliary model \eqref{Eq:DDModel} and matrix $\mathbf{\Sigma}^{\sf w}$ in \eqref{Eq:DDModel_noise_variance}, the pairs $(\mathbf{A}, \mathbf{B})$ and $(\mathbf{A}, \mathbf{\Sigma}^{\sf w})$ are stabilizable and the pair $(\mathbf{A}, \mathbf{C})$ is detectable.
\end{lemma}
\begin{proof} See \ref{APPENDIX:PROOF:LEMMA:AuxModel_Stabilizability_Detectability}. \end{proof}

\subsubsection{Auxiliary State Initial Condition}
The auxiliary model \eqref{Eq:DDModel} with the same input-output behavior to system \eqref{Eq:LTI} is a key component in constructing a data-driven counterpart to SMPC, while another essential is the relation between the states $x_t$ and $\mathbf{x}_t$, which we introduce next.
Suppose we are at a control step $t=k$ in a receding-horizon process. 
Similar to \eqref{Eq:initial_condition}, we model the auxiliary state $\mathbf{x}_k$ from \eqref{Eq:DDModel} following a prior distribution
\begin{align} \label{Eq:DDModel:initial_condition}
    \mathbf{x}_k \;\sim\; \normal(\boldsymbol{\mu}^{\sf x}_k, \mathbf{\Sigma}^{\sf x}).
\end{align}
The mean $\boldsymbol{\mu}^{\sf x}_k \in \real^{n_\auxiliary}$ in \eqref{Eq:DDModel:initial_condition} interpolates between two fixed vectors $\boldsymbol{\mu}^{\sf \hat x}_k, \boldsymbol{\mu}^{\sf \bar x}_k \in \real^{n_\auxiliary}$ with a decision variable $\theta \in [0,1]$,
\begin{align} \label{Eq:DDModel:interpolating_initial_condition}
    \boldsymbol{\mu}^{\sf x}_k := (1 - \theta) \, \boldsymbol{\mu}^{\sf \hat x}_k + \theta \, \boldsymbol{\mu}^{\sf \bar x}_k
\end{align}
wherein $\boldsymbol{\mu}^{\sf \hat x}_k$ and $\boldsymbol{\mu}^{\sf \bar x}_k$ are produced by a state estimator (see \eqref{Eq:DDModel:Kalman_filter}) and a noise-free model (see \eqref{Eq:DDModel:nominal_model}), respectively, of last control step.
At initial time $k=0$, the initial state mean $\boldsymbol{\mu}^{\sf x}_0$ is given as a parameter $\boldsymbol{\mu}^{\sf x}_\initial \in \real^{n_\auxiliary}$, i.e., we let $\boldsymbol{\mu}^{\sf \hat x}_0 := \boldsymbol{\mu}^{\sf x}_\initial$ and $\boldsymbol{\mu}^{\sf \bar x}_0 := \boldsymbol{\mu}^{\sf x}_\initial$.
The variance $\mathbf{\Sigma}^{\sf x} \in \symmetric^{n_\auxiliary}_+$ in \eqref{Eq:DDModel:initial_condition} is fixed as the unique positive semidefinite solution to the DARE \eqref{Eq:DDModel:state_variance:DARE},
\begin{subequations} \label{Eq:DDModel:state_variance} \begin{align} \label{Eq:DDModel:state_variance:DARE}
    & \mathbf{\Sigma}^{\sf x} = (\mathbf{A} - \mathbf{L}_{\sf L} \mathbf{C}) \mathbf{\Sigma}^{\sf x} \mathbf{A}^\transpose + \mathbf{\Sigma}^{\sf w} \\
\label{Eq:DDModel:state_variance:gain}
    & \mathbf{L}_{\sf L} := \mathbf{A} \mathbf{L}_{\sf K}, \quad 
    \mathbf{L}_{\sf K} := \mathbf{\Sigma}^{\sf x} \mathbf{C}^\transpose (\mathbf{C} \mathbf{\Sigma}^{\sf x} \mathbf{C}^\transpose + \Sigma^{\sf v})^{-1}
\end{align} \end{subequations}
given detectable $(\mathbf{A}, \mathbf{C})$ and stabilizable $(\mathbf{A}, \mathbf{\Sigma}^{\sf w})$ via Lemma \ref{LEMMA:AuxModel_Stabilizability_Detectability}, where we define Kalman gain $\mathbf{L}_{\sf K} \in \real^{n_\auxiliary \times p}$ and Luenberger observer gain $\mathbf{L}_{\sf L} \in \real^{n_\auxiliary \times p}$.
Not surprisingly, there is a close relationship between the distributions of $x_k$ and $\mathbf{x}_k$.

\begin{lemma}[Related Means of $x_k$ and $\mathbf{x}_k$] \label{LEMMA:state_mean_relation}
    For system \eqref{Eq:LTI} and auxiliary state $\mathbf{x}_t$ in \eqref{Eq:DDModel_state_definition}, if $\mu_k$ is the mean of $x_k$ and $\boldsymbol{\mu}_k$ is the mean of $\mathbf{x}_k$, then we have
\begin{align} \label{Eq:state_mean_relation} 
    \mu_k = \Phi_\original \,\widetilde{\mu}_k, \qquad
    \boldsymbol{\mu}_k = \Phi_\auxiliary \,\widetilde{\mu}_k,
\end{align}
    for some $\widetilde{\mu}_k \in \real^{mL+n(L+1)}$, with the matrices $\Phi_\original$ and $\Phi_\auxiliary$ defined in \ref{APPENDIX:PROOF:DDModel}.
\end{lemma}
\begin{proof}
    Given $x_k = \Phi_\original \,\xi_k$ and $\mathbf{x}_k = \Phi_\auxiliary \,\xi_k$ via Claim \ref{CLAIM:x_xi_relation}, we have \eqref{Eq:state_mean_relation} by choosing $\widetilde{\mu}_k$ as the mean of $\xi_k$.
\end{proof}
The result \eqref{Eq:state_mean_relation} will be leveraged in establishing equivalence between SMPC and our proposed method, as we will see in \eqref{Eq:initial_condition_relation} in Proposition \ref{PROPOSITION:equivalence_of_optimization_problems} and (d) in Assumption \ref{ASSUMPTION:parameter_choice_of_DDSMPC_algorithm}.

\subsubsection{Auxiliary State Estimation and Feedback}
The auxiliary model \eqref{Eq:DDModel} will now be used for both estimation and control purposes.
Analogous to the Kalman filter \eqref{Eq:Kalman_filter}, we formulate a Kalman filter for the auxiliary model \eqref{Eq:DDModel} as
\begin{subequations} \label{Eq:DDModel:Kalman_filter} \begin{align}
    \label{Eq:DDModel:Kalman_filter:postetior}
    \hat{\mathbf{x}}^\plus_t &:= \hat{\mathbf{x}}^\minus_t + \mathbf{L}_{\sf K} (y_t - \mathbf{C} \hat {\mathbf{x}}^\minus_t - D u_t), &&  t \in \integer_{[k,k+N)}
    \\
    \label{Eq:DDModel:Kalman_filter:prior}
    \hat{\mathbf{x}}^\minus_{t+1} &:= \mathbf{A} \hat{\mathbf{x}}^\plus_t + \mathbf{B} u_t, &&  t \in \integer_{[k,k+N)}
    \\
    \label{Eq:DDModel:Kalman_filter:initial}
    \hat{\mathbf{x}}^\minus_k &:= \boldsymbol{\mu}^{\sf x}_k
\end{align} \end{subequations}
where $\hat{\mathbf{x}}^\plus_t$ and $\hat{\mathbf{x}}^\minus_t$ are the posterior and prior estimates of $\mathbf{x}_t$, respectively, with $\mathbf{L}_{\sf K} \in \real^{n_\auxiliary \times p}$ in \eqref{Eq:DDModel:state_variance:gain}.
A noise-free model can be formed similarly as \eqref{Eq:nominal_model}, given initial condition \eqref{Eq:DDModel:initial_condition},
\begin{subequations} \label{Eq:DDModel:nominal_model} \begin{align}
    \label{Eq:DDModel:nominal_model:state}
    \thickbar{\mathbf{x}}_{t+1} &:= \mathbf{A} \thickbar{\mathbf{x}}_t + \mathbf{B} \thickbar u_t, &&  t \in \integer_{[k,k+N)} \\ 
    \label{Eq:DDModel:nominal_model:output}
    \thickbar{\mathbf{y}}_t &:= \mathbf{C} \thickbar{\mathbf{x}}_t + D \thickbar u_t, &&  t \in \integer_{[k,k+N)} \\
    \label{Eq:DDModel:nominal_model:initial}
    \thickbar{\mathbf{x}}_k &:= \boldsymbol{\mu}^{\sf x}_k, 
\end{align} \end{subequations}
where $\thickbar u_t \in \real^m$ is the nominal input decided through optimization, and $\thickbar{\mathbf{x}}_t \in \real^{n_\auxiliary}$ and $\thickbar{\mathbf{y}}_t \in \real^p$ are the resulting nominal state and output, respectively.
The affine output feedback policy \eqref{Eq:feedback_policy} from SMPC is now extended as
$\boldsymbol{\pi}_t(\cdot)$,
\begin{align} \label{Eq:DDModel:feedback_policy}
    u_t \gets \boldsymbol{\pi}_t(\hat{\mathbf{x}}^\minus_t) := \thickbar u_t - \mathbf{K} (\hat{\mathbf{x}}^\minus_t - \thickbar{\mathbf{x}}_t)
\end{align}
where the feedback gain $\mathbf{K} \in \real^{m\times n_\auxiliary}$ must be selected such that $\mathbf{A} - \mathbf{B}\mathbf{K}$ is Schur stable. We may again use an LQR-based design as in \eqref{Eq:feedback_gain}, yielding
\begin{align} \label{Eq:DDModel:feedback_gain}
    \mathbf{K} := (R + \mathbf{B}^\transpose \mathbf{P} \mathbf{B} + D^\transpose Q D )^{-1} (\mathbf{B}^\transpose \mathbf{P} \mathbf{A} + D^\transpose Q \mathbf{C}),
\end{align}
where $\mathbf{P} \in \symmetric^{n_\auxiliary}_+$ is the unique positive semidefinite solution, given the stabilizability of $(\mathbf{A}, \mathbf{B})$ and detectability of $(\mathbf{A}, \mathbf{C})$ by Lemma \ref{LEMMA:AuxModel_Stabilizability_Detectability}, to the DARE
\begin{align} \label{Eq:DDModel:feedback_DARE}
    \mathbf{P} = \mathbf{A}^\transpose \mathbf{P} (\mathbf{A} - \mathbf{B} \mathbf{K}) + \mathbf{C}^\transpose Q (\mathbf{C} - D \mathbf{K}).
\end{align}
The state estimator and feedback policy designs \eqref{Eq:DDModel:Kalman_filter} and \eqref{Eq:DDModel:feedback_policy} directly parallel the SMPC framework, providing a clear bridge between the model-based and data-driven control settings. While ultimately these state estimates are internal variables within the subsequent optimization problem and may be eliminated, retaining this structure provides conceptual clarity and illustrates the potential to similarly obtain data-driven versions of other estimator-based SMPC schemes.
{\setlength{\abovedisplayskip}{4pt}
\setlength{\belowdisplayskip}{4pt}
\subsection{Optimization Problem} \label{SECTION:DDSMPC:optimization}

\subsubsection{SDDPC Optimization Problem}
With results of Section \ref{SECTION:DDSMPC:auxiliary}, we are now ready to mirror the steps which led to \eqref{Eq:SMPC_reduced} and formulate a Stochastic Data-Driven Predictive Control (SDDPC) optimization problem.
First, following a similar process as that which led to \eqref{Eq:input_output_distribution}, we may combine \eqref{Eq:DDModel}, \eqref{Eq:DDModel_noise_distribution}, \eqref{Eq:DDModel:initial_condition}, \eqref{Eq:DDModel:Kalman_filter}, \eqref{Eq:DDModel:nominal_model} and \eqref{Eq:DDModel:feedback_policy}, to conclude that the input-output trajectory $(u_t, y_t)$ for $t \in \integer_{[k,k+N)}$ is normally distributed as $\normal(\col(\thickbar u_t, \thickbar{\mathbf{y}}_t), \mathbf{\Delta}_{t-k})$, where the covariance matrices $\mathbf{\Delta}_s \in \symmetric^{m+p}_+$ for $s \in \integer_{[0,N)}$ are computed as \eqref{Eq:DDModel:input_output_variance:Delta} using $\mathbf{\Lambda}_s \in \symmetric^{n_\auxiliary}_+$ defined as \eqref{Eq:DDModel:input_output_variance:Lambda},
\begin{subequations} \label{Eq:DDModel:input_output_variance} 
\begin{align}
\label{Eq:DDModel:input_output_variance:Delta}
    & \mathbf{\Delta}_s := \compactmat{- \mathbf{K} \\ \mathbf{C} \!-\! D \mathbf{K}} \mathbf{\Lambda}_s \compactmat{- \mathbf{K} \\ \mathbf{C} \!-\! D \mathbf{K}}^\transpose + \compactmat{0_{m \times m}\!\!\!\! \\& \mathbf{C} \mathbf{\Sigma}^{\sf x} \mathbf{C}^\transpose \!+\! \Sigma^{\sf v}} \\
\label{Eq:DDModel:input_output_variance:Lambda}
    & \mathbf{\Lambda}_s := \textstyle{\sum_{r=0}^{s-1}} (\mathbf{A} \!-\! \mathbf{B} \mathbf{K})^r \mathbf{L}_{\sf L} (\mathbf{C} \mathbf{\Sigma}^{\sf x} \mathbf{C}^\transpose \!+\! \Sigma^{\sf v}) \mathbf{L}_{\sf L}\!{}^\transpose (\mathbf{A} \!-\! \mathbf{B} \mathbf{K})^{r \transpose}
\end{align} \end{subequations}
with $\mathbf{L}_{\sf L}$ in \eqref{Eq:DDModel:state_variance:gain} and $\mathbf{K}$ in \eqref{Eq:DDModel:feedback_gain}.
Then, the SDDPC problem for computing $\thickbar u$ and $\theta$ at control step $t=k$ is written as
\begin{align} \label{Eq:DDSMPC_reduced} \begin{aligned}
    \minimize_{\thickbar u,\, \theta,\, p_{i,t}} \quad& \textstyle{\sum_{t=k}^{k+N-1}} J_t (\thickbar u_t, \thickbar{\mathbf{y}}_t) + \lambda_\theta \,\theta
    \\
    \mathrm{subject\;to} \quad& 
    \text{\eqref{Eq:DDModel:safety_constraint_reduced} for $t \in \integer_{[k,k+N)}$, \eqref{Eq:DDModel:interpolating_initial_condition}}, \eqref{Eq:DDModel:nominal_model}, \eqref{Eq:DDModel:terminal_condition}, 
\end{aligned} \end{align}
with the safety constraint for $t \in \integer_{[k,k+N)}$,
\begin{align} \label{Eq:DDModel:safety_constraint_reduced} \begin{aligned} &
    e_i\!{}^\transpose \neatmat{\thickbar u_t \\[-.25em] \thickbar{\mathbf{y}}_t} \leq f_i \!-\! \sqrt{e_i\!{}^\transpose \mathbf{\Delta}_{t-k} \, e_i} \, \mathrm{icdfn} (1 \!-\! p_{i,t}), \;\; i \in \integer_{[1,q]} \\&
    \textstyle{\sum_{i=1}^q} p_{i,t} = p, \qquad\quad
    p_{i,t} > 0, \;\; i \in \integer_{[1,q]}
\end{aligned} \end{align}
and with the terminal equality constraint
\begin{align} \label{Eq:DDModel:terminal_condition} \begin{aligned}
    & \thickbar u_{k+N-L} = \thickbar u_{k+N-L+1} = \cdots = \thickbar u_{k+N-1},
    \\
    & \thickbar{\mathbf{y}}_{k+N-L} = \thickbar{\mathbf{y}}_{k+N-L+1} = \cdots = \thickbar{\mathbf{y}}_{k+N-1}.
\end{aligned} \end{align}
Problem \eqref{Eq:DDSMPC_reduced} not only mirrors problem \eqref{Eq:SMPC_reduced} through the use of auxiliary model \eqref{Eq:DDModel}, but also introduces a novel formulation that explicitly delineates which quantities are replaced by their data-driven counterparts and which remain unchanged.

\subsubsection{Equivalence to SMPC Optimization Problem}

We now establish that the SDDPC problem \eqref{Eq:DDSMPC_reduced} and the SMPC problem \eqref{Eq:SMPC_reduced} have equal feasible sets and equal optimal sets, when the initial-condition parameters are related in the form of \eqref{Eq:state_mean_relation}.

\begin{proposition}[Equivalence of Optimization Problems] \label{PROPOSITION:equivalence_of_optimization_problems}
    If the parameters $\mu^{\sf \hat x}_k, \mu^{\sf \bar x}_k, \boldsymbol{\mu}^{\sf \hat x}_k, \boldsymbol{\mu}^{\sf \bar x}_k$ satisfy 
\begin{align} \label{Eq:initial_condition_relation}  \begin{aligned}
    \mu^{\sf \hat x}_k = \Phi_\original \,\widetilde{\mu}^{\sf \,\hat x}_k, \qquad
    \boldsymbol{\mu}^{\sf \hat x}_k = \Phi_\auxiliary \,\widetilde{\mu}^{\sf \,\hat x}_k, \\
    \mu^{\sf \bar x}_k = \Phi_\original \,\widetilde{\mu}^{\sf \,\bar x}_k, \qquad
    \boldsymbol{\mu}^{\sf \bar x}_k = \Phi_\auxiliary \,\widetilde{\mu}^{\sf \,\bar x}_k,
\end{aligned} \end{align}
    with the matrices $\Phi_\original, \Phi_\auxiliary$ defined in \ref{APPENDIX:PROOF:DDModel}, 
    for some vectors $\widetilde{\mu}^{\sf \,\hat x}_k, \widetilde{\mu}^{\sf \,\bar x}_k \in \real^{mL+n(L+1)}$,
    then the optimal (resp. feasible) solution set of SDDPC problem \eqref{Eq:DDSMPC_reduced} is equal to the optimal (resp. feasible) solution set of SMPC problem \eqref{Eq:SMPC_reduced}.
\end{proposition}
\begin{proof}
    See \ref{APPENDIX:PROOF:equivalence_of_optimization_problems}.
\end{proof}

We conclude by noting that problem \eqref{Eq:DDSMPC_reduced} produces a unique optimal $(\thickbar u, \theta)$ when feasible, following from Proposition \ref{PROPOSITION:equivalence_of_optimization_problems} and the fact that problem \eqref{Eq:SMPC_reduced} gives a unique optimal $(\thickbar u, \theta)$ when it is feasible, as mentioned in Section \ref{SECTION:problem:SMPC}.

}
\subsection{Online Control Algorithm} \label{SECTION:DDSMPC:control}

\subsubsection{SDDPC Control Algorithm}

We now describe the online implementation of our SDDPC. At control step $t = k$, the nominal input $\thickbar u_{[k,k+N)}$ and interpolation variable $\theta$ are solved from \eqref{Eq:DDSMPC_reduced}, and the policies $\boldsymbol{\pi}_{[k, k+N)}$ are constructed via \eqref{Eq:DDModel:feedback_policy}, where the first $N_{\rm c}$ policies are implemented. At the next control step $t = k + N_{\rm c}$, the initial condition interpolates as in \eqref{Eq:DDModel:interpolating_initial_condition} between vectors $\boldsymbol{\mu}^{\sf \hat x}_{k+N_{\rm c}}$ and $\boldsymbol{\mu}^{\sf \bar x}_{k+N_{\rm c}}$ decided as
\begin{align} \label{Eq:DDModel:initial_condition_iteration}
    \boldsymbol{\mu}^{\sf \hat x}_{k+N_{\rm c}} := \hat{\mathbf{x}}^\minus_{k+N_{\rm c}}, \qquad
    \boldsymbol{\mu}^{\sf \bar x}_{k+N_{\rm c}} := \thickbar{\mathbf{x}}_{k+N_{\rm c}}.
\end{align}
The method is formally summarized in Algorithm \ref{ALGO:DDSMPC}. Note that Algorithm \ref{ALGO:DDSMPC} is for control of system \eqref{Eq:LTI}, although the auxiliary model is used in its design. Algorithm \ref{ALGO:DDSMPC} presents a novel control scheme, with analogy to Algorithm \ref{ALGO:SMPC}, where some components are replaced by data-driven counterparts.

\begin{algorithm}
\caption{Stochastic Data-Driven Predictive Control (SDDPC)} \label{ALGO:DDSMPC}
\begin{algorithmic}[1] 
    \Require horizon lengths $L, N, N_{\rm c}$, offline data $u^\data, y^\data$, noise variances $\Sigma^\rho, \Sigma^{\sf v}$, initial-state mean $\boldsymbol{\mu}^{\sf x}_\initial$, cost matrices $Q, R$, constraint coefficients $E, f$, probability bound $p$, interpolation penalty coefficient $\lambda_\theta$.
    \State Compute matrices $\mathbf{\Gamma}$ and $D$ as in Section \ref{SECTION:DDSMPC:offline} using data $u^\data, y^\data$, and formulate matrices $\mathbf{A}, \mathbf{B}, \mathbf{C}$ as in Section \ref{SECTION:DDSMPC:auxiliary}.
    \State Compute Kalman gain $\mathbf{L}_{\sf K}$ via \eqref{Eq:DDModel:state_variance:gain}, feedback gain $\mathbf{K}$ via \eqref{Eq:DDModel:feedback_gain}, and covariance matrices $\mathbf{\Delta}_{[0,N)}$ via \eqref{Eq:DDModel:input_output_variance}.
    \State Initialize the control step $k \gets 0$ and set the initial condition $\boldsymbol{\mu}^{\sf \hat x}_0 \gets \boldsymbol{\mu}^{\sf x}_\initial$ and $\boldsymbol{\mu}^{\sf \bar x}_0 \gets \boldsymbol{\mu}^{\sf x}_\initial$. 
    \While{\texttt{true}}
        \State Solve $\thickbar u_{[k,k+N)}$ and $\theta$ from problem \eqref{Eq:DDSMPC_reduced}.
        \State Obtain $\boldsymbol{\mu}^{\sf x}_k$ via \eqref{Eq:DDModel:interpolating_initial_condition} and obtain $\thickbar{\mathbf{x}}_{[k,k+N]}$ via \eqref{Eq:DDModel:nominal_model}.
        \State Obtain policies $\boldsymbol{\pi}_{[k,k+N)}$ from \eqref{Eq:DDModel:feedback_policy}.
    \For{{\bf $t$ from $k$ to $k+N_{\rm c}-1$}}
        \State Compute $\hat{\mathbf{x}}^\minus_t$ via \eqref{Eq:DDModel:Kalman_filter}.
        \State Input $u_t \gets \boldsymbol{\pi}_t(\hat{\mathbf{x}}^\minus_t)$ to the system \eqref{Eq:LTI}.
        \State Measure $y_t$ from the system \eqref{Eq:LTI}.
    \EndFor
        \State Set $\boldsymbol{\mu}^{\sf \hat x}_{k+N_{\rm c}} \gets \hat{\mathbf{x}}^\minus_{k+N_{\rm c}}$ and $\boldsymbol{\mu}^{\sf \bar x}_{k+N_{\rm c}} \gets \thickbar{\mathbf{x}}_{k+N_{\rm c}}$ as \eqref{Eq:DDModel:initial_condition_iteration}.
        \State Set $k \gets k + N_{\rm c}$.
    \EndWhile
\end{algorithmic}
\end{algorithm}

\subsubsection{Closed-loop Properties of SDDPC}
Similar to Lemma \ref{LEMMA:recursive_feasibility} and Lemma \ref{LEMMA:closed_loop_stability}, Algorithm \ref{ALGO:DDSMPC} possesses recursive feasibility and closed-loop stability, as formally stated below.

\setcounter{proposition}{\getrefnumber{LEMMA:recursive_feasibility}}
\setcounter{corollary}{0}
\begin{corollary}[SDDPC Recursive feasibility] \label{COROLLARY:recursive_feasibility}
    Assume $p \in (0,\frac12]$. In Algorithm \ref{ALGO:DDSMPC}, if the problem \eqref{Eq:DDSMPC_reduced} is feasible at control step $k=\kappa$, then it is feasible at next control step $k = \kappa + N_{\rm c}$.
\end{corollary}

\setcounter{proposition}{\getrefnumber{LEMMA:closed_loop_stability}}
\setcounter{corollary}{0}
\begin{corollary}[SDDPC Closed-loop Stability] \label{COROLLARY:closed_loop_stability}
    Assume $\{z : E z \leq f\}$ is a bounded set. Let system \eqref{Eq:LTI} be controlled by Algorithm \ref{ALGO:DDSMPC}, where problem \eqref{Eq:DDSMPC_reduced} is assumed feasible at all control steps, and the reference signal $r_t = r$ is fixed.
    Then, the expectation of optimal cost values at consecutive control steps differ as \eqref{Eq:LEMMA:closed_loop_stability:value_difference}, and the asymptotic expected cost is upper bounded as \eqref{Eq:LEMMA:closed_loop_stability:average_cost} with some $c \geq 0$.
\end{corollary}

\setcounter{proposition}{\getrefnumber{PROPOSITION:equivalence_of_optimization_problems}}
The proofs of the above corollaries are analogies to the proofs of Lemma \ref{LEMMA:recursive_feasibility} and Lemma \ref{LEMMA:closed_loop_stability}, respectively, where the auxiliary model \eqref{Eq:DDModel} is considered in place of model \eqref{Eq:LTI}.

\subsubsection{Equivalence to SMPC Algorithm}

We present in Theorem \ref{PROPOSITION:equivalence_of_control_algorithms} our main result, which says that under idealized conditions, our proposed SDDPC control method and the benchmark SMPC method will result in identical control actions.

\begin{assumption}[SDDPC Parameter Choice w.r.t. SMPC] \label{ASSUMPTION:parameter_choice_of_DDSMPC_algorithm}
    Given the parameters in Algorithm \ref{ALGO:SMPC}, we assume the parameters in Algorithm \ref{ALGO:DDSMPC} satisfy the following.
\begin{enumerate}[(a)]
    \item $L$ is sufficiently large so that $\mathcal{O}$ has full column rank. 
    \item Data $u^\data, y^\data$ comes from the deterministic system \eqref{Eq:LTI_deterministic}, and input data $u^\data$ is persistently exciting of order $L+n+1$.
    \item Given $\Sigma^{\sf w}$ in Algorithm \ref{ALGO:SMPC}, the parameter $\Sigma^\rho$ in Algorithm \ref{ALGO:DDSMPC} is set equal to $\mathcal{O} \Sigma^{\sf w} \mathcal{O}^\transpose$.
    \item Given $\mu^{\sf x}_\initial$ in Algorithm \ref{ALGO:SMPC}, the parameter $\boldsymbol{\mu}^{\sf x}_\initial$ in Algorithm \ref{ALGO:DDSMPC} is selected as $\Phi_\auxiliary \widetilde{\mu}^{\sf \,x}_\initial$ for some $\widetilde{\mu}^{\sf \,x}_\initial \in \real^{mL+(n+1)L}$ satisfying $\mu^{\sf x}_\initial = \Phi_\original \widetilde{\mu}^{\sf \,x}_\initial$, with matrices $\Phi_\original, \Phi_\auxiliary$ defined in \ref{APPENDIX:PROOF:DDModel}.
    (Such $\widetilde{\mu}^{\sf \,x}_\initial$ always exists because $\Phi_\original$ has full row rank.)
\end{enumerate}
\end{assumption}

\begin{theorem}[\bf Equivalence of SMPC and SDDPC] \label{PROPOSITION:equivalence_of_control_algorithms}
    Consider the stochastic system \eqref{Eq:LTI} with a specific initial state $x_0$ and a specific noise realization $\{w_t, v_t\}_{t=0}^\infty$, and consider the following two control processes: 
\begin{enumerate}[a)]
    \item decide control actions $\{u_t\}_{t=0}^\infty$ by executing Algorithm \ref{ALGO:SMPC};
    \item decide control actions $\{u_t\}_{t=0}^\infty$ by executing Algorithm \ref{ALGO:DDSMPC}, where the parameters satisfy Assumption \ref{ASSUMPTION:parameter_choice_of_DDSMPC_algorithm}.
\end{enumerate}
    Then, the state-input-output trajectories $\{x_t, u_t, y_t\}_{t=0}^\infty$ resulting from process a) and from process b) are the same.
\end{theorem}
\begin{proof} See \ref{APPENDIX:PROOF:equivalence_of_control_algorithms}. \end{proof}

Theorem \ref{PROPOSITION:equivalence_of_control_algorithms}  should be interpreted as equivalence between SMPC and SDDPC in the idealized setting. Specifically, it establishes that if the proposed SDDPC algorithm is provided with noise-free offline data, if the initial conditions set within SMPC and SDDPC match, and if the process noise variance $\Sigma^{\rho}$ in the algorithm is set in a specific idealized fashion relative to the original process noise variance $\Sigma^{\sf w}$, then the method will produce identical results to those obtained by applying SMPC. While in practice these assumptions will not hold, noisy offline data can be accommodated as discussed in Section \ref{SECTION:DDSMPC:offline}, and $\Sigma^{\rho}$ becomes a tuning parameter of our SDDPC method.


\section{Numerical Case Study} 
\label{SECTION:simulations}

In this section, we numerically test our proposed method on the nonlinear grid-connected power converter system from \cite{DeePC:huang2021}, shown in Fig. \ref{FIG:power_converter}, and we compare the results with those of several benchmark model-based and data-based techniques. 


\begin{figure}[b]
\includegraphics[width=.485\textwidth]{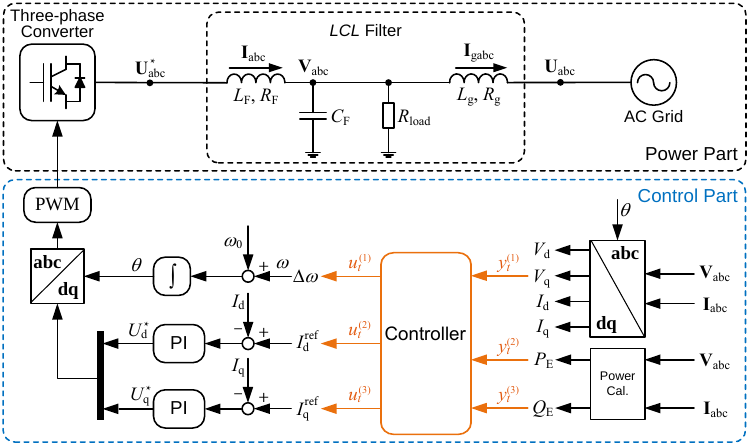}
\caption{The one-line diagram of a grid-connected power converter \cite[Fig. 1]{DeePC:huang2021}.}
\label{FIG:power_converter}
\end{figure}


The AC grid in the power part of Fig. \ref{FIG:power_converter} is modeled as an infinite bus with fixed voltage ($1$ p.u.) and fixed frequency ($1$ p.u.). This model has $n=6$ states, $m=3$ inputs and $p=3$ outputs. The inputs are the angular frequency correction $\Delta \omega$ and current references $I^{\rm ref}_{\rm d}$ and $I^{\rm ref}_{\rm q}$ of d- and q-axes, respectively. The outputs to be controlled are the q-axis voltage $V_{\rm q}$, the active power $P_{\rm E}$ and the reactive power $Q_{\rm E}$. The LCL-filter parameters and the PI parameters in Fig. \ref{FIG:power_converter} are consistent with \cite{DeePC:huang2021}, whereas we choose the load resistance $R_{\rm load}$ as a Gaussian signal with mean $4$ p.u. and noise power $10^{-3}$ p.u., which introduces process noise.
The measurement noise on each output is Gaussian with noise power $10^{-7}$ p.u., consistent with \cite{DeePC:huang2021}.

\subsection{Benchmark Control Methods}
\label{SECTION:simulations:benchmark}

In this subsection, we review several existing receding-horizon control methods which are performed in our simulations and compared to our proposed SDDPC.

\subsubsection{Stochastic MPC and (Deterministic) MPC}
We investigate two model-based methods, namely Stochastic MPC (SMPC) as in Section \ref{SECTION:problem:SMPC} and deterministic MPC (or MPC). 
For both SMPC and MPC, we use an identified system model in place of the true model $(A, B, C, D)$, through N4SID system identification method \cite{n4sid:van1994} using offline data $u^\data, y^\data$ collected from the system.
MPC follows a similar receding-horizon control process as SMPC, whereas the control action $u_t$ is the decision $\thickbar u_t$ by optimization, instead of using a feedback policy; the MPC optimization problem is similar to \eqref{Eq:SMPC_reduced}, but a deterministic safety constraint
\begin{align} \label{Eq:constraints_deterministic}
    E \, \col(\thickbar u_t, \thickbar y_t) \leq f
\end{align}
should be used in place of \eqref{Eq:safety_constraint_reduced}.

\subsubsection{DeePC and SPC}
We investigate L2-regularized DeePC \cite{DeePC:huang2021} and regularized SPC \cite{SPC:huang2008} as benchmark data-driven methods.
In DeePC and SPC, the decisions $\thickbar u_t$ of optimization are applied as control actions $u_t$, and the deterministic constraint \eqref{Eq:constraints_deterministic} is considered.
Using offline data $u^\data, y^\data$, we formulate data Hankel matrices $U_{\rm p}, U_{\rm f}, Y_{\rm p}, Y_{\rm f}$ similar to matrices $U_1, U_2, Y_1, Y_2$ in \eqref{Eq:data_matrices}, but matrices $U_{\rm p}, U_{\rm f}, Y_{\rm p}, Y_{\rm f}$ have $mL, mN, pL, pN$ rows respectively.
The DeePC optimization problem at control step $t=k$,
\begin{align*} \begin{aligned}
    \minimize_{g,\, \sigma_{\rm y}}\quad& \textstyle{\sum_{t=k}^{k+N-1}} J_t (\thickbar u_t, \thickbar y_t) + \lambda_{\rm y} \Vert \sigma_{\rm y} \Vert^2_2 + \lambda_{\rm g} \Vert g \Vert^2_2 \\
    \mathrm{subject\;to} \quad& \compactmat{U_{\rm p} \\ Y_{\rm p}} g = \compactmat{u_{\rm ini} \\ y_{\rm ini} \!+\! \sigma_{\rm y}},\;\;
    \compactmat{\thickbar u_{[k,k+N)} \\ \thickbar y_{[k,k+N)}} := \compactmat{U_{\rm f} \\ Y_{\rm f}} g \\
    & \eqref{Eq:constraints_deterministic} \text{ for } t \in \integer_{[k,k+N)}
\end{aligned} \end{align*}
where $u_{\rm ini} := u_{[k-L,k)}$ and $y_{\rm ini} := y_{[k-L,k)}$ are past inputs and outputs, and $\lambda_{\rm y} > 0$ and $\lambda_{\rm g} > 0$ are regularization parameters.
The SPC optimization problem at control step $t=k$,
\begin{align*}
    \minimize_{\thickbar u} \quad& \textstyle{\sum_{t=k}^{k+N-1}} J_t (\thickbar u_t, \thickbar y_t)\\
    \mathrm{subject\;to} \quad& \thickbar y_{[k,k+N)} := \widehat{\mathcal{P}}_{\sf spc} \; \col(u_{\rm ini}, y_{\rm ini}, \thickbar u_{[k,k+N)}) \\
    & \eqref{Eq:constraints_deterministic} \text{ for } t \in \integer_{[k,k+N)}
\end{align*}
where $\widehat{\mathcal{P}}_{\sf spc}$ is the Tikhonov regularization of the prediction matrix $\mathcal{P}_{\sf spc} := Y_{\rm f} \, \col(U_{\rm p}, Y_{\rm p}, U_{\rm f})^\dagger$, obtained similarly as $\widehat{\mathcal{P}}$ in Section \ref{SECTION:DDSMPC:offline}, with a regularization parameter $\lambda > 0$.

\subsection{Offline Data Collection}
\label{SECTION:simulations:offline}

Offline data is required in all our investigated control methods, for use in either data matrices (SDDPC, DeePC and SPC) or for system identification (MPC and SMPC).
In our simulation, the data collection process lasted for $1$ second and produced a data trajectory of length $T_{\rm d} = 1000$ with a sampling period of $1$ms.
The input data was generated as follows: $\Delta \omega$ (input 1) was set as the phase-locked loop (PLL) control action (see e.g. \cite{DeePCApp:huang2019gridConnected}) plus a white-noise signal, $I^{\rm ref}_{\rm d}$ (input 2) was set as $0.4$ p.u. plus a white-noise signal, and $I^{\rm ref}_{\rm q}$ (input 3) was set at $0$ p.u. plus a white-noise signal.
Each white noise signal had noise power of $10^{-6}$ p.u..

{
\begin{table}[t]
\caption{Control Parameters} \label{TABLE:controller_parameters}
\centering 
\begin{tabular}{@{}ll@{}} 
    \toprule
    \multicolumn{2}{c}{\textbf{Time Horizon Lengths}} \\ \midrule
    Initial-condition horizon length & $L = 5$ \\
    Prediction horizon length & $N = 30$ \\
    Control horizon length & $N_{\rm c} = 10$ \\
    \toprule
    \multicolumn{2}{c}{\textbf{Problem Setup Parameters}} \\ \midrule
    Sampling Period & $T_{\rm s} = 1$ms\\
    Cost matrices & $Q = 10^4 I_p$, $R = I_m$ \\
    Constraint coefficients & $E = I_{m+p} \otimes \smallmat{1\\-1}$, \\
    & $f = \smallmat{0.6 \times \mathbf{1}_{2m \times 1} \\ 0.4 \times \mathbf{1}_{2p \times 1}}$ \\
    Risk probability bound & $p = 0.2$ \\
    Interpolation penalty & $\lambda_\theta = 10$ \\
    Variance of $v_t$ for SMPC/SDDPC & $\Sigma^{\sf v} = 10^{-8} I_p$ \\
    Variance of $\rho_t$ for SDDPC & $\Sigma^\rho = 10^{-4} I_{pL}$ \\
    Variance of $w_t$ for SMPC$^{\rm a}$ & $\Sigma^{\sf w} = \mathcal{O}^\dagger \Sigma^\rho \mathcal{O}^{\dagger \transpose}$ \\ \toprule
    \multicolumn{2}{c}{\textbf{Regularization Parameters}} \\ \midrule
    Regularization in DeePC & $\lambda_{\rm y} = 10^6$, \, $\lambda_{\rm g} = 10^3$ \\ 
    Regularization of $\mathcal{P}$ in SDDPC & $\lambda = 10^{-3}$ \\
    Regularization of $\mathcal{P}_{\sf spc}$ in SPC & $\lambda = 10^{-3}$ \\ \bottomrule 
    \multicolumn{2}{@{$^{\rm a}$}l@{}}{In computation of $\Sigma^{\sf w}$, matrix $\mathcal{O}$ is obtained given the} \\
    \multicolumn{2}{l@{}}{identified system $(A,B,C,D)$ in SMPC.}
\end{tabular} \end{table}
}

\afterpage{
\begin{figure}[t]
\includegraphics[width=.48\textwidth]{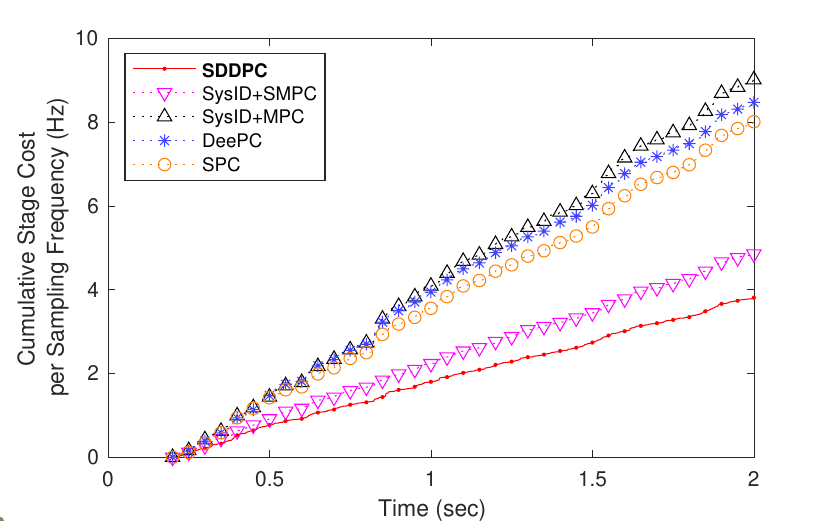}
\caption{Cumulative stage cost with different controllers, $N_{\rm c} = 10$.}
\label{FIG:cost_nc_10}
\end{figure}

\begin{figure}[t]
\includegraphics[width=.48\textwidth]{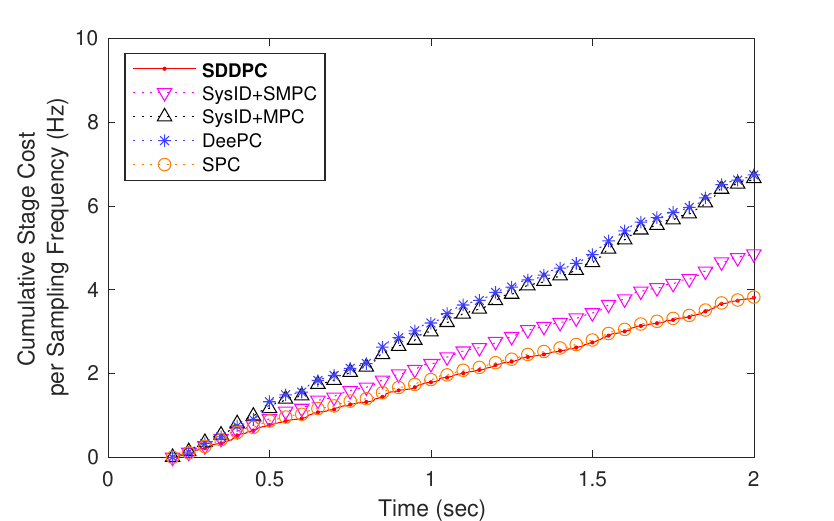}
\caption{Cumulative stage cost with different controllers,  $N_{\rm c} = 1$.}
\label{FIG:cost_nc_01}
\end{figure}
}
\subsection{Results}
\label{SECTION:simulations:result}

All controller parameters are reported in Table \ref{TABLE:controller_parameters}. Our simulation consists of two parts. 
In the first part, we compare the tracking performances of the different controllers.
In the second part, we examine the ability of the controllers to maintain safety constraints.

\subsubsection{Tracking Performance}
For each controller, we perform the following control process.
From time $0$s to time $0.2$s, the controller is switched off, and the inputs $I_{\rm d}^{\rm ref}$ and $I_{\rm q}^{\rm ref}$ are set to zero, with $\Delta \omega$ generated from the PLL.
After time $0.2$s, the controller is switched on, and the output reference signal is $r_t = [0,0,0]^\transpose$ before time $0.5$s and $r_t = [0,0.3,0]^\transpose$ after time $0.5$s. To quantitatively compare the results, Fig. \ref{FIG:cost_nc_10} shows the stage cost accumulated over the first two seconds for each controller.
The result shows that the stochastic control methods (SMPC and SDDPC) outperformed the deterministic control methods (DeePC, SPC and MPC) in terms of their cumulative costs.
This observation aligns with our expectation that stochastic control performs better with stochastic systems, since the stochastic control methods receive feedback at each time step -- more frequently than the deterministic control methods which receive feedback only at each control step, i.e., every $N_{\rm c}=10$ time steps.
However, this benefit of stochastic control vanishes when we select shorter control horizons.
Fig. \ref{FIG:cost_nc_01} shows the cumulative stage costs when the control horizon has length $N_{\rm c} = 1$, where we no longer observe a performance gap between all stochastic methods and all deterministic methods. SDDPC and SPC outperformed other controllers.
Although we showed the results with different $N_{\rm c}$, we emphasize significance of the $N_{\rm c} = 10$ setting, which requires less computation since the optimization problems are solved less frequently.

\subsubsection{Output Constraint Satisfaction}
We next evaluate for each controller its ability to meet the output safety constraints.
We repeat the control process above, but the reference signal becomes $r_t = [0,0,0]^\transpose$ before time $0.5$s and $r_t = [0,0.5,0]^\transpose$ after time $0.5$s.
Note that the reference value $0.5$ for the second output channel after time $0.5$s is beyond the range of output safety constraint (with $E, f$ in TABLE \ref{TABLE:controller_parameters}), which restricts all output channels within the range of $[-0.4, 0.4]$. As a result, in our simulations, the second output channel remained close to the upper safety bound of $0.4$ after time $0.5$s for all controllers; for example, the trace of the second output under SPC and SDDPC is displayed in Fig. \ref{FIG:constrained}.

To quantify the constraint satisfaction with each controller, from time $0.5$s to time $2.0$s ($1500$ time steps), we count the number and compute the rate of time steps where the measurement of the second output channel violates the safety constraint. As a second metric, we sum the amount of constraint violation that occurs between $0.5$s to $2.0$s for each controller.
The results are displayed in TABLE \ref{TABLE:constraint_violation}, where we also displayed the results of SMPC and SDDPC with parameter $p$ changed from $0.2$ (as in TABLE \ref{TABLE:controller_parameters}) to $0.05$.
As the result shows, both violation rates of SMPC and SDDPC declined as we decrease $p$, while the violation rate of SDDPC shrank more effectively than that of SMPC.
The total violation amounts of SMPC and SDDPC also reduced when we decrease $p$.
Among the methods using deterministic safety constraint, DeePC had a lower violation rate and a smaller violation amount than MPC and SPC.


\begin{figure}[t]
\includegraphics[width=.48\textwidth]{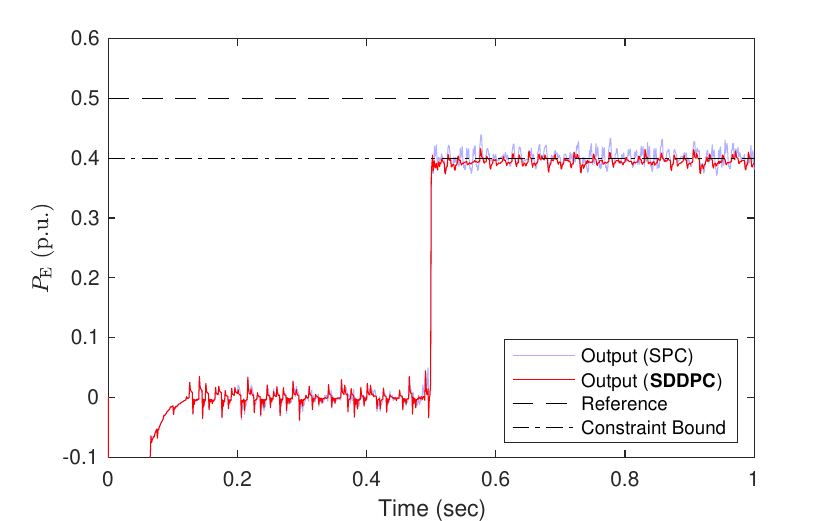}
\caption{The second output signals with SPC (light blue) and SDDPC (red) in the constraint satisfaction test. }
\label{FIG:constrained}
\end{figure}

\begin{table}[t]
\caption{Statistics of Constraint Violation \\[-.3em] of the Second Output Channel from $0.5$s to $2.0$s} \label{TABLE:constraint_violation}
\centering 
\begin{tabular}{@{}lcc@{}} 
    \toprule
    \textbf{Controller} & \begin{tabular}{c}\textbf{Violation Rate}\end{tabular} & \begin{tabular}{c}\textbf{Total Violation} \\ \textbf{Amount} \end{tabular} \\ \midrule
    SDDPC ($p = 0.2$) & $0.15$ & 1.10 \\
    SDDPC ($p = 0.05$) & $0.03$ & 0.05 \\
    \midrule
    SysID+SMPC ($p = 0.2$) & $0.19$ & 1.55 \\
    SysID+SMPC ($p = 0.05$) & $0.11$ & 0.52 \\
    SysID+MPC & $0.57$ & 6.79 \\
    DeePC & $0.20$ & 1.46 \\
    SPC & $0.49$ & 8.42 \\
    \bottomrule 
\end{tabular} \end{table}

\section{Conclusions} \label{SECTION:conclusion}
\color{black}

We introduced a novel direct data-driven control framework named Stochastic Data-Driven Predictive Control (SDDPC).
Analogous to Stochastic MPC (SMPC), SDDPC accounts for process and measurement noise in the control design, and produces closed-loop control policies through optimization.
On the theoretical front, we proved that SDDPC can produce control inputs equivalent to those of SMPC under specific conditions.
Simulation results indicate that the proposed approach provides benefits in terms of both cumulative stage cost and output constraint violation. 
Future work will seek to improve the computational efficiency of the approach, and to analyze and enhance the robustness with noisy offline data.
Other important directions include extension to non-Gaussian noise, optimization over the feedback gain $K$, and restriction of violation amount through, e.g., CVaR safety constraints.


\setcounter{section}{0}
\renewcommand{\thesection}{Appendix \Alph{section}}

\ifversionB{
\section{Proof of \eqref{Eq:input_output_distribution}} \label{APPENDIX:PROOF:input_output_distribution}

\begin{proof}

Define $e_t := \col(x_t - \hat x^\minus_t, \hat x^\minus_t - \thickbar x_t) \in \real^{2n}$.
We first show that $e_t$ follows the distribution
\begin{align} \label{Eq:PROOF:input_output_distribution:relation_1}
    e_t \sim \normal(0_{2n \times 1},\, \Diag(\Sigma^{\sf x}, \Lambda_{t-k}))
\end{align}
for $t \in \integer_{[k,k+N)}$, with $\Lambda_s$ in \eqref{Eq:input_output_variance:Lambda}, by induction on $t$.
\textbf{Base Case $t=k$.} With $\hat x^\minus_k = \mu^{\sf x}_k$ as \eqref{Eq:Kalman_filter:initial} and $\thickbar x_k = \mu^{\sf x}_k$ as \eqref{Eq:nominal_model:initial}, we have $e_k = \col(x_k - \mu^{\sf x}_k, 0_{n\times 1})$ which is distributed as $\normal(0_{2n \times 1}, \Diag(\Sigma^{\sf x}, 0_{n\times n}))$ via \eqref{Eq:initial_condition}. This shows the $t=k$ case of \eqref{Eq:PROOF:input_output_distribution:relation_1} given $\Lambda_0 = 0_{n \times n}$ from \eqref{Eq:input_output_variance:Lambda}.
\textbf{Inductive Step.} Assume \eqref{Eq:PROOF:input_output_distribution:relation_1} for $t = \tau \in \integer_{[k,k+N-2]}$. Note the relation \cite{OFSMPC:farina2015}
\begin{align}
\label{Eq:PROOF:input_output_distribution:relation_2}
    e_{\tau+1} = \Theta_0 e_\tau + \Theta_1 \col(w_\tau, \, v_\tau)
\end{align}
by expressing $x_{\tau+1}, \hat x^\minus_{\tau+1}, \thickbar x_{\tau+1}$ in terms of $x_\tau, \hat x^\minus_\tau, \thickbar x_\tau, w_\tau, v_\tau$ given \eqref{Eq:LTI:state}, \eqref{Eq:nominal_model:state}, \eqref{Eq:Kalman_filter:postetior}, \eqref{Eq:Kalman_filter:prior}, \eqref{Eq:feedback_policy}, where we define
\begin{align}
\label{Eq:PROOF:input_output_distribution:definition}
    \Theta_0 := \compactmat{A - L_{\sf L} C & 0_{n \times n}\\ L_{\sf L} C & A - B K}, \quad
    \Theta_1 := \compactmat{I_n & - L_{\sf L} \\ 0_{n \times n} & L_{\sf L}}.
\end{align}
Through the system \eqref{Eq:LTI} and the estimator \eqref{Eq:Kalman_filter}, both $w_\tau$ and $v_\tau$ are independent of $x_\tau$ and $\hat x^\minus_\tau$ and thus independent of $e_\tau$.
It follows from the relation \eqref{Eq:PROOF:input_output_distribution:relation_2}, the (independent) distribution of $w_\tau, v_\tau$ in \eqref{Eq:noise_distribution} and the distribution of $e_\tau$ in \eqref{Eq:PROOF:input_output_distribution:relation_1} that $e_{\tau+1}$ is distributed as
\begin{align}
\label{Eq:PROOF:input_output_distribution:relation_3}
    e_{\tau+1} \sim \normal (0_{2n\times 1}, \, \Theta_0 \smallmat{\Sigma^{\sf x} \\[-.1em] & \Lambda_{\tau-k}} \Theta_0^\transpose + \Theta_1 \smallmat{\Sigma^{\sf w}\\& \Sigma^{\sf v}} \Theta_1^\transpose ).
\end{align}
The variance in \eqref{Eq:PROOF:input_output_distribution:relation_3} is equal to what follows, by substitution of $\Theta_0$ and $\Theta_1$ in \eqref{Eq:PROOF:input_output_distribution:definition} and direct matrix multiplication, 
\begin{align}
\label{Eq:PROOF:input_output_distribution:expression}
    \compactmat{ \mathcal{S}_0 - \mathcal{S}_1 - \mathcal{S}_1^\transpose + \mathcal{S}_2 + \Sigma^{\sf x} 
    & \mathcal{S}_1^\transpose - \mathcal{S}_0 \\ 
    \mathcal{S}_1 - \mathcal{S}_0 
    & \mathcal{S}_0 + (A\!-\!BK)\Lambda_{\tau-k} (A\!-\!BK)^\transpose}
\end{align}
where we use shortcuts $\mathcal{S}_0 := L_{\sf L} (C \Sigma^{\sf x} C^\transpose + \Sigma^{\sf v}) L_{\sf L}\!{}^\transpose$, $\mathcal{S}_1 := L_{\sf L} C \Sigma^{\sf x} A^\transpose$ and $\mathcal{S}_2 := A \Sigma^{\sf x} A^\transpose + \Sigma^{\sf w} - \Sigma^{\sf x}$.
Notice that $\mathcal{S}_0 = \mathcal{S}_1$ by definition of $L_{\sf L}$ in \eqref{Eq:state_variance:gain}, and $\mathcal{S}_1 = \mathcal{S}_2$ via \eqref{Eq:DDModel:state_variance:DARE}.
One can also verify that $\mathcal{S}_0 + (A-BK) \Lambda_s (A-BK)^\transpose = \Lambda_{s+1}$ for all $s \in \natural_{\geq 0}$, using definition \eqref{Eq:input_output_variance:Lambda}.
Thus, the matrix \eqref{Eq:PROOF:input_output_distribution:expression} is equal to $\Diag(\Sigma^{\sf x}, \Lambda_{\tau-k+1})$, which implies that \eqref{Eq:PROOF:input_output_distribution:relation_3} is the $t = \tau+1$ case of \eqref{Eq:PROOF:input_output_distribution:relation_1}.
Induction on $t$ shows \eqref{Eq:PROOF:input_output_distribution:relation_1} for $t \in \integer_{[k,k+N)}$.

Finally, we show \eqref{Eq:input_output_distribution} for $t \in \integer_{[k,k+N)}$ by noting that
\begin{align}
\label{Eq:PROOF:input_output_distribution:relation_4}
    \compactmat{u_t \\ y_t} \!=\! \compactmat{\thickbar u_t \\ \thickbar y_t} \!+\! \Psi e_t \!+\! \compactmat{0_{m \times 1} \\ v_t}
    \;\text{with}\; \Psi \!:=\! \compactmat{0_{m \times n} & -K \\ C & C\!-\!DK},
\end{align}
given \eqref{Eq:LTI:output} and \eqref{Eq:feedback_policy}. With the distribution \eqref{Eq:PROOF:input_output_distribution:relation_1} of $e_t$ and the distribution of $v_t$ in \eqref{Eq:noise_distribution}, where $e_t$ and $v_t$ are independent, it follows from \eqref{Eq:PROOF:input_output_distribution:relation_4} that
\begin{align*}
    \compactmat{u_t \\ y_t} \sim \normal \big( \compactmat{\thickbar u_t \\ \thickbar y_t}, \Psi \smallmat{\Sigma^{\sf x} \\[-.1em] & \Lambda_{t-k}} \Psi^\transpose + \smallmat{0_{m \times m} \\ & \Sigma^{\sf v}} \big),
\end{align*}
in which the variance can be verified equal to $\Delta_{t-k}$ defined in \eqref{Eq:input_output_variance:Delta} through direct calculation, and thus the above distribution is equivalent to \eqref{Eq:input_output_distribution}.
\end{proof}

\section{Iterative Risk Allocation} \label{APPENDIX:iterative_risk_allocation}

We record here an efficient method for solving the convex problem \eqref{Eq:SMPC_reduced}, known as Iterative Risk Allocation \cite{Ono2008}, described in Algorithm \ref{ALGO:iterative_risk_allocation}.

To begin, note that if we fix all variables $p_{i,t}$, then problem \eqref{Eq:SMPC_reduced} is reduced into the quadratic problem
\begin{align} \label{Eq:SMPC_subproblem} \begin{aligned}
    \minimize_{\thickbar u,\, \theta} \quad& \textstyle{\sum_{t=k}^{k+N-1}} J_t (\thickbar u_t, \thickbar y_t) + \lambda_\theta \, \theta
    \\
    \mathrm{subject\;to} \quad& 
    \!\!
    \text{\eqref{Eq:safety_constraint_reduced:constraint} for $t \in \integer_{[k,k+N)}$}, \eqref{Eq:nominal_model}, \eqref{Eq:interpolating_initial_condition}, \eqref{Eq:terminal_condition},
\end{aligned} \!\! \end{align}
which can be efficiently solved. The optimal solution to \eqref{Eq:SMPC_reduced} is the infimum of the solution to \eqref{Eq:SMPC_subproblem} over all $p_{i,t}$ satisfying \eqref{Eq:safety_constraint_reduced:risk_allocation}.
Hence, we solve problem \eqref{Eq:SMPC_subproblem} repeatedly with updated $p_{i,t}$ until the objective value converges with no significant change.
The entire process shows in Algorithm \ref{ALGO:iterative_risk_allocation}, which extends \cite[Algorithm 1]{Ono2008} from their single chance constraint into our separate chance constraints over time steps.
Newly introduced parameters are a shrinkage rate $\alpha \in (0,1)$ and a termination threshold $\varepsilon > 0$.
The initialization at line \ref{LINE:IRA:initialization} ensures feasibility of problem \eqref{Eq:SMPC_subproblem}, due to recursive feasibility.
From line \ref{LINE:IRA:getting_indicators}, we obtain binary indicators $a_{i,t} \in \{0,1\}$ showing whether constraint \eqref{Eq:safety_constraint_reduced:constraint} is active or not for each $(i,t)$.
This indicator is utilized in the process of updating $p_{i,t}$ in lines \ref{LINE:IRA:risk_update:begin}-\ref{LINE:IRA:risk_update:end}.
Note that, when the condition in line \ref{LINE:IRA:risk_update:condition} is true, the update routine in lines \ref{LINE:IRA:risk_update:begin}-\ref{LINE:IRA:risk_update:end} no longer makes change on $p_{i,t}$, so in this case the iteration terminates.
In line \ref{LINE:IRA:input_risk_update}, $\mathrm{cdfn}(z) := \frac12 + \frac12 \mathrm{erf} ( z / \sqrt2 )$ is the cumulative density function (c.d.f.) of the standard normal distribution, with $\mathrm{erf}$ the error function.

Similarly, problem \eqref{Eq:DDSMPC_reduced} can also be solved by Algorithm \ref{ALGO:iterative_risk_allocation} with $A$, $B$, $C$, $\mu^{\sf x}_k$, $\mu^{\sf \hat x}_k$, $\mu^{\sf \bar x}_k$, $\Delta_s$, $\thickbar y_t$ replaced by $\mathbf{A}$, $\mathbf{B}$, $\mathbf{C}$, $\boldsymbol{\mu}^{\sf x}_k$, $\boldsymbol{\mu}^{\sf \hat x}_k$, $\boldsymbol{\mu}^{\sf \bar x}_k$, $\mathbf{\Delta}_s$, $\thickbar{\mathbf{y}}_t$ respectively.

\begin{algorithm} \caption{Iterative Risk Allocation for solving \eqref{Eq:SMPC_reduced}} \label{ALGO:iterative_risk_allocation}

\begin{algorithmic}[1]
    \Require horizon lengths $L, N$, system matrices $A,B,C,D$, interpolation options $\mu^{\sf \hat x}_k, \mu^{\sf \bar x}_k$, cost matrices $Q, R$, constraint coefficients $E, f$, probability bound $p$, interpolation penalty coefficient $\lambda_\theta$, input-output variances $\Delta_{[0,N)}$, shrinkage rate $\alpha$, termination threshold $\varepsilon$, and the risk allocation $p^{\sf last}_{i,t}$ solved at last control step.
    \Ensure An approximate solution $(\thickbar u, \theta, p_{i,t})$ to problem \eqref{Eq:SMPC_reduced}.
    \State Initialize $p_{i,t} \gets p^{\sf last}_{i,s(t)}$ for $t \in \integer_{[k,k+N)}$ and $i \in \{1, \ldots, q\}$, where $s(t) := \min(t, k+N-N_{\rm c}-1)$.
    \label{LINE:IRA:initialization}
    \State Initialize $J^\star_{\sf prev} \gets +\infty$.
    \While{\texttt{true}}
    \State \parbox[t]{\linewidth-\algorithmicindent}{Solve $(\thickbar u, \thickbar y, \theta)$ from problem \eqref{Eq:SMPC_subproblem} and obtain the cost value $J^\star$. Record whether the constraints \eqref{Eq:safety_constraint_reduced:constraint} is active or not for each $(i, t)$.}
    \State \textbf{if} $|J^\star_{\sf prev} - J^\star| \leq \varepsilon$ \textbf{then break else} $J^\star_{\sf prev} \gets J^\star$. 
    \State \parbox[t]{\linewidth-\algorithmicindent}{For $t \in \integer_{[k,k+N)}$ and $i \in \{1, \ldots, q\}$, let $a_{i,t} \gets 1$ if constraint \eqref{Eq:safety_constraint_reduced:constraint} is active for $(i,t)$, otherwise $a_{i,t} \gets 0$.} 
    \label{LINE:IRA:getting_indicators}
    \State \parbox[t]{\linewidth-\algorithmicindent}{$a^{\sf sum}_t \gets \sum_{i=1}^q a_{i,t}$ for all $t \in \integer_{[k,k+N)}$.} 
    \State \textbf{if} {$a^{\sf sum}_t \in \{0,q\}$ for all $t \in \integer_{[k,k+N)}$} \textbf{then break}.
    \label{LINE:IRA:risk_update:condition}

    \For{$t \in \integer_{[k,k+N)}$ such that $0 < a^{\sf sum}_t < q$}
    \label{LINE:IRA:risk_update:begin}
        \For{all $i \in \{1, \ldots, q\}$ such that $a_{i,t} = 0$}
        \State $p_{i,t} \gets \alpha p_{i,t} \!+\! (1\!-\!\alpha) \Big( 1 \!-\! \mathrm{cdfn} \Big( \displaystyle{\frac{f_i \!-\! e^\transpose_i \compactmat{\thickbar u_t \\ \thickbar y_t}} {\sqrt{e_i\!{}^\transpose \Delta_{t-k} e_i}}} \Big)\! \Big)$.
        \label{LINE:IRA:input_risk_update}
        \EndFor
        \State $p^{\sf residual}_t \gets p - \sum_{i=1}^q p_{i,t}$.
        \For{all $i \in \{1, \ldots, q\}$ such that $a_{i,t}=1$}
        \State $p_{i,t} \gets p_{i,t} + p^{\sf residual}_t / a^{\sf sum}_t$.
        \label{LINE:IRA:risk_update:end}
        \EndFor
    \EndFor
    \EndWhile
\end{algorithmic}
\end{algorithm}

\section{Proof of Lemma \ref{LEMMA:recursive_feasibility}} \label{APPENDIX:PROOF:recursive_feasibility}
\setcounter{proposition}{\getrefnumber{LEMMA:recursive_feasibility}}
\setcounter{claim}{0}

\begin{proof}
\setlength{\abovedisplayskip}{4pt}
\setlength{\belowdisplayskip}{4pt}
Let $\kappa^\plus := \kappa + N_{\rm c}$, and let $|_k$ denote variables calculated at control step $k \in \{\kappa, \kappa^\plus\}$.
Let $(\thickbar u^*, \theta^*, p^*_{i,t}) |_\kappa$ be the optimal solution to problem \eqref{Eq:SMPC_reduced} at $k = \kappa$, and consider the following solution $(\thickbar u^\diamond, \theta^\diamond, p^\diamond_{i,t}) |_{\kappa^\plus}$ at $k = \kappa^\plus$, cf. \cite{Kohler2022}, 
\begin{align}
\label{Eq:PROOF:LEMMA:recursive_feasibility:solution} 
    \thickbar u^\diamond_t |_{\kappa^\plus} := \thickbar u^*_{s(t)} |_\kappa, \quad
    \theta^\diamond |_{\kappa^\plus} := 1, \quad
    p^\diamond_{i,t} |_{\kappa^\plus} := p^*_{i,s(t)} |_\kappa,
\end{align}
for all $t \in \integer_{[\kappa^\plus, \kappa^\plus + N)}$ and $i \in \integer_{[1,q]}$,
where we let $s(t) := \min(t, \kappa+N-1)$.
In this proof, we will show that \eqref{Eq:PROOF:LEMMA:recursive_feasibility:solution} is a feasible solution to problem \eqref{Eq:SMPC_reduced}.
Let $\thickbar y^* |_\kappa$ (resp.  $\thickbar y^\diamond |_{\kappa^\plus}$) denote the resulting nominal output via \eqref{Eq:nominal_model}, \eqref{Eq:interpolating_initial_condition} given $(\thickbar u^*, \theta^*) |_\kappa$ (resp. $(\thickbar u^\diamond, \theta^\diamond) |_{\kappa^\plus}$), and we have the following.

\begin{claim} \label{CLAIM:resulting_nominal_output}
    Given $(\thickbar u^\diamond, \theta^\diamond) |_{\kappa^\plus}$ in \eqref{Eq:PROOF:LEMMA:recursive_feasibility:solution}, the nominal output is $\thickbar y^\diamond_t |_{\kappa^\plus} = \thickbar y^*_{s(t)} |_\kappa$ for $t \in \integer_{[\kappa^\plus, \kappa^\plus + N)}$.
\end{claim}
\begin{proof}
Since we choose $\theta^\diamond |_{\kappa^\plus} = 1$ in \eqref{Eq:PROOF:LEMMA:recursive_feasibility:solution}, the nominal states $\thickbar x_{\kappa^\plus}$ are the same over control steps $k \in \{\kappa, \kappa^\plus\}$, as
\begin{align} \label{Eq:PROOF:LEMMA:recursive_feasibility:same_initial_state}
    \thickbar x^\diamond_{\kappa^\plus} |_{\kappa^\plus} 
    \overset{\text{via \eqref{Eq:nominal_model:initial}}}= \mu^{\sf x}_{\kappa^\plus} 
    \overset{\text{via \eqref{Eq:interpolating_initial_condition}}}= \mu^{\sf \bar x}_{\kappa^\plus} 
    \overset{\text{via \eqref{Eq:initial_condition_iteration}}}= \thickbar x^*_{\kappa^\plus} |_\kappa.
\end{align}
Given the same nominal states $\thickbar x_{\kappa^\plus}$ in \eqref{Eq:PROOF:LEMMA:recursive_feasibility:same_initial_state} and same nominal inputs $\thickbar u_{[\kappa^\plus, \kappa + N)}$ via \eqref{Eq:PROOF:LEMMA:recursive_feasibility:solution} over control steps $k \in \{\kappa, \kappa^\plus\}$,
the resulting nominal states and outputs are the same, i.e.,
\begin{subequations} \begin{align}
\label{Eq:PROOF:LEMMA:recursive_feasibility:same_state}
    & \thickbar x^\diamond_t |_{\kappa^\plus} = \thickbar x^*_t |_\kappa, \qquad t \in \integer_{[\kappa^\plus, \kappa+N]}, \\
\label{Eq:PROOF:LEMMA:recursive_feasibility:same_output}
    & \thickbar y^\diamond_t |_{\kappa^\plus} = \thickbar y^*_t |_\kappa, \qquad t \in \integer_{[\kappa^\plus, \kappa+N)}.
\end{align} \end{subequations}
Due to the terminal condition \eqref{Eq:terminal_condition} where $L$ is at least the system \emph{lag}, the observable component of the terminal state $\thickbar x^*_{\kappa + N} |_\kappa$ is in equilibrium with input $\thickbar u^*_{\kappa+N-1} |_\kappa$ and output $\thickbar y^*_{\kappa+N-1} |_\kappa$ \cite[Sec. 2.3]{DDMPC:berberich2020c}. (This statement does not require observability of the system and thus can be generalized for the proof of Corollary \ref{COROLLARY:recursive_feasibility} where the auxiliary system is considered.)
Given $\thickbar x^\diamond_{\kappa + N} |_{\kappa^\plus} = \thickbar x^*_{\kappa + N} |_\kappa$ via \eqref{Eq:PROOF:LEMMA:recursive_feasibility:same_state} and $\thickbar u^\diamond_t |_{\kappa^\plus} = \thickbar u^*_{\kappa+N-1} |_\kappa$ via \eqref{Eq:PROOF:LEMMA:recursive_feasibility:solution} for $t \in \integer_{[\kappa + N, \kappa^\plus+N)}$, the nominal output is in equilibrium as
\begin{align*}
    \thickbar y^\diamond_t |_{\kappa^\plus} = \thickbar y^*_{\kappa+N-1} |_\kappa, \quad t \in \integer_{[\kappa+N, \kappa^\plus+N)},
\end{align*}
which result together with \eqref{Eq:PROOF:LEMMA:recursive_feasibility:same_output} shows the claim.
\renewcommand\qedsymbol{$\blacklozenge$}
\end{proof}

We finish the proof by showing that the solution \eqref{Eq:PROOF:LEMMA:recursive_feasibility:solution} satisfies both constraints \eqref{Eq:safety_constraint_reduced} and \eqref{Eq:terminal_condition}.
The terminal constraint \eqref{Eq:terminal_condition} holds with solution \eqref{Eq:PROOF:LEMMA:recursive_feasibility:solution}, since we have $(\thickbar u^\diamond_t, \thickbar y^\diamond_t) |_{\kappa^\plus}$ for $t \in \integer_{[\kappa^\plus+N-L, \kappa^\plus +N)}$ all equal to
\begin{align*}
    (\thickbar u^\diamond_t, \thickbar y^\diamond_t) |_{\kappa^\plus} 
    = (\thickbar u^*_{s(t)}, \thickbar y^*_{s(t)}) |_\kappa 
    = (\thickbar u^*_{\kappa+N-1}, \thickbar y^*_{\kappa+N-1}) |_\kappa, 
\end{align*}
where the first equality is from \eqref{Eq:PROOF:LEMMA:recursive_feasibility:solution} and Claim \ref{CLAIM:resulting_nominal_output}, and the second equality is because constraint \eqref{Eq:terminal_condition} holds at $k = \kappa$.
Before showing satisfaction of \eqref{Eq:safety_constraint_reduced}, we claim a useful result.

\begin{claim} \label{CLAIM:variance_matrices_inequality}
    For $\Delta_s$ in \eqref{Eq:input_output_variance}, we have $\Delta_0 \preceq \Delta_1 \preceq \cdots \preceq \Delta_{N-1}$.
\end{claim}
\begin{proof}
    Given \eqref{Eq:input_output_variance:Delta}, the result follows from the fact $\Lambda_0 \preceq \Lambda_1 \preceq \cdots \preceq \Lambda_{N-1}$, which is clear from \eqref{Eq:input_output_variance:Lambda}.
\renewcommand\qedsymbol{$\blacklozenge$}
\end{proof}

Define $\mathcal{R} (\Delta_{t-k}) \subseteq \real^{m} \times \real^p \times \real^q$ the set of all $(\thickbar u_t, \thickbar y_t, p_{\boldsymbol\cdot,t})$ satisfying \eqref{Eq:safety_constraint_reduced}, where we let $p_{\boldsymbol\cdot,t} := \col(p_{1,t}, \ldots, p_{q,t}) \in \real^q$.
To show that constraint \eqref{Eq:safety_constraint_reduced} is satisfied by solution \eqref{Eq:PROOF:LEMMA:recursive_feasibility:solution}, it is equivalent to show that
\newcommand{\thinreduce}{\hspace{-.2ex}}
\begin{align*}
    & (\thinreduce \thickbar u^\diamond_t \thinreduce, \thickbar y^\diamond_t \thinreduce, p^\diamond_{\boldsymbol\cdot, t} \thinreduce) |_{\kappa^\plus} 
    \!=\! (\thinreduce \thickbar u^*_{s (\thinreduce t \thinreduce)} \thinreduce, \thickbar y^*_{s (\thinreduce t \thinreduce)} \thinreduce, p^*_{\boldsymbol\cdot, s (\thinreduce t \thinreduce)} \thinreduce) |_\kappa
    \!\in\! \mathcal{R} \thinreduce(\thinreduce \Delta_{s (\thinreduce t \thinreduce)-\kappa} \thinreduce)
    \!\subseteq\! \mathcal{R} \thinreduce(\thinreduce \Delta_{t-\kappa^\plus} \thinreduce)
\end{align*}
for all $t \in \integer_{[\kappa^\plus, \kappa^\plus+N)}$,
where the first equality uses \eqref{Eq:PROOF:LEMMA:recursive_feasibility:solution} and Claim \ref{CLAIM:resulting_nominal_output}, the belong sign ($\in$) is because constraint \eqref{Eq:safety_constraint_reduced} holds at $k = \kappa$, and the inclusion ($\subseteq$) comes from the fact $s(t) - \kappa \geq t - \kappa^\plus$ for $t \in \integer_{[\kappa^\plus, \kappa^\plus+N)}$ (implied by definition of $s(t)$ and $\kappa^\plus$) and from the fact $\mathcal{R} (\Delta_0) \supseteq \mathcal{R} (\Delta_1) \supseteq \ldots \supseteq \mathcal{R} (\Delta_{N-1})$, which is obtained given Claim \ref{CLAIM:variance_matrices_inequality}, given the definition of $\mathcal{R}(\cdot)$ based on \eqref{Eq:safety_constraint_reduced}, and given $\mathrm{icdfn}(1 - p_{i,t}) > 0$ for all $p_{i,t} < p \leq \frac12$.

Thus, the solution \eqref{Eq:PROOF:LEMMA:recursive_feasibility:solution} at $k = \kappa^\plus$ satisfies both constraints \eqref{Eq:safety_constraint_reduced} and \eqref{Eq:terminal_condition}, and the recursive feasibility is proved.
\end{proof}%

\section{Proof of Lemma \ref{LEMMA:closed_loop_stability}}
\label{APPENDIX:PROOF:closed_loop_stability}

\begin{proof}
\setlength{\abovedisplayskip}{4pt}
\setlength{\belowdisplayskip}{4pt}
Let $J(u_t, y_t)$ denote the cost \eqref{Eq:stage_cost} with the fixed reference $r_t = r$. Consider the optimal solution $(\thickbar u^*, \theta^*, p^*_{i,t}) |_k$ to problem \eqref{Eq:SMPC_reduced} at control step $k \in \{ \kappa, \kappa^\plus\}$, where $\kappa^\plus := \kappa + N_{\rm c}$, and let $\thickbar y^* |_k$ be the resulting nominal output. Similar to \eqref{Eq:cost_reduced}, the expectation of optimal cost $V^*_k$ for $k \in \{\kappa, \kappa^\plus\}$ is
\begin{align} \label{Eq:PROOF:closed_loop_stability:predictive_cost_expectation}
    \mathbb{E}_\kappa [V^*_k] &= \textstyle \sum_{t=k}^{k+N-1} [J(\thickbar u^*_t, \thickbar y^*_t) |_k + J^{\sf var}_{t-\kappa}].
\end{align}
Similarly, the expected cost over ${[\kappa, \kappa^\plus)}$ is
\begin{align} \label{Eq:PROOF:closed_loop_stability:control_horizon_cost_expectation}
    \textstyle \sum_{t=\kappa}^{\kappa^\plus-1} \mathbb{E}_\kappa [J_t (u_t, y_t)] = \sum_{t=\kappa}^{\kappa^\plus-1} [J(\thickbar u^*_t, \thickbar y^*_t) |_\kappa + J^{\sf var}_{t-\kappa}].
\end{align}
Through $[\eqref{Eq:PROOF:closed_loop_stability:predictive_cost_expectation}\;\text{of}\; k=\kappa] - [\eqref{Eq:PROOF:closed_loop_stability:predictive_cost_expectation}\;\text{of}\; k=\kappa^\plus] + \eqref{Eq:PROOF:closed_loop_stability:control_horizon_cost_expectation}$ and eliminating identical terms, we obtain the relation
\begin{align} \label{Eq:PROOF:closed_loop_stability:combined_result}
    & \textstyle \mathbb{E}_\kappa [V^*_{\kappa^\plus} \!-\! V^*_\kappa] = -\! \sum_{t=\kappa}^{\kappa^\plus-1} \mathbb{E}_\kappa [J_t (u_t, y_t)] \!+\! J_0 \!-\! J_1 \!+\! J_2,\!\!\!\!
\end{align}
where we used shortcuts $J_0 := \sum_{t=\kappa^\plus}^{\kappa^\plus+N-1} J(\thickbar u^*_t, \thickbar y^*_t) |_{\kappa^\plus}$, $J_1 := \sum_{t=\kappa^\plus}^{\kappa+N-1} J(\thickbar u^*_t, \thickbar y^*_t) |_\kappa$, and $J_2 := \sum_{t=\kappa+N}^{\kappa^\plus+N-1} J^{\sf var}_{t-\kappa}$.
Consider the feasible solution $(\thickbar u^\diamond, \theta^\diamond, p^\diamond_{i,t}) |_{\kappa^\plus}$ in \eqref{Eq:PROOF:LEMMA:recursive_feasibility:solution} to problem \eqref{Eq:SMPC_reduced}. Given \eqref{Eq:PROOF:LEMMA:recursive_feasibility:solution}, Claim \ref{CLAIM:resulting_nominal_output}, and the definition of $s(t)$, we have
\begin{align} \label{Eq:PROOF:closed_loop_stability:diamond_next}
    & \textstyle \sum_{t=\kappa^\plus}^{\kappa^\plus+N-1} J (\thickbar u^\diamond_t, \thickbar y^\diamond_t) |_{\kappa^\plus} 
    = \textstyle J_1 + N_{\rm c} J^{\sf ter}_\kappa
\end{align}
with $J^{\sf ter}_\kappa := J (\thickbar u^*_{\kappa+N-1}, \thickbar y^*_{\kappa+N-1}) |_\kappa$.
Let $J^{\sf sup}$ be the supremum of $J(\thickbar u_t, \thickbar y_t)$ over all $(\thickbar u_t, \thickbar y_t)$ in the feasible set $\mathcal{R}(\Delta_{N-1})$ defined in \ref{APPENDIX:PROOF:recursive_feasibility}; such $J^{\sf sup}$ is finite since $\mathcal{R}(\Delta_{N-1}) \subseteq \{z : E z \leq f\}$ is bounded. It follows that $J^{\sf ter}_\kappa \leq J^{\sf sup}$ by feasibility.
Moreover, we know by optimality that
\begin{align}
\label{Eq:PROOF:closed_loop_stability:optimality}
    J_0 + \lambda_\theta \, \theta^* |_{\kappa^\plus} \leq \textstyle \sum_{t=\kappa^\plus}^{\kappa^\plus+N-1} J (\thickbar u^\diamond_t, \thickbar y^\diamond_t) |_{\kappa^\plus} + \lambda_\theta \, \theta^\diamond |_{\kappa^\plus}.
\end{align}
Combining \eqref{Eq:PROOF:closed_loop_stability:diamond_next}, \eqref{Eq:PROOF:closed_loop_stability:optimality} and eliminating $J(\thickbar u^\diamond_t, \thickbar y^\diamond_t) |_{\kappa^\plus}$, we have
\begin{align} \label{Eq:PROOF:closed_loop_stability:partial_bound}
    J_0 \!-\! J_1 \leq \lambda_\theta (\theta^\diamond |_{\kappa^\plus} \!-\! \theta^* |_{\kappa^\plus}) + N_{\rm c} J^{\sf ter}_\kappa \leq \lambda_\theta \!+\! N_{\rm c} J^{\sf sup},
\end{align}
where the second inequality used $\theta \in [0,1]$ and $J^{\sf ter}_\kappa \leq J^{\sf sup}$.
Substituting \eqref{Eq:PROOF:closed_loop_stability:partial_bound} into \eqref{Eq:PROOF:closed_loop_stability:combined_result}, we obtain \eqref{Eq:LEMMA:closed_loop_stability:value_difference} with $c := J^{\sf sup} + (\lambda_\theta + J_2) / N_{\rm c}$.
Summing \eqref{Eq:LEMMA:closed_loop_stability:value_difference} over control steps $k = \kappa \in \{0, N_{\rm c}, 2 N_{\rm c}, \ldots, (T_{\rm c}-1) N_{\rm c}\}$ with some $T_{\rm c} \in \natural$ and then dividing it by $T_{\rm c} N_{\rm c}$, we have 
\begin{align*}
    \textstyle{\frac{1}{T_{\rm c} N_{\rm c}}} \mathbb{E}[V^*_{T_{\rm c} N_{\rm c}} - V^*_0] \leq c - \textstyle{\frac{1}{T_{\rm c} N_{\rm c}}} \textstyle{\sum_{t=0}^{T_{\rm c} N_{\rm c}-1}} \mathbb{E}_\kappa [J(u_t, y_t)],
\end{align*}
which implies \eqref{Eq:LEMMA:closed_loop_stability:average_cost} by taking $T := T_{\rm c} N_{\rm c}$ and $T_{\rm c} \to  \infty$.
\end{proof}

}\fi
\section{Proof of Lemma \ref{LEMMA:system_quantity_data_representation}} \label{APPENDIX:PROOF:LEMMA:system_quantity_data_representation}

\begin{proof}
    Let $(x^\data,u^\data,y^{\data})$ be the state-input-output trajectory of \eqref{Eq:LTI_deterministic}, and define $X_1, X_2 \in \real^{n \times h}$ as
\begin{align*}
    X_1 := \big[ x^{\rm d}_1, x^{\rm d}_2, \ldots, x^{\rm d}_h \big], \;\;
    X_2 := \big[ x^{\rm d}_{1+L}, x^{\rm d}_{2+L}, \ldots, x^{\rm d}_{h+L} \big].
\end{align*}
    It follows by straightforward algebra that  data matrices satisfy\!
\begin{subequations} \label{Eq:PROOF:data_representation:1} \begin{align}
    \label{Eq:PROOF:data_representation:1a}
    X_2 &= A^L X_1 + \mathcal{C} U_1, \\
    \label{Eq:PROOF:data_representation:1b}
    Y_1 &= \mathcal{O} X_1 + \mathcal{G} U_1, \\
    \label{Eq:PROOF:data_representation:1c}
    Y_2 &= C X_2 + D U_2.
\end{align} \end{subequations}
    Under our assumptions of controllability and persistent excitation, it follows from \cite[Corollary 2(iii)]{Willems2005} that the matrix $\col(X_1, U_1, U_2)$ has full row rank. Moreover, $\smallmat{I_{mL} \\ \mathcal{G} & \mathcal{O}}$ has full column rank, as it is block lower triangular and its diagonal blocks each has full column rank (Section \ref{SECTION:DDSMPC:offline}).

    First, the matrix $Y_2$ can be represented in terms of $X_1, U_1, U_2$ by combining \eqref{Eq:PROOF:data_representation:1a} and \eqref{Eq:PROOF:data_representation:1c} and eliminating $X_2$, i.e.,
\begin{align}
    \label{Eq:PROOF:data_representation:2a}
    Y_2 = \big[ C \mathcal{C} , C A^L , D \big] \, \col(U_1, X_1, U_2).
\end{align}
    We can also express $\col(U_1, Y_1, U_2)$ in terms of $X_1, U_1, U_2$ as
\begin{align*}
    \col(U_1, Y_1, U_2)
    = \Diag \big( \smallmat{I_{mL} \\ \mathcal{G} & \mathcal{O}}, I_m \big) \, \col(U_1, X_1, U_2)
\end{align*}
    using \eqref{Eq:PROOF:data_representation:1b}.
    As we know that $\Diag \big( \smallmat{I_{mL} \\ \mathcal{G} & \mathcal{O}}, I_m \big)$ has full column rank and $\col(U_1, X_1, U_2)$ has full row rank, the pseudo-inverse of above is \cite{greville1966}
\begin{align*}
    \col(U_1, Y_1, U_2)^\dagger
    = \col(U_1, X_1, U_2)^\dagger \, \Diag \big( \smallmat{I_{mL} \\ \mathcal{G} & \mathcal{O}}, I_m \big)^\dagger.
\end{align*}
    By multiplying \eqref{Eq:PROOF:data_representation:2a} and the relation above, we find the result
\begin{align*}
    & Y_2 \, \col(U_1, Y_1, U_2)^\dagger
    = \big[C \mathcal{C}, C A^L, D \big] \, 
    \Diag \big( \smallmat{I_{mL} \\ \mathcal{G} & \mathcal{O}}, I_m \big)^\dagger \\
    &= \mat{\big[C \mathcal{C}, C A^L \big] \smallmat{I_{mL} \\ \mathcal{G} & \mathcal{O}}^\dagger & D  }
    \overset{\text{via \eqref{Eq:Gamma_definition}}}{=} [\mathbf{\Gamma}_{\sf U}, \mathbf{\Gamma}_{\sf Y}, D ].
\qedhere
\end{align*}
\end{proof}

\section{Proof of Lemma \ref{LEMMA:AuxModel}} \label{APPENDIX:PROOF:DDModel}
\setcounter{proposition}{\getrefnumber{LEMMA:AuxModel}}
\setcounter{claim}{0}
\begin{proof}
    
We start with intermediate results Claim \ref{CLAIM:x_xi_relation} and Claim \ref{CLAIM:c_relation}. Define matrices $\Phi_\original \in \real^{n \times n_\xi}$ and $\Phi_\auxiliary \in \real^{n_\auxiliary \times n_\xi}$
\begin{align*}
    \Phi_\original := [\mathcal{C}, A^L, \mathcal{C}_{\sf w}], \quad
    \Phi_\auxiliary := \compactmat{ I_{mL} \\ \mathcal{G} & \mathcal{O} & \mathcal{G}_{\sf w} \\ & & I_L \otimes \mathcal{O} }
\end{align*}
    with matrix $\mathcal{O}$ in Section \ref{SECTION:problem}, matrices $\mathcal{C}, \mathcal{G}$ in Section \ref{SECTION:DDSMPC:offline} and $\mathcal{C}_{\sf w} := [A^{L-1}, \ldots, A, I_n]$ and $\mathcal{G}_{\sf w} := \smallmat{ 0_{p\times n} \\ C & 0_{p\times n} \\[-.3em] \scalebox{0.6}{\bf$\vdots$} & \scalebox{0.6}{\bf$\ddots$} & \scalebox{0.6}{\bf$\ddots$} \\ C A^{L-2} & \cdots & C & 0_{p\times n} }$.
\begin{claim} \label{CLAIM:x_xi_relation}
    For system \eqref{Eq:LTI} and the auxiliary state $\mathbf{x}_t$ in \eqref{Eq:DDModel_state_definition}, we have $x_t = \Phi_\original \,\xi_t$ and $\mathbf{x}_t = \Phi_\auxiliary \,\xi_t$, where we let $\xi_t := \col(u_{[t-L,t)}, x_{t-L}, w_{[t-L,t)}) \in \real^{n_\xi}$ with $n_\xi := mL+n(L+1)$.
\end{claim}
\begin{proof}
\renewcommand\qedsymbol{$\blacklozenge$}
    Given the system model \eqref{Eq:LTI}, the state $x_t$ and noise-free outputs $y^\circ_{[t-L,t)}$ can be expressed in terms of the previous state $x_{t-L}$, inputs $u_{[t-L,t)}$ and disturbances $w_{[t-L,t)}$ via
\begin{subequations} \begin{align}
\label{Eq:PROOF:DDModel:1a}
    x_t &= A^L \, x_{t-L} + \mathcal{C} \, u_{[t-L,t)} + \mathcal{C}_{\sf w} \, w_{[t-L,t)}, \\
\label{Eq:PROOF:DDModel:1b}
    y^\circ_{[t-L,t)} &= \mathcal{O} \, x_{t-L} + \mathcal{G} \, u_{[t-L,t)} + \mathcal{G}_{\sf w} \, w_{[t-L,t)}.
\end{align} \end{subequations}
    Thus, we have $x_t = \Phi_\original \,\xi_t$ given \eqref{Eq:PROOF:DDModel:1a} and the definitions of $\xi_t$ and $\Phi_\original$.
    Given the definition of $\mathbf{x}_t$ in \eqref{Eq:DDModel_state_definition} with $\rho_t := \mathcal{O} w_t$, we have $\mathbf{x}_t = \Phi_\auxiliary \,\xi_t$ implied by \eqref{Eq:PROOF:DDModel:1b}.
\end{proof}

\begin{claim} \label{CLAIM:c_relation}
    For system \eqref{Eq:LTI} and the auxiliary state $\mathbf{x}_t$ in \eqref{Eq:DDModel_state_definition}, we have $C x_t = \mathbf{C} \mathbf{x}_t$. Moreover, $C \Phi_\original = \mathbf{C} \Phi_\auxiliary$.
\end{claim}
\begin{proof}
\renewcommand\qedsymbol{$\blacklozenge$}
    With Claim \ref{CLAIM:x_xi_relation}, it suffices to show $C \Phi_\original = \mathbf{C} \Phi_\auxiliary$.
    Given the definitions of $\Phi_\original, \Phi_\auxiliary, \mathbf{C}$, we compute $C \Phi_\original$ as
\begin{align*}
    C \Phi_\original = [C \mathcal{C}, C A^L, C \mathcal{C}_{\sf w}]
\end{align*}
    and calculate $\mathbf{C} \Phi_\auxiliary$ as
\begin{align*}
    \mathbf{C} \Phi_\auxiliary &= [\mathbf{\Gamma}_{\sf U} + \mathbf{\Gamma}_{\sf Y} \mathcal{G},\, \mathbf{\Gamma}_{\sf Y} \mathcal{O},\, \mathbf{\Gamma}_{\sf Y} \mathcal{G}_{\sf w} + (\mathbf{F} - \mathbf{\Gamma}_{\sf Y} \mathbf{E}) (I_L \otimes \mathcal{O})] \\
    &= [\mathbf{\Gamma}_{\sf U} + \mathbf{\Gamma}_{\sf Y} \mathcal{G},\, \mathbf{\Gamma}_{\sf Y} \mathcal{O},\, C \mathcal{C}_{\sf w}]
    = [C \mathcal{C}, C A^L, C \mathcal{C}_{\sf w}],
\end{align*}
    where the second equality used the facts that $C \mathcal{C}_{\sf w} = \mathbf{F} (I_L \otimes \mathcal{O})$ and $\mathcal{G}_{\sf w} = \mathbf{E} (I_L \otimes \mathcal{O})$ which can be verified from the definitions of $\mathbf{E}, \mathbf{F}, \mathcal{C}_{\sf w}, \mathcal{G}_{\sf w}$, and the last equality above used the relation
\begin{align*}
    [\mathbf{\Gamma}_{\sf U} + \mathbf{\Gamma}_{\sf Y} \mathcal{G},\, \mathbf{\Gamma}_{\sf Y} \mathcal{O}]
    = [\mathbf{\Gamma}_{\sf U},\, \mathbf{\Gamma}_{\sf Y}] \smallmat{I_{mL} \\ \mathcal{G} & \mathcal{O}}
    = [C \mathcal{C}, C A^L]
\end{align*}
    where the last equality is due to the definition $[\mathbf{\Gamma}_{\sf U},\, \mathbf{\Gamma}_{\sf Y}] := [C \mathcal{C}, C A^L] \smallmat{I_{mL} \\ \mathcal{G} & \mathcal{O}}^\dagger$ where $\smallmat{I_{mL} \\ \mathcal{G} & \mathcal{O}}$ has full column rank.
    Comparing the above results of calculation, we have $C \Phi_\original = \mathbf{C} \Phi_\auxiliary$, and thus the result follows from Claim \ref{CLAIM:x_xi_relation}.
\end{proof}

    We directly have \eqref{Eq:DDModel:output} from Claim \ref{CLAIM:c_relation} and \eqref{Eq:LTI:output}.
    To show \eqref{Eq:DDModel:state}, we know by substitution that $\mathbf{A} \mathbf{x}_t + \mathbf{B} u_t + \mathbf{w}_t$ equals
\begin{align*}
    \col \Big( \compactmat{u_{[t-L+1,t)} \\ u_t}, \compactmat{y^\circ_{[t-L+1,t)} \\ \mathbf{C} \mathbf{x}_t + D u_t}, \compactmat{\rho_{[t-L+1,t)} \\ \rho_t} \Big),
\end{align*}
    given the definitions of $\mathbf{x}_t, \mathbf{w}_t, \mathbf{A}, \mathbf{B}$ with matrix $\mathbf{A}$ consisting of upper-shift matrices and the matrix $\mathbf{C}$.
    The above is $\mathbf{x}_{t+1}$ by definition, given the fact $\mathbf{C} \mathbf{x}_t + D u_t = y_t - v_t = y^\circ_t$ via \eqref{Eq:DDModel:output}, and thereby \eqref{Eq:DDModel:state} is obtained.
\end{proof}

\section{Proof of Lemma \ref{LEMMA:AuxModel_Stabilizability_Detectability}} \label{APPENDIX:PROOF:LEMMA:AuxModel_Stabilizability_Detectability}
\setcounter{proposition}{\getrefnumber{LEMMA:AuxModel_Stabilizability_Detectability}}
\setcounter{claim}{0}

\begin{proof}

The pair $(\mathbf{A}, \mathbf{C})$ is detectable by definition since there exists a matrix $\mathbf{L}^* := \col(0_{mL\times p}, 0_{p(L-1)\times p}, I_p, 0_{pL^2 \times p})$ such that $\mathbf{A} - \mathbf{L}^* \mathbf{C}$ equal to $\Diag(\mathcal{D}_m, \mathcal{D}_p, \mathcal{D}_{pL})$ is Schur stable, where $\mathcal{D}_q := \smallmat{& I_{q(L-1)} \\[-.25em] 0_{q \times q}} \in \real^{qL \times qL}$.

We show stabilizability of $(\mathbf{A}, \mathbf{B})$ and $(\mathbf{A}, \mathbf{\Sigma}^{\sf w})$ by establishing stabilizing gains.
Recall matrices $\Phi_\auxiliary \in \real^{n_\auxiliary \times n_\xi}$ and $\Phi_\original \in \real^{n \times n_\xi}$ defined in \ref{APPENDIX:PROOF:DDModel}, with $n_\auxiliary := mL + pL + pL^2$ and $n_\xi := mL + n + nL$. 
Define matrix $\Phi := [\Phi_{\sf U}, \Phi_{\sf Y}, \Phi_{\sf P}] \in \real^{n \times n_\auxiliary}$ whose sub-blocks are defined as
\begin{align*}
    & [\Phi_{\sf U}, \Phi_{\sf Y}] := [\mathcal{C}, A^L] \smallmat{I_{mL} \\ \mathcal{G} & \mathcal{O}}^\dagger \in \real^{n \times (mL + pL)}, \\
    & \Phi_{\sf P} := (\mathcal{C}_{\sf w} - \Phi_{\sf Y} \mathcal{G}_{\sf w}) (I_L \otimes \mathcal{O})^\dagger \in \real^{n \times pL^2}.
\end{align*}
We start with some basic results Claim \ref{CLAIM:Phi_relation} and Claim \ref{CLAIM:AuxModel_TildeModel_relation}.

\begin{claim} \label{CLAIM:Phi_relation} 
    $\Phi_\original = \Phi \Phi_\auxiliary$. 
\end{claim}
\begin{proof}
\renewcommand\qedsymbol{$\blacklozenge$}
    Given the definitions of $\Phi$ and $\Phi_\auxiliary$, with both $\smallmat{I_{mL} \\ \mathcal{G} & \mathcal{O}}$ and $I_L \otimes \mathcal{O}$ having full column rank, the product $\Phi \Phi_\auxiliary$ is calculated as $[\mathcal{C}, A^L, \mathcal{C}_{\sf w}]$, equal to $\Phi_\original$ by definition.
\end{proof}

\begin{claim}
\label{CLAIM:AuxModel_TildeModel_relation}
For matrices $\mathbf{A},\!\mathbf{B}$ in \eqref{Eq:DDModel} and $\mathbf{\Sigma}^{\sf w}\!$ in \eqref{Eq:DDModel_noise_variance}, we have
\begin{align*}
    \mathbf{A} \Phi_\auxiliary = \Phi_\auxiliary \widetilde{A}, \quad
    \mathbf{B} = \Phi_\auxiliary \widetilde{B}, \quad
    \mathbf{\Sigma}^{\sf w} = \Phi_\auxiliary \widetilde{\Sigma}^{\sf w} \Phi_\auxiliary^\transpose,
\end{align*}
with matrices $\widetilde{A} \in \real^{n_\xi \times n_\xi}, \widetilde{B} \in \real^{n_\xi \times m}, \widetilde{\Sigma}^{\sf w} \in \symmetric^{n_\xi}_+$ defined as
\begin{align} \label{Eq:TildeMatrices_Definition} \begin{aligned}
    \widetilde{A} &:=
    \left[ \compact{\begin{array}{cc|c|cc}
        & I_{m(L-1)} &&& \\
        0_{m \times m} &&&& \\ \hline
        B & 0_{n\times m(L-1)} & A & I_n & 0_{n\times n(L-1)} \\ \hline
        &&&& I_{n(L-1)} \\
        &&& 0_{n \times n} &
    \end{array}} \right],
    \\
    \widetilde{B} &:= 
    \left[ \compact{\begin{array}{c}
        \begin{matrix} 0_{m(L-1) \times m} \\ I_m \end{matrix} \\ \hline 
        0_{n\times m} \\ \hline 
        0_{nL\times m}
    \end{array}} \right]
    , \quad
    \widetilde{\Sigma}^{\sf w} := \compactmat{0_{(n_\xi-n) \times (n_\xi-n)} \\ & \Sigma^{\sf w}}.
\end{aligned}\!\! \end{align}
\end{claim}
\begin{proof}
\renewcommand\qedsymbol{$\blacklozenge$}
    $\mathbf{B} = \Phi_\auxiliary \widetilde{B}$ and $\mathbf{\Sigma}^{\sf w} = \Phi_\auxiliary \widetilde{\Sigma}^{\sf w} \Phi_\auxiliary^\transpose$ follow directly by expressing $\Phi_\auxiliary, \mathbf{B}, \mathbf{\Sigma}^{\sf w}, \widetilde{B}, \widetilde{\Sigma}^{\sf w}$ and underlying $\mathcal{G}, \mathcal{O}, \mathcal{G}_{\sf w}$ in terms of $A, B, C, D$ by definition.
    To show $\mathbf{A} \Phi_\auxiliary = \Phi_\auxiliary \widetilde{A}$, we first replace a subblock $\mathbf{C} \Phi_\auxiliary$ of $\mathbf{A} \Phi_\auxiliary$ (since $\mathbf{C}$ is a subblock of $\mathbf{A}$) using the relation $\mathbf{C} \Phi_\auxiliary = C \Phi_\original$ (shown in the proof of Claim \ref{CLAIM:c_relation}), and then similarly express all the matrices in terms of $A, B, C, D$ by definition.
\end{proof}

Define the following matrices, 
\begin{subequations} \label{Eq:PROOF:AuxModel_Stabilizability_Detectability:gains} \begin{align}
\label{Eq:PROOF:AuxModel_Stabilizability_Detectability:bold_gains}
    & \mathbf{K}^* := K \Phi,
    && \mathbf{K}^{\sf w} := \Phi_\auxiliary^{\dagger \transpose} \, \col(0_{(n_\xi-n) \times n}, K^{\sf w}) \, \Phi \\
\label{Eq:PROOF:AuxModel_Stabilizability_Detectability:tilde_gains}
    & \widetilde{K}^* := K \Phi_\original, 
    && \widetilde{K}^{\sf w} := \col( 0_{(n_\xi-n) \times n}, K^{\sf w}) \, \Phi_\original
\end{align} \end{subequations}
where $K$ is the feedback gain from \eqref{Eq:feedback_gain} and $K^{\sf w} \in \real^{n \times n}$ is an arbitrary matrix such that $A - \Sigma^{\sf w} K^{\sf w}$ is Schur stable.
It follows from Claim \ref{CLAIM:AuxModel_TildeModel_relation} and the definitions \eqref{Eq:PROOF:AuxModel_Stabilizability_Detectability:gains} that
\begin{subequations} \label{Eq:PROOF:AuxModel_Stabilizability_Detectability:partial_relations} \begin{align} \label{Eq:PROOF:AuxModel_Stabilizability_Detectability:partial_relations:K}
    (\mathbf{A} - \mathbf{B} \mathbf{K}^*) \Phi_\auxiliary &= \Phi_\auxiliary (\widetilde{A} - \widetilde{B} \widetilde{K}^*), \\
\label{Eq:PROOF:AuxModel_Stabilizability_Detectability:partial_relations:Kw}
    (\mathbf{A} - \mathbf{\Sigma}^{\sf w} \mathbf{K}^{\sf w}) \Phi_\auxiliary &= \Phi_\auxiliary (\widetilde{A} - \widetilde{\Sigma}^{\sf w} \widetilde{K}^{\sf w}),
\end{align} \end{subequations}
given Claim \ref{CLAIM:Phi_relation} and $\Phi_\auxiliary^\transpose \Phi_\auxiliary^{\dagger \transpose} = I_{n_\xi}$ for $\Phi_\auxiliary$ of full column rank.
We then claim several intermediate results Claim \ref{CLAIM:PROOF:AuxModel_Stabilizability_Detectability:tilde_stable}, Claim \ref{CLAIM:PROOF:AuxModel_Stabilizability_Detectability:bold_K_stable}, and Claim \ref{CLAIM:PROOF:AuxModel_Stabilizability_Detectability:bold_Kw_stable}.

\begin{claim}
\label{CLAIM:PROOF:AuxModel_Stabilizability_Detectability:tilde_stable} 
    For matrices $\widetilde{A}, \widetilde{B}, \widetilde{\Sigma}^{\sf w}$ in \eqref{Eq:TildeMatrices_Definition} and $\widetilde{K}^*, \widetilde{K}^{\sf w}$ in \eqref{Eq:PROOF:AuxModel_Stabilizability_Detectability:tilde_gains},
    both $\widetilde{A} - \widetilde{B} \widetilde{K}^*$ and $\widetilde{A} - \widetilde{\Sigma}^{\sf w} \widetilde{K}^{\sf w}$ are Schur stable.
\end{claim}
\begin{proof}
\renewcommand\qedsymbol{$\blacklozenge$}
Define $\xi_t := \col(u_{[t-L,t)}, x_{t-L}, w_{[t-L,t)}) \in \real^{n_\xi}$ and $\delta_t := \col(0_{(n_\xi - n) \times 1}, w_t) \in \real^{n_\xi}$. We have the relation
\begin{align}
\label{Eq:PROOF:AuxModel_Stabilizability_Detectability:tilde_stable:tilde_system} 
    \xi_{t+1} &= \widetilde{A} \xi_t + \widetilde{B} u_t + \delta_t
\end{align}
which can be verified given the system model \eqref{Eq:LTI:state} and the definition of $\widetilde{A}, \widetilde{B}$ in Claim \ref{CLAIM:AuxModel_TildeModel_relation}.

To show that $\widetilde{A} - \widetilde{B} \widetilde{K}^*$ is stable, consider the following process of system \eqref{Eq:LTI:state} starting at time $t = -L$: the initial state $x_{-L}$, the inputs $u_{[-L,0)}$ and the noises $w_{[-L,0)}$ are arbitrarily chosen (i.e., $\xi_0$ is arbitrary), the noise is $w_t = 0$ for $t \geq 0$, and the inputs $u_t$ for $t \geq 0$ are generated by state feedback $u_t = - K x_t$.
With this process, we have $x_{t+1} = (A-BK) x_t$ for $t \geq 0$, and hence $x_t \to 0$ as $t \to \infty$ because $A-BK$ is Schur stable. We therefore have $u_t, w_t \to 0$ and thus $\xi_t \to 0$ as $t \to \infty$, given the relations $u_t = - K x_t$ and $w_t = 0$ for $t \geq 0$ and the definition of $\xi_t$.
On the other hand, with the process, we have $\delta_t = 0$ since $w_t = 0$ for $t \geq 0$, and the state feedback $u_t = - K x_t$ can be written as $u_t = - K \Phi_\original \xi_t$ given the relation $x_t = \Phi_\original \xi_t$ from Claim \ref{CLAIM:x_xi_relation}, so we have $u_t = - \widetilde{K}^* \xi_t$ with $\widetilde{K}^*$ defined in \eqref{Eq:PROOF:AuxModel_Stabilizability_Detectability:tilde_gains}.
Therefore, the evolution \eqref{Eq:PROOF:AuxModel_Stabilizability_Detectability:tilde_stable:tilde_system} is reduced as $\xi_{t+1} = (\widetilde{A} - \widetilde{B} \widetilde{K}^*) \xi_t$ for $t \geq 0$, which implies that $\xi_t = (\widetilde{A} - \widetilde{B} \widetilde{K}^*)^t \xi_0$ for $t \geq 0$.
Since $\xi_t \to 0$ as $t \to \infty$ and $\xi_0$ is arbitrarily chosen, we conclude that $(\widetilde{A} - \widetilde{B} \widetilde{K}^*)^t \to 0$ as $t \to \infty$, i.e., $\widetilde{A} - \widetilde{B} \widetilde{K}^*$ is Schur stable.

To show that $\widetilde{A} - \widetilde{\Sigma}^{\sf w} \widetilde{K}^{\sf w}$ is stable, consider a similar process of system \eqref{Eq:LTI:state} from initial time $t = -L$: the initial state $x_{-L}$, the inputs $u_{[-L,0)}$ and the noises $w_{[-L,0)}$ are arbitrarily chosen (i.e., $\xi_0$ is arbitrary), the input is $u_t = 0$ for $t \geq 0$, and the disturbances $w_t$ for $t \geq 0$ are realized as $w_t = - \Sigma^{\sf w} K^{\sf w} x_t$. 
With the process, we have $x_{t+1} = (A - \Sigma^{\sf w} K^{\sf w}) x_t$ for $t \geq 0$, and hence $x_t \to 0$ as $t \to \infty$ because $A - \Sigma^{\sf w} K^{\sf w}$ is Schur stable. We therefore have $u_t, w_t \to 0$ and thus $\xi_t \to 0$ as $t \to \infty$, given the relations $u_t = 0$ and $w_t = - \Sigma^{\sf w} K^{\sf w} x_t$ for $t \geq 0$ and the definition of $\xi_t$.
On the other hand, with the process, we have $\delta_t = - \widetilde{\Sigma}^{\sf w} \widetilde{K}^{\sf w} \xi_t$ for $t \geq 0$, given the definition of $\delta_t$, the choice of noise $w_t = - \Sigma^{\sf w} K^{\sf w} x_t$, the relation $x_t = \Phi_\original \xi_t$ from Claim \ref{CLAIM:x_xi_relation} and the definitions of $\widetilde{K}^{\sf w}$ in \eqref{Eq:PROOF:AuxModel_Stabilizability_Detectability:tilde_gains} and $\widetilde{\Sigma}^{\sf w}$ in Claim \ref{CLAIM:AuxModel_TildeModel_relation}. 
Therefore, the evolution \eqref{Eq:PROOF:AuxModel_Stabilizability_Detectability:tilde_stable:tilde_system} is reduced as $\xi_{t+1} = (\widetilde{A} - \widetilde{\Sigma}^{\sf w} \widetilde{K}^{\sf w}) \xi_t$ for $t \geq 0$, which implies that $\xi_t = (\widetilde{A} - \widetilde{\Sigma}^{\sf w} \widetilde{K}^{\sf w})^t \xi_0$ for $t \geq 0$.
Since $\xi_t \to 0$ as $t \to \infty$ and $\xi_0$ is arbitrarily chosen, we conclude that $(\widetilde{A} - \widetilde{\Sigma}^{\sf w} \widetilde{K}^{\sf w})^t \to 0$ as $t \to \infty$, i.e., $\widetilde{A} - \widetilde{\Sigma}^{\sf w} \widetilde{K}^{\sf w}$ is Schur stable.
\end{proof}

\begin{claim} \label{CLAIM:PROOF:AuxModel_Stabilizability_Detectability:bold_K_stable} 
    For matrices $\mathbf{A}, \mathbf{B}$ in \eqref{Eq:DDModel} and $\mathbf{K}^*$ in \eqref{Eq:PROOF:AuxModel_Stabilizability_Detectability:bold_gains},
    if 
\begin{align} \label{Eq:PROOF:AuxModel_Stabilizability_Detectability:bold_K_stable_condition} 
    (\mathbf{A} - \mathbf{B} \mathbf{K}^*)^t \Phi_\auxiliary \to 0 \quad\text{as}\quad t \to \infty,
\end{align}
    then $\mathbf{A} - \mathbf{B} \mathbf{K}^*$ is Schur stable.
\end{claim}
\begin{proof}
\renewcommand\qedsymbol{$\blacklozenge$}
We calculate $\mathbf{A} - \mathbf{B} \mathbf{K}^*$ as
\begin{align*}
    \begin{aligned}[b] & \quad\, \overbrace{\qquad\qquad\qquad\qquad\qquad\qquad\quad}^{\compact{=: \mathcal{A}}} \\[-.6em] 
    & \left[ \compact{\begin{array}{c|c} \left[ \begin{array}{c|c}
        \begin{matrix} \begin{matrix} 0 & I_{m(L-1)} \end{matrix} \\ \hline - K \Phi_{\sf U} \end{matrix} & \begin{matrix} 0 \\ - K \Phi_{\sf Y} \end{matrix} \\ \hline
        \begin{matrix} 0 \\ (C - DK) \Phi_{\sf U} \end{matrix}
        & \begin{matrix} \begin{matrix} 0 & I_{p(L-1)} \end{matrix} \\ \hline (C - DK) \Phi_{\sf Y} \end{matrix}
    \end{array} \right]
    & \begin{matrix} 0 \\ - K \Phi_{\sf P} \\ 0 \\ \mathbf{F} - \mathbf{\Gamma}_{\sf Y} \mathbf{E} \end{matrix} \\ \hline
    0 & \begin{smallmatrix} & I_{pL(L-1)} \\ 0_{pL \times pL} \end{smallmatrix}
    \end{array}} \right], \end{aligned}
\end{align*}
which is Schur stable if, and only if, its sub-matrix $\mathcal{A}$ is Schur stable.
Moreover, since both $\mathbf{A} - \mathbf{B} \mathbf{K}^* = \smallmat{\mathcal{A} & * \\ 0 & *}$ and $\Phi_\auxiliary = \smallmat{\mathcal{S} & * \\ 0 & *}$ are upper block-triangular, $(\mathbf{A} - \mathbf{B} \mathbf{K}^*)^t \, \Phi_\auxiliary = \smallmat{\mathcal{A}^t \mathcal{S} & * \\ 0 & *}$ is also upper block-triangular.
Since $(\mathbf{A} - \mathbf{B} \mathbf{K}^*)^t \Phi_\auxiliary \to 0$ as $t \to \infty$ via \eqref{Eq:PROOF:AuxModel_Stabilizability_Detectability:bold_K_stable_condition}, its sub-matrix yields $\mathcal{A}^t \mathcal{S} \to 0$ as $t \to \infty$.

Let $\mathcal{L} := \lim_{t \to \infty} \mathcal{A}^t$ denote the limiting value.
Given the definition $[\Phi_{\sf U}, \Phi_{\sf Y}] := [\mathcal{C}, A^L] \mathcal{S}^\dagger$ where $\mathcal{S}$ denotes $\smallmat{I_{mL} \\ \mathcal{G} & \mathcal{O}}$, $\mathcal{A}$ can be written as $\mathcal{A} = \mathcal{D} + \mathcal{E} \mathcal{S}^\dagger$ where
\begin{align*}
    \mathcal{D} &:= \Diag \big( \smallmat{& I_{m(L-1)} \\ 0_{m\times m}}, \smallmat{& I_{p(L-1)} \\ 0_{p\times p}} \big), \\
    \mathcal{E} &:= \col(0_{m(L-1) \times n}, -K, 0_{p(L-1) \times n}, C-DK) \, [\mathcal{C}, A^L].
\end{align*}
Define $\mathcal{P} := I - \mathcal{S} \mathcal{S}^\dagger$ as a projection matrix. With the fact $\mathcal{S}^\dagger \mathcal{P} = \mathcal{S}^\dagger (I - \mathcal{S} \mathcal{S}^\dagger) = 0$, it follows that 
\begin{align*}
    \mathcal{A} \mathcal{S} \mathcal{S}^\dagger
    = \mathcal{A} - \mathcal{A} \mathcal{P}
    = \mathcal{A} - (\mathcal{D} + \mathcal{E} \mathcal{S}^\dagger) \mathcal{P} = \mathcal{A} - \mathcal{D} \mathcal{P}.
\end{align*}
Left-multiplying the above by $\mathcal{A}^{t-1}$ and taking the limit as $t\to\infty$, we find that
\begin{align*}
    \lim_{t \to \infty} \mathcal{A}^t \mathcal{S} \mathcal{S}^\dagger
    = \underbrace{\lim_{t \to \infty} \mathcal{A}^t}_{=\mathcal{L}} - \underbrace{\lim_{t \to \infty} \mathcal{A}^{t-1}}_{=\mathcal{L}} \mathcal{D} \mathcal{P}
\end{align*}
Since $\mathcal{A}^t \mathcal{S} \to 0$ as $t \to \infty$, the left-hand side of above is zero, so the above further reduces to $0 = \mathcal{L} (I - \mathcal{D} \mathcal{P})$.
Therefore, to show $\mathcal{L} = 0$, it suffices to show that $I - \mathcal{D} \mathcal{P}$ is non-singular.
Suppose a vector $z$ in $\mathrm{Null}(I - \mathcal{D} \mathcal{P})$.
If $z \notin \mathrm{Range}(\mathcal{P})$, then $\Vert \mathcal{P} z \Vert_2 < \Vert z \Vert_2$ for a projection matrix $\mathcal{P}$, and then we have
\begin{align*}
    \Vert z \Vert_2 = \Vert \mathcal{D} \mathcal{P} z \Vert_2 
    \leq \underbrace{\Vert \mathcal{D} \Vert_2}_{=1} \underbrace{\Vert \mathcal{P} z \Vert_2}_{< \Vert z \Vert_2}
    < \Vert z \Vert_2,
\end{align*}
which is a contradiction.
Hence, we know that $z \in \mathrm{Range}(\mathcal{P})$, which implies that $\mathcal{P} z = z$ because $\mathcal{P}$ is projection. Combining $z = \mathcal{D} \mathcal{P} z$ and $\mathcal{P} z = z$, we have $(I - \mathcal{D}) z = 0$, which implies $z = 0$ since $I - \mathcal{D}$ is non-singular. Therefore, we conclude that $\mathrm{Null}(I - \mathcal{D} \mathcal{P}) = \{0\}$ and $I - \mathcal{D} \mathcal{P}$ is non-singular, so we have $\mathcal{L} = 0$, which implies that $\mathcal{A}$ is Schur stable. Thus, $\mathbf{A} - \mathbf{B} \mathbf{K}^*$ is Schur stable.
\end{proof}
 
\begin{claim} \label{CLAIM:PROOF:AuxModel_Stabilizability_Detectability:bold_Kw_stable} 
    For matrices $\mathbf{A}, \mathbf{B}$ in \eqref{Eq:DDModel}, $\mathbf{\Sigma}^{\sf w}$ in \eqref{Eq:DDModel_noise_variance} and $\mathbf{K}^*, \mathbf{K}^{\sf w}$ in \eqref{Eq:PROOF:AuxModel_Stabilizability_Detectability:bold_gains},
    $\mathbf{A} - \mathbf{B} \mathbf{K}^*$ is Schur stable if, and only if, $\mathbf{A} - \mathbf{\Sigma}^{\sf w} \mathbf{K}^{\sf w}$ is Schur stable.
\end{claim}
\begin{proof}
\renewcommand\qedsymbol{$\blacklozenge$}
Since $\Phi_\auxiliary \in \real^{n_\auxiliary \times n_\xi}$ by definition has full column rank, there exists a matrix $\Phi_\orthogonal \in \real^{n_\auxiliary \times (n_\auxiliary - n_\xi)}$ such that $\mathrm{Range}(\Phi_\orthogonal) = \mathrm{Null}(\Phi_\auxiliary^\transpose)$; it follows that
\begin{align} \label{Eq:PROOF:AuxModel_Stabilizability_Detectability:Phi_orthogonal_property}
    \Phi_\auxiliary \Phi_\auxiliary^\dagger + \Phi_\orthogonal \Phi_\orthogonal^\dagger = I_{n_\auxiliary}.
\end{align}
Define matrices $\mathcal{S}^*, \mathcal{S}^{\sf w}, \mathcal{R}^*, \mathcal{R}^{\sf w}$,
\begin{align*}
    \mathcal{S}^* &:= \Phi_\orthogonal^\dagger (\mathbf{A} \!-\! \mathbf{B} \mathbf{K}^*) \Phi_\orthogonal, \!\!\!&
    \mathcal{S}^{\sf w} &:= \Phi_\orthogonal^\dagger (\mathbf{A} \!-\! \mathbf{\Sigma}^{\sf w} \mathbf{K}^{\sf w}) \Phi_\orthogonal \\
    \mathcal{R}^* &:= \Phi_\auxiliary^\dagger (\mathbf{A} \!-\! \mathbf{B} \mathbf{K}^*) \Phi_\orthogonal, \!\!\!&
    \mathcal{R}^{\sf w} &:= \Phi_\auxiliary^\dagger (\mathbf{A} \!-\! \mathbf{\Sigma}^{\sf w} \mathbf{K}^{\sf w}) \Phi_\orthogonal
\end{align*}
and it follows from \eqref{Eq:PROOF:AuxModel_Stabilizability_Detectability:Phi_orthogonal_property} that
\begin{align} \label{CLAIM:PROOF:AuxModel_Stabilizability_Detectability:residual_relations} \begin{aligned}
    (\mathbf{A} - \mathbf{B} \mathbf{K}^*) \Phi_\orthogonal &= \Phi_\auxiliary \mathcal{R}^* + \Phi_\orthogonal \mathcal{S}^*, \\
    (\mathbf{A} - \mathbf{\Sigma}^{\sf w} \mathbf{K}^{\sf w}) \Phi_\orthogonal &= \Phi_\auxiliary \mathcal{R}^{\sf w} + \Phi_\orthogonal \mathcal{S}^{\sf w}.
\end{aligned} \end{align}
We moreover notice that $\mathcal{S}^* = \mathcal{S}^{\sf w} = \Phi_\orthogonal^\dagger \mathbf{A} \Phi_\orthogonal$ given the definitions of $\mathcal{S}^*, \mathcal{S}^{\sf w}$ and the facts $\Phi_\orthogonal^\dagger \mathbf{B} = 0$ and $\Phi_\orthogonal^\dagger \mathbf{\Sigma}^{\sf w} = 0$ which follow from the fact $\Phi_\orthogonal^\dagger \Phi_\auxiliary = 0$ via \eqref{Eq:PROOF:AuxModel_Stabilizability_Detectability:Phi_orthogonal_property} and the relations
$\mathbf{B} = \Phi_\auxiliary \widetilde{B}$ and $\mathbf{\Sigma}^{\sf w} = \Phi_\auxiliary \widetilde{\Sigma}^{\sf w} \Phi_\auxiliary^\transpose$ from Claim \ref{CLAIM:AuxModel_TildeModel_relation}.

Define $\Phi_{\sf full} := [\Phi_\auxiliary, \Phi_\orthogonal] \in \real^{n_\auxiliary \times n_\auxiliary}$ which is non-singular given \eqref{Eq:PROOF:AuxModel_Stabilizability_Detectability:Phi_orthogonal_property};
the horizontal stack of \eqref{Eq:PROOF:AuxModel_Stabilizability_Detectability:partial_relations} and \eqref{CLAIM:PROOF:AuxModel_Stabilizability_Detectability:residual_relations} yields
\begin{align} \label{CLAIM:PROOF:AuxModel_Stabilizability_Detectability:similarity} \begin{aligned}
    (\mathbf{A} - \mathbf{B} \mathbf{K}^*) \Phi_{\sf full} &= \Phi_{\sf full} \, \compactmat{\widetilde{A} - \widetilde{B} \widetilde{K}^* & \mathcal{R}^* \\ 0 & \mathcal{S}^*}, \\
    (\mathbf{A} - \mathbf{\Sigma}^{\sf w} \mathbf{K}^{\sf w}) \Phi_{\sf full} &= \Phi_{\sf full} \, \compactmat{\widetilde{A} - \widetilde{\Sigma}^{\sf w} \widetilde{K}^{\sf w} & \mathcal{R}^{\sf w} \\ 0 & \mathcal{S}^{\sf w}}.
\end{aligned} \end{align}
Since $\widetilde{A} - \widetilde{B} \widetilde{K}^*$ and $\widetilde{A} - \widetilde{\Sigma}^{\sf w} \widetilde{K}^{\sf w}$ are Schur stable through Claim \ref{CLAIM:PROOF:AuxModel_Stabilizability_Detectability:tilde_stable}, the matrix similarity relations \eqref{CLAIM:PROOF:AuxModel_Stabilizability_Detectability:similarity} imply that $\mathbf{A} - \mathbf{B} \mathbf{K}^*$ (resp. $\mathbf{A} - \mathbf{\Sigma}^{\sf w} \mathbf{K}^{\sf w}$) is Schur stable if, and only if, $\mathcal{S}^*$ (resp. $\mathcal{S}^{\sf w}$) is Schur stable.
Hence, the result follows from the fact $\mathcal{S}^* = \mathcal{S}^{\sf w}$.
\end{proof}

By applying \eqref{Eq:PROOF:AuxModel_Stabilizability_Detectability:partial_relations:K} repeatedly, we have $(\mathbf{A} - \mathbf{B} \mathbf{K}^*)^t \Phi_\auxiliary = \Phi_\auxiliary (\widetilde{A} - \widetilde{B} \widetilde{K}^*)^t$ for all $t \in \natural$.
Combining this relation with the fact $(\widetilde{A} - \widetilde{B} \widetilde{K}^*)^t \to 0$ as $t \to \infty$ via Schur stability in Claim \ref{CLAIM:PROOF:AuxModel_Stabilizability_Detectability:tilde_stable}, we have \eqref{Eq:PROOF:AuxModel_Stabilizability_Detectability:bold_K_stable_condition}, which implies Schur stability of $\mathbf{A} - \mathbf{B} \mathbf{K}^*$ through Claim \ref{CLAIM:PROOF:AuxModel_Stabilizability_Detectability:bold_K_stable}. 
Given Claim \ref{CLAIM:PROOF:AuxModel_Stabilizability_Detectability:bold_Kw_stable}, both $\mathbf{A} - \mathbf{B} \mathbf{K}^*$ and $\mathbf{A} - \mathbf{\Sigma}^{\sf w} \mathbf{K}^{\sf w}$ are Schur stable, which indicates that both pairs $(\mathbf{A}, \mathbf{B})$ and $(\mathbf{A}, \mathbf{\Sigma}^{\sf w})$ are stabilizable.
\end{proof}

\section{Proof of Proposition \ref{PROPOSITION:equivalence_of_optimization_problems}} \label{APPENDIX:PROOF:equivalence_of_optimization_problems}
\setcounter{proposition}{\getrefnumber{PROPOSITION:equivalence_of_optimization_problems}}
\setcounter{claim}{0}

We present preliminary results in Subsection A and prove Proposition \ref{PROPOSITION:equivalence_of_optimization_problems} in Subsection B.  

\subsection{Preliminary Results}

We begin by establishing useful identities in Claim \ref{CLAIM:OrigModel_AuxModel_relation} that will be leveraged in the remainder of the proof.
Recall the matrices $\Phi_\original \in \real^{n \times n_\xi}$, $\Phi_\auxiliary \in \real^{n_\auxiliary \times n_\xi}$ defined in Claim \ref{CLAIM:x_xi_relation}\Versions{}{ and matrix $\Phi = [\Phi_{\sf U}, \Phi_{\sf Y}, \Phi_{\sf P}] \in \real^{n\times n_\auxiliary}$ defined in Claim \ref{CLAIM:Phi_relation}}, with $n_\auxiliary := mL + pL + pL^2$ and $n_\xi := mL+n(L+1)$. 

\begin{claim} \label{CLAIM:OrigModel_AuxModel_relation}
    For the system \eqref{Eq:LTI} and auxiliary model \eqref{Eq:DDModel}, it holds for all $t \in \natural_{\geq 0}$ that
\begin{align*}&
    x_t = \Phi \mathbf{x}_t
    &&
    A \Phi \Phi_\auxiliary = \Phi \mathbf{A} \Phi_\auxiliary
    &&
    B = \Phi \mathbf{B}
    \\&
    w_t = \Phi \mathbf{w}_t
    &&
    C \Phi \Phi_\auxiliary = \mathbf{C} \Phi_\auxiliary
    &&
    \Sigma^{\sf w} = \Phi \mathbf{\Sigma}^{\sf w} \Phi^\transpose.
\end{align*}
\end{claim}
\begin{proof}
    The relation $x_t = \Phi \mathbf{x}_t$ follows from Claim \ref{CLAIM:x_xi_relation} and Claim \ref{CLAIM:Phi_relation}. 
    We have $C \Phi \Phi_\auxiliary = \mathbf{C} \Phi_\auxiliary$ from Claim \ref{CLAIM:c_relation}.
    To show $w_t = \Phi \mathbf{w}_t$ and $\Sigma^{\sf w} = \Phi \mathbf{\Sigma}^{\sf w} \Phi^\transpose$, recall from the definition that $\mathbf{w}_t = J_0 w_t$ and $\mathbf{\Sigma}^{\sf w} = J_0 \Sigma^{\sf w} J_0^\transpose$ where $J_0 := \col(0_{(n_\auxiliary - pL) \times n}, \mathcal{O})$.
    By direct calculation one can verify that $\Phi J_0 = I_n$, using which we obtain $w_t = \Phi \mathbf{w}_t$ given $\mathbf{w}_t = J_0 w_t$ and obtain $\Sigma^{\sf w} = \Phi \mathbf{\Sigma}^{\sf w} \Phi^\transpose$ given $\mathbf{\Sigma}^{\sf w} = J_0 \Sigma^{\sf w} J_0^\transpose$. 
    We have $B = \Phi_\original \widetilde{B} = \Phi \Phi_\auxiliary \widetilde{B} = \Phi \mathbf{B}$, using $\Phi_\original = \Phi \Phi_\auxiliary$ as Claim \ref{CLAIM:Phi_relation}, $\Phi_\auxiliary \widetilde{B} = \mathbf{B}$ from Claim \ref{CLAIM:AuxModel_TildeModel_relation} and $B = \Phi_\original \widetilde{B}$ which can be verified by definitions of $\Phi_\original$ and $\widetilde{B}$.
    We finally have $A \Phi \Phi_\auxiliary = A \Phi_\original = \Phi_\original \widetilde{A} = \Phi \Phi_\auxiliary \widetilde{A} = \Phi \mathbf{A} \Phi_\auxiliary$, where we used $\Phi_\original= \Phi \Phi_\auxiliary$ as Claim \ref{CLAIM:Phi_relation}, $\mathbf{A} \Phi_\auxiliary = \Phi_\auxiliary \widetilde{A}$ in Claim \ref{CLAIM:AuxModel_TildeModel_relation} and $A \Phi_\original = \Phi_\original \widetilde{A}$ which can be verified given the definitions of $\Phi_\original$ and $\widetilde{A}$.
\renewcommand\qedsymbol{$\blacklozenge$}
\end{proof}

Next, we relate the LQR feedback gains $K$ and $\mathbf{K}$. 

\begin{claim} \label{CLAIM:equivalence_of_optimization_problems:feedback_gain}
    For matrices $K$ in \eqref{Eq:feedback_gain} and $\mathbf{K}$ in \eqref{Eq:DDModel:feedback_gain}, it holds that $K \Phi \Phi_\auxiliary = \mathbf{K} \Phi_\auxiliary$.
\end{claim}
{\tb
\begin{proof}
\renewcommand\qedsymbol{$\blacklozenge$}
    Let $\widetilde C := C \Phi_\original$ and let $\widetilde A, \widetilde B$ be as in \eqref{Eq:TildeMatrices_Definition}. We first show the pair $(\widetilde A, \widetilde C)$ is detectable.
    For $\lambda \in \complex$, define $H_{\sf obs} := \col(\lambda I_{n_\xi} - \widetilde A, \widetilde C)$, which can be permuted into the form
\begin{align}
\label{Eq:PROOF:CLAIM:prerequisite_3:Hautus_matrix}
    \left[\compact{ \begin{array}{c|c|c}
        \lambda I_{mL} - \mathcal{D}_m && \\ \hline & \lambda I_{nL} - \mathcal{D}_n & \\ \hline \begin{matrix} -B & 0_{n\times m(L-1)}\end{matrix} & \begin{matrix} -I_n & 0_{n\times n(L-1)} \end{matrix} & \lambda I_n - A \\ \hline C \mathcal{C} & C \mathcal{C}_{\sf w} & C A^L
    \end{array} }\right],
\end{align}
    wherein $\mathcal{D}_q := \smallmat{& I_{q(L-1)} \\ 0_{q\times q}}$.
    Since the blocks $\lambda I_{mL} - \mathcal{D}_m$ and $\lambda I_{nL} - \mathcal{D}_n$ in \eqref{Eq:PROOF:CLAIM:prerequisite_3:Hautus_matrix} are non-singular for all $\lambda \neq 0$, to show that \eqref{Eq:PROOF:CLAIM:prerequisite_3:Hautus_matrix} has full column rank when $|\lambda| \geq 1$, we only need to verify the rank of the last block column in \eqref{Eq:PROOF:CLAIM:prerequisite_3:Hautus_matrix}.
    Since $(A, C)$ is observable, $\mathcal{O}_n := \col(C, CA, \ldots, CA^{n-1})$ has full column rank, so we have $\mathrm{Null}(\mathcal{O}_n A^L) = \mathrm{Null}(A^L)$ where $\mathrm{Null}$ denotes the null space.
    Note that $\mathcal{O}_n A^L$ is the observability matrix of the pair $(A, C A^L)$, and thus $\mathrm{Null}(\mathcal{O}_n A^L)$ is the unobservable space of the pair $(A, C A^L)$.
    Given $\mathrm{Null}(\mathcal{O}_n A^L) = \mathrm{Null}(A^L)$, all unobservable states $x_{\sf nobs}$ of $(A, C A^L)$ satisfy $A^L x_{\sf nobs} = 0$ and hence are strictly stable, which implies that $(A, C A^L)$ is detectable.
    From the Hautus lemma, $\col(\lambda I_n - A, C A^L)$ has full column rank for all $\lambda$ that $|\lambda| \geq 1$.
    With diagonal blocks $\lambda I_{mL} - \mathcal{D}_m$, $\lambda I_{nL} - \mathcal{D}_n$ and $\col(\lambda I_n - A, C A^L)$ having full column rank, the matrix \eqref{Eq:PROOF:CLAIM:prerequisite_3:Hautus_matrix} has full column rank when $|\lambda| \geq 1$, and so does the pre-permutational matrix $H_{\sf obs}$, which implies that $(\widetilde A, \widetilde C)$ is detectable through Hautus lemma.
    
    Next, we show that $\widetilde P_1 = \widetilde P_2$ with $\widetilde P_1 := \Phi_\original^\transpose P \Phi_\original$ and $\widetilde P_2 := \Phi_\auxiliary^\transpose \mathbf{P} \Phi_\auxiliary$,
    where $P$ is the solution to \eqref{Eq:feedback_DARE} and $\mathbf{P}$ the solution to \eqref{Eq:DDModel:feedback_DARE}.
    By left- and right-multiplying \eqref{Eq:feedback_DARE} by $\Phi$ and $\Phi^\transpose$ respectively, the resulting equation can be written as
\begin{align} \label{Eq:PROOF:CLAIM:prerequisite_3:DARE:tilde_1}
    \widetilde P_1 = \widetilde A\,^\transpose \widetilde P_1 (\widetilde A - \widetilde B \widetilde K_1) + \widetilde C\,^\transpose Q (\widetilde C - D \widetilde K_1)
\end{align}
    wherein $\widetilde K_1 := (R + D^\transpose Q D + \widetilde B\,^\transpose \widetilde P_1 \widetilde B)^{-1} (\widetilde B\,^\transpose \widetilde P_1 \widetilde A + D^\transpose Q \widetilde C)$, 
    provided the definitions $\widetilde C := C \Phi_\original$, $\widetilde P_1 := \Phi_\original^\transpose P \Phi_\original$ and the relations $A \Phi_\original = \Phi_\original \widetilde A$, $B = \Phi_\original \widetilde B$ implied by Claim \ref{CLAIM:Phi_relation}, Claim \ref{CLAIM:AuxModel_TildeModel_relation} and Claim \ref{CLAIM:OrigModel_AuxModel_relation}.
    Similarly, by left- and right-multiplying \eqref{Eq:DDModel:feedback_DARE} by $\Phi_\auxiliary^\transpose$ and $\Phi_\auxiliary$ respectively, the resulting equation can be written in the form
\begin{align} \label{Eq:PROOF:CLAIM:prerequisite_3:DARE:tilde_2}
    \widetilde P_2 = \widetilde A\,^\transpose \widetilde P_2 (\widetilde A - \widetilde B \widetilde K_2) + \widetilde C\,^\transpose Q (\widetilde C - D \widetilde K_2)
\end{align}
    with $\widetilde K_2 := (R + D^\transpose Q D + \widetilde B\,^\transpose \widetilde P_2 \widetilde B)^{-1} (\widetilde B\,^\transpose \widetilde P_2 \widetilde A + D^\transpose Q \widetilde C)$,
    given the definitions $\widetilde C := C \Phi_\original$, $\widetilde P_2 := \Phi_\auxiliary^\transpose \mathbf{P} \Phi_\auxiliary$ and the relations $\mathbf{A} \Phi_\auxiliary = \Phi_\auxiliary \widetilde A$, $\mathbf{B} = \Phi_\auxiliary \widetilde B$, $\mathbf{C} \Phi_\auxiliary = C \Phi_\original$ according to Claim \ref{CLAIM:Phi_relation} and Claim \ref{CLAIM:OrigModel_AuxModel_relation}. 
    Observing \eqref{Eq:PROOF:CLAIM:prerequisite_3:DARE:tilde_1} and \eqref{Eq:PROOF:CLAIM:prerequisite_3:DARE:tilde_2}, we know that both $\widetilde P_1$ and $\widetilde P_2$ are (positive semi-definite) solutions to a similar DARE to \eqref{Eq:feedback_DARE} and \eqref{Eq:DDModel:feedback_DARE}, for dynamical system $(\widetilde A, \widetilde B, \widetilde C, D)$.
    In fact, this DARE has a unique positive semi-definite solution, given stabilizable $(\widetilde A, \widetilde B)$ via Claim \ref{CLAIM:PROOF:AuxModel_Stabilizability_Detectability:tilde_stable}, detectable $(\widetilde A, \widetilde C)$ as proved before and $Q \succ 0$.
    Hence, the solutions $\widetilde P_1, \widetilde P_2$ are equal.

    Finally, we obtain the result by noting the relations $B^\transpose P A \Phi_\original = \mathbf{B}^\transpose \mathbf{P} \mathbf{A} \Phi_\auxiliary$ and $B^\transpose P B = \mathbf{B}^\transpose \mathbf{P} \mathbf{B}$, which can be verified given $\Phi_\original^\transpose P \Phi_\original = \Phi_\auxiliary^\transpose \mathbf{P} \Phi_\auxiliary$ (as $\widetilde P_1 = \widetilde P_2$) and given Claim \ref{CLAIM:Phi_relation}, Claim \ref{CLAIM:AuxModel_TildeModel_relation} and Claim \ref{CLAIM:OrigModel_AuxModel_relation}.
    It follows from the definitions \eqref{Eq:feedback_gain}, \eqref{Eq:DDModel:feedback_gain} of $K$ and $\mathbf{K}$ that $K \Phi_\original = \mathbf{K} \Phi_\auxiliary$, which is the result given $\Phi_\original = \Phi \Phi_\auxiliary$ (Claim \ref{CLAIM:Phi_relation}).
\end{proof}
}

We mention in Claim \ref{CLAIM:equivalence_of_optimization_problems:corollaries} some identities which will be used multiple times in the rest of the proof.

\begin{claim} \label{CLAIM:equivalence_of_optimization_problems:corollaries}
    If $v \in \real^{n}$, $\boldsymbol{v} \in \real^{n_\auxiliary}$ and $\tilde{v} \in \real^{n_\xi}$ are such that $v = \Phi \boldsymbol{v}$ and $\boldsymbol{v} = \Phi_\auxiliary \tilde{v}$, then
\begin{align*}
    C v &= \mathbf{C} \boldsymbol{v}, &
    K v &= \mathbf{K} \boldsymbol{v}, &
    A v &= \Phi \mathbf{A} \boldsymbol{v}, &
    \mathbf{A} \boldsymbol{v} &= \Phi_\auxiliary \widetilde{A} \tilde{v}.
\end{align*}
    If $M \in \symmetric^n_+$, $\mathbf{M} \in \symmetric^{n_\auxiliary}_+$ and $\widetilde{M} \in \symmetric^{n_\xi}_+$ are such that $M = \Phi \mathbf{M} \Phi^\transpose$ and $\mathbf{M} = \Phi_\auxiliary \widetilde{M} \Phi_\auxiliary^\transpose$, then
\begin{align*} 
    & C M = \mathbf{C} \mathbf{M} \Phi^\transpose, \;\;
    C M C^\transpose = \mathbf{C} \mathbf{M} \mathbf{C}^\transpose, \;\;
    C M K^\transpose = \mathbf{C} \mathbf{M} \mathbf{K}^\transpose, \\
    & K M = \mathbf{K} \mathbf{M} \Phi^\transpose, \;\;
    K M K^\transpose = \mathbf{K} \mathbf{M} \mathbf{K}^\transpose.
\end{align*}
\end{claim}
\begin{proof}
\renewcommand\qedsymbol{$\blacklozenge$}
    Using $C \Phi \Phi_\auxiliary = \mathbf{C} \Phi_\auxiliary$ (Claim \ref{CLAIM:OrigModel_AuxModel_relation}), we have $C v = C \Phi \boldsymbol{v} = C \Phi \Phi_\auxiliary \tilde{v} = \mathbf{C} \Phi_\auxiliary \tilde{v} = \mathbf{C} \boldsymbol{v}$, and one can show $C M = \mathbf{C} \mathbf{M} \Phi^\transpose$ and $C M C^\transpose = \mathbf{C} \mathbf{M} \mathbf{C}^\transpose$ given the facts $M = \Phi \mathbf{M} \Phi^\transpose = \Phi \Phi_\auxiliary \widetilde{M} \Phi_\auxiliary^\transpose \Phi^\transpose$ and $\mathbf{M} = \Phi_\auxiliary \widetilde{M} \Phi_\auxiliary^\transpose$.
    Similarly, using $K \Phi \Phi_\auxiliary = \mathbf{K} \Phi_\auxiliary$ (Claim \ref{CLAIM:equivalence_of_optimization_problems:feedback_gain}), we prove $K v = \mathbf{K} \boldsymbol{v}$, $K M = \mathbf{K} \mathbf{M} \Phi^\transpose$, $K M K^\transpose = \mathbf{K} \mathbf{M} \mathbf{K}^\transpose$ and $C M K^\transpose = \mathbf{C} \mathbf{M} \mathbf{K}^\transpose$ in the same way by replacing $(C, \mathbf{C})$ into $(K, \mathbf{K})$.
    Using $A \Phi \Phi_\auxiliary = \Phi \mathbf{A} \Phi_\auxiliary$ (Claim \ref{CLAIM:OrigModel_AuxModel_relation}) and $\mathbf{A} \Phi_\auxiliary = \Phi_\auxiliary \widetilde{A}$ (Claim \ref{CLAIM:AuxModel_TildeModel_relation}), we show that $A v = A \Phi \boldsymbol{v} = A \Phi \Phi_\auxiliary \tilde{v} = \Phi \mathbf{A} \Phi_\auxiliary \tilde{v} = \Phi \mathbf{A} \boldsymbol{v}$ and also $\mathbf{A} \boldsymbol{v} = \mathbf{A} \Phi_\auxiliary \tilde{v} = \Phi_\auxiliary \widetilde{A} \tilde{v}$.
\end{proof}

In the following claim, the state variances $\Sigma^{\sf x}, \mathbf{\Sigma}^{\sf x}$, Kalman gains $L_{\sf K}, \mathbf{L}_{\sf K}$ and Luenberger gains $L_{\sf L}, \mathbf{L}_{\sf L}$ are related.

\begin{claim} \label{CLAIM:equivalence_of_optimization_problems:Kalman_gain}
    For matrices $\Sigma^{\sf x}, L_{\sf K}, L_{\sf L}$ in \eqref{Eq:state_variance} and $\mathbf{\Sigma}^{\sf x}, \mathbf{L}_{\sf K}, \mathbf{L}_{\sf L}$ in \eqref{Eq:DDModel:state_variance}, it holds that 
\begin{enumerate}[(a)]
    \item $\Sigma^{\sf x} = \Phi \mathbf{\Sigma}^{\sf x} \Phi^\transpose$ and $\mathbf{\Sigma}^{\sf x} = \Phi_\auxiliary \widetilde{\Sigma}^{\sf x} \Phi_\auxiliary^\transpose$ for some $\widetilde{\Sigma}^{\sf x} \in \symmetric^{n_\xi}_+$;
    \item $L_{\sf K} = \Phi \mathbf{L}_{\sf K}$ and $\mathbf{L}_{\sf K} = \Phi_\auxiliary \widetilde{L}_{\sf K}$ for some $\widetilde{L}_{\sf K} \in \real^{n_\xi \times p}$;
    \item $L_{\sf L} = \Phi \mathbf{L}_{\sf L}$ and $\mathbf{L}_{\sf L} = \Phi_\auxiliary \widetilde{L}_{\sf L}$ for some $\widetilde{L}_{\sf L} \in \real^{n_\xi \times p}$.
\end{enumerate}
\end{claim}
\begin{proof}
{\tb
    We first show $\mathbf{\Sigma}^{\sf x} = \mathbf{\Sigma}^\prime := \Phi_\auxiliary \widetilde{\Sigma}^{\sf x} \Phi_\auxiliary^\transpose$ in (a). Let $\widetilde C := C \Phi_\original$ and let $\widetilde A, \widetilde B$ be as in Claim \ref{CLAIM:AuxModel_TildeModel_relation}.
    Since $(\widetilde A, \widetilde \Sigma^{\sf w})$ is stabilizable through Claim \ref{CLAIM:PROOF:AuxModel_Stabilizability_Detectability:tilde_stable} and $(\widetilde A, \widetilde C)$ is detectable as shown in the proof of Claim \ref{CLAIM:equivalence_of_optimization_problems:feedback_gain}, the DARE 
\begin{align} \label{Eq:CLAIM:equivalence_of_optimization_problems:Kalman_gain:relation_1}
    \widetilde{\Sigma}^{\sf x} = \widetilde{A} \widetilde{\Sigma}^{\sf x} \!\widetilde{A}\,^\transpose \!+\! \widetilde{\Sigma}^{\sf w} \!-\! \widetilde{A} \widetilde{\Sigma}^{\sf x} \widetilde{C}^\transpose (\widetilde{C} \widetilde{\Sigma}^{\sf x} \widetilde{C}^\transpose \!+\! \Sigma^{\sf v})^{-1} \widetilde{C} \widetilde{\Sigma}^{\sf x} \!\widetilde{A}\,^\transpose
\end{align}
    has a unique positive semi-definite solution $\widetilde{\Sigma}^{\sf x}$.
    Left- and right-multiply \eqref{Eq:CLAIM:equivalence_of_optimization_problems:Kalman_gain:relation_1} by $\Phi_\auxiliary$ and by $\Phi_\auxiliary^\transpose$ respectively, and the resulting equation can be written in the form
\begin{align} 
\label{Eq:CLAIM:equivalence_of_optimization_problems:Kalman_gain:relation_2}
    \mathbf{\Sigma}^\prime \!=\! \mathbf{A}  \mathbf{\Sigma}^\prime \mathbf{A}\!^\transpose  \!+\! \mathbf{\Sigma}^{\sf w} \!-\! \mathbf{A} \mathbf{\Sigma}^\prime \mathbf{C}^\transpose (\mathbf{C} \mathbf{\Sigma}^{\prime} \mathbf{C}^\transpose \!+\! \Sigma^{\sf v})^{-1} \mathbf{C} \mathbf{\Sigma}^\prime \widetilde{A}\,^\transpose
\end{align}
    with substitutions $\Phi_\auxiliary \widetilde{A} = \mathbf{A} \Phi_\auxiliary$, $\mathbf{\Sigma}^{\sf w} = \Phi_\auxiliary \widetilde{\Sigma}^{\sf w} \Phi_\auxiliary^\transpose$ and $\widetilde{C} = C \Phi_\original = C \Phi \Phi_\auxiliary = \mathbf{C} \Phi_\auxiliary$ via Claim \ref{CLAIM:Phi_relation} and Claim \ref{CLAIM:OrigModel_AuxModel_relation}.
    Due to \eqref{Eq:CLAIM:equivalence_of_optimization_problems:Kalman_gain:relation_2}, $\mathbf{\Sigma}^\prime$ is a positive semi-definite solution to the DARE \eqref{Eq:DDModel:state_variance:DARE}.
    Since \eqref{Eq:DDModel:state_variance:DARE} has a unique positive semi-definite solution $\mathbf{\Sigma}^{\sf x}$, we have $\mathbf{\Sigma}^{\sf x} = \mathbf{\Sigma}^\prime$.

    Next, we show $\Sigma^{\sf x} = \Sigma^\prime := \Phi \mathbf{\Sigma}^{\sf x} \Phi^\transpose$ in (a).
    Left- and right-multiply \eqref{Eq:DDModel:state_variance:DARE} by $\Phi$ and by $\Phi^\transpose$ respectively, and the resulting equality can be written as
\begin{align} \label{Eq:CLAIM:equivalence_of_optimization_problems:Kalman_gain:relation_3} \begin{aligned}
    \Sigma^\prime \!=\! A \Sigma^\prime A^\transpose \!+\! \Sigma^{\sf w} \!-\! A \Sigma^\prime C^\transpose (C \Sigma^\prime C^\transpose \!+\! \Sigma^{\sf v})^{-1} C \Sigma^\prime A^\transpose
\end{aligned} \end{align}
    given $\Phi \mathbf{\Sigma}^{\sf w} \Phi^\transpose = \Sigma^{\sf w}$ (Claim \ref{CLAIM:OrigModel_AuxModel_relation}) and the substitutions $\Phi \mathbf{A} \mathbf{\Sigma}^{\sf x} = A \Phi \mathbf{\Sigma}^{\sf x}$ and $\mathbf{C} \mathbf{\Sigma}^{\sf x} = C \Phi \mathbf{\Sigma}^{\sf x}$, which are implied by $\Phi \mathbf{A} \Phi_\auxiliary = A \Phi \Phi_\auxiliary$ and $\mathbf{C} \Phi_\auxiliary = C \Phi \Phi_\auxiliary$ (Claim \ref{CLAIM:OrigModel_AuxModel_relation}) respectively, provided $\mathbf{\Sigma}^{\sf x} = \Phi_\auxiliary \widetilde{\Sigma}^{\sf x} \Phi_\auxiliary^\transpose$.
    Due to \eqref{Eq:CLAIM:equivalence_of_optimization_problems:Kalman_gain:relation_3}, $\Sigma^\prime$ is a positive semi-definite solution to the DARE \eqref{Eq:state_variance:DARE}. Since \eqref{Eq:state_variance:DARE} has a unique positive definite solution $\Sigma^{\sf x}$, we have $\Sigma^{\sf x} = \Sigma^\prime$.
}

    We finally show (b) and (c). Given the definitions \eqref{Eq:DDModel:state_variance:gain} of $\mathbf{L}_{\sf K}$ and $\mathbf{L}_{\sf L}$, we obtain $\mathbf{L}_{\sf K} = \Phi_\auxiliary \widetilde{L}_{\sf K}$ and $\mathbf{L}_{\sf L} = \Phi_\auxiliary \widetilde{L}_{\sf L}$ as
\begin{align*}
    \mathbf{L}_{\sf K} &:= \mathbf{\Sigma}^{\sf x} \mathbf{C}^\transpose (\mathbf{C} \mathbf{\Sigma}^{\sf x} \mathbf{C}^\transpose + \Sigma^{\sf v})^{-1} \\
    &= \Phi_\auxiliary \widetilde{\Sigma}^{\sf x} \Phi_\auxiliary^\transpose \mathbf{C}^\transpose (\mathbf{C} \mathbf{\Sigma}^{\sf x} \mathbf{C}^\transpose + \Sigma^{\sf v})^{-1} = \Phi_\auxiliary \widetilde{L}_{\sf K} \\
    \mathbf{L}_{\sf L} &:= \mathbf{A} \mathbf{L}_{\sf K} = \mathbf{A} \Phi_\auxiliary \widetilde{L}_{\sf K} = \Phi_\auxiliary \widetilde{A} \widetilde{L}_{\sf K} = \Phi_\auxiliary \widetilde{L}_{\sf L}
\end{align*}
    with $\widetilde{L}_{\sf K} := \widetilde{\Sigma}^{\sf x} \Phi_\auxiliary^\transpose \mathbf{C}^\transpose (\mathbf{C} \mathbf{\Sigma}^{\sf x} \mathbf{C}^\transpose \!\!+\! \Sigma^{\sf v})^{-1}$ and $\widetilde{L}_{\sf L} := \widetilde{A} \widetilde{L}_{\sf K}$, where we used $\mathbf{\Sigma}^{\sf x} = \Phi_\auxiliary \widetilde{\Sigma}^{\sf x} \Phi_\auxiliary^\transpose$ in (a) and $\mathbf{A} \Phi_\auxiliary = \Phi_\auxiliary \widetilde{A}$ in Claim \ref{CLAIM:AuxModel_TildeModel_relation}.
    With definitions \eqref{Eq:state_variance:gain}, \eqref{Eq:DDModel:state_variance:gain} of $L_{\sf K}, L_{\sf L}, \mathbf{L}_{\sf K}, \mathbf{L}_{\sf L}$, we have
\begin{align*}
    L_{\sf K} &:= \Sigma^{\sf x} C^\transpose (C \Sigma^{\sf x} C^\transpose + \Sigma^{\sf v})^{-1} \\
    &= \Phi \mathbf{\Sigma}^{\sf x} \mathbf{C}^\transpose (\mathbf{C} \mathbf{\Sigma}^{\sf x} \mathbf{C}^\transpose + \Sigma^{\sf v})^{-1} = \Phi \mathbf{L}_{\sf K} \\
    L_{\sf L} &:= A L_{\sf K} = \Phi \mathbf{A} \mathbf{L}_{\sf K} = \Phi \mathbf{L}_{\sf L},
\end{align*}
    where we used $C \Sigma^{\sf x} = \mathbf{C} \mathbf{\Sigma}^{\sf x} \Phi^\transpose$ and $C \Sigma^{\sf x} C^\transpose = \mathbf{C} \mathbf{\Sigma}^{\sf x} \mathbf{C}^\transpose$ through Claim \ref{CLAIM:equivalence_of_optimization_problems:corollaries} with selection $(M, \mathbf{M}, \widetilde{M}) \gets (\Sigma^{\sf x}, \mathbf{\Sigma}^{\sf x}, \widetilde{\Sigma}^{\sf x})$ given (a), and used $A L_{\sf K} = \Phi \mathbf{A} \mathbf{L}_{\sf K}$ implied by $A v = \Phi \mathbf{A} \boldsymbol{v}$ from Claim \ref{CLAIM:equivalence_of_optimization_problems:corollaries} where $v, \boldsymbol{v}, \tilde{v}$ are chosen as the $i$-th columns of $L_{\sf K}, \mathbf{L}_{\sf K}, \widetilde{L}_{\sf K}$, respectively, for $i \in \{1,\ldots,p\}$.
\renewcommand\qedsymbol{$\blacklozenge$}
\end{proof} 

\subsection{Proof of Proposition \ref{PROPOSITION:equivalence_of_optimization_problems}}

\begin{proof}
    We first show in Claim \ref{CLAIM:equivalence_of_optimization_problems:uy_variance_relation} that the matrices $\Delta_s$ and $\mathbf{\Delta}_s$ are identical, and then in Claim \ref{CLAIM:equivalence_of_optimization_problems:nominal_state}(c) that the nominal outputs $\thickbar y_t$ and $\thickbar{\mathbf{y}}_t$ are equal.

\begin{claim}
\label{CLAIM:equivalence_of_optimization_problems:uy_variance_relation}
    For matrices $\Delta_s$ in \eqref{Eq:input_output_variance:Delta} and $\mathbf{\Delta}_s$ in \eqref{Eq:DDModel:input_output_variance:Delta}, we have $\Delta_s = \mathbf{\Delta}_s$ for $s \in \integer_{[0,N)}$.
\end{claim}
\begin{proof}
\renewcommand\qedsymbol{$\blacklozenge$}
    We first show $\Lambda_s = \Phi \mathbf{\Lambda}_s \Phi^\transpose$ and $\mathbf{\Lambda}_s = \Phi_\auxiliary \widetilde{\Lambda}_s \Phi_\auxiliary^\transpose$ for $s \in \integer_{[0,N)}$, where $\Lambda_s$ and $\mathbf{\Lambda}_s$ are defined in \eqref{Eq:input_output_variance:Lambda} and \eqref{Eq:DDModel:input_output_variance:Lambda}, and let $\widetilde{\Lambda}_s := \textstyle{\sum_{r=0}^s} (\widetilde{A} - \widetilde{B} \widetilde{K})^r \,\widetilde{L}_{\sf L} (\mathbf{C} \mathbf{\Sigma}^{\sf x} \mathbf{C}^\transpose \!+\! \Sigma^{\sf v}) \widetilde{L}_{\sf L} (\widetilde{A} - \widetilde{B} \widetilde{K})^{r \transpose}$, with $\widetilde{L}_{\sf L}$ in Claim \ref{CLAIM:equivalence_of_optimization_problems:Kalman_gain} and $\widetilde{K} := \mathbf{K} \Phi_\auxiliary$.
    Given Claim \ref{CLAIM:equivalence_of_optimization_problems:Kalman_gain}(c), the definitions of $\Lambda_s, \mathbf{\Lambda}_s$, and the identity $C \Sigma^{\sf x} C^\transpose = \mathbf{C} \mathbf{\Sigma}^{\sf x} \mathbf{C}^\transpose$ shown in the proof Claim \ref{CLAIM:equivalence_of_optimization_problems:Kalman_gain}(b-c), it suffices to show
\begin{align*}
    & (A - B K)^r \, \Phi \Phi_\auxiliary = \Phi \, (\mathbf{A} - \mathbf{B} \mathbf{K})^r \, \Phi_\auxiliary \\
    & (\mathbf{A} - \mathbf{B} \mathbf{K})^r \, \Phi_\auxiliary = \Phi_\auxiliary \, (\widetilde{A} - \widetilde{B} \widetilde{K})^r
\end{align*}
    for all $r \in \natural_{\geq 0}$, which can be obtained by repeatedly applying $(A - B K) \Phi \Phi_\auxiliary = \Phi (\mathbf{A} - \mathbf{B} \mathbf{K}) \Phi_\auxiliary$ and $(\mathbf{A} - \mathbf{B} \mathbf{K}) \Phi_\auxiliary = \Phi_\auxiliary (\widetilde{A} - \widetilde{B} \widetilde{K})$ respectively, which follow from $\mathbf{A} = \Phi_\auxiliary \widetilde{A}$ and $\mathbf{B} = \Phi_\auxiliary \widetilde{B}$ in Claim \ref{CLAIM:AuxModel_TildeModel_relation}, $A \Phi \Phi_\auxiliary = \Phi \mathbf{A} \Phi_\auxiliary$ and $B = \Phi \mathbf{B}$ in Claim \ref{CLAIM:OrigModel_AuxModel_relation}, and $K \Phi \Phi_\auxiliary = \mathbf{K} \Phi_\auxiliary$ in Claim \ref{CLAIM:equivalence_of_optimization_problems:feedback_gain}.

    We finally show $\Delta_s = \mathbf{\Delta}_s$ for $s \in \integer_{[0,N)}$. Given the definitions of $\Delta_s$ and $\mathbf{\Delta}_s$ and the relation $C \Sigma^{\sf x} C^\transpose = \mathbf{C} \mathbf{\Sigma}^{\sf x} \mathbf{C}^\transpose$, it suffices to show the relations $C \Lambda_s C^\transpose = \mathbf{C} \mathbf{\Lambda}_s \mathbf{C}^\transpose$, $K \Lambda_s K^\transpose = \mathbf{K} \mathbf{\Lambda}_s \mathbf{K}^\transpose$ and $C \Lambda_s K^\transpose = \mathbf{C} \mathbf{\Lambda}_s \mathbf{K}^\transpose$, which are obtained through Claim \ref{CLAIM:equivalence_of_optimization_problems:corollaries} with selection $(M, \mathbf{M}, \widetilde{M}) \gets (\Lambda_s, \mathbf{\Lambda}_s, \widetilde{\Lambda}_s)$ given $\Lambda_s = \Phi \mathbf{\Lambda}_s \Phi^\transpose$ and $\mathbf{\Lambda}_s = \Phi_\auxiliary \widetilde{\Lambda}_s \Phi_\auxiliary^\transpose$ as proved.
\end{proof}

\begin{claim} \label{CLAIM:equivalence_of_optimization_problems:nominal_state}
    If parameters $\mu^{\sf \hat x}_k, \mu^{\sf \bar x}_k, \boldsymbol{\mu}^{\sf \bar x}_k, \boldsymbol{\mu}^{\sf \hat x}_k$ satisfy \eqref{Eq:initial_condition_relation} for some $\widetilde{\mu}^{\sf \,\hat x}_k, \widetilde{\mu}^{\sf \,\bar x}_k \in \real^{n_\xi}$, then, for all $\thickbar u_{[k,k+N)}$ and $\theta$, we have 
\begin{enumerate}[(a)]
    \item $\mu^{\sf x}_k = \Phi \boldsymbol{\mu}^{\sf x}_k$ and $\boldsymbol{\mu}^{\sf x}_k = \Phi_\auxiliary \widetilde{\mu}^{\sf \,x}_k$ with some $\widetilde{\mu}^{\sf \,x}_k \in \real^{n_\xi}$,
\end{enumerate}
    and, for all $t \in \integer_{[k,k+N)}$, we have
\begin{enumerate}[(a)]
\setcounter{enumi}{1}
    \item $\thickbar x_t = \Phi \, \thickbar{\mathbf{x}}_t$ and $\thickbar{\mathbf{x}}_t = \Phi_\auxiliary \, \widetilde{\thickbar{x}}_t$ with some $\widetilde{\thickbar{x}}_t \in \real^{n_\xi}$,
    \item  $\thickbar y_t = \thickbar{\mathbf{y}}_t$.
\end{enumerate}
\end{claim}
\begin{proof}
\renewcommand\qedsymbol{$\blacklozenge$}
    To prove (a), we obtain $\mu^{\sf x}_k = \Phi_\original \widetilde{\mu}^{\sf \,x}_k$ by combining \eqref{Eq:interpolating_initial_condition} and \eqref{Eq:initial_condition_relation}, and obtain $\boldsymbol{\mu}^{\sf x}_k = \Phi_\auxiliary \widetilde{\mu}^{\sf \,x}_k$ by combining \eqref{Eq:DDModel:interpolating_initial_condition} and \eqref{Eq:initial_condition_relation}, where we let $\widetilde{\mu}^{\sf \,x}_k := (1 - \theta) \widetilde{\mu}^{\sf \,\hat x}_k + \theta \widetilde{\mu}^{\sf \,\bar x}_k$.
    Then, $\mu^{\sf x}_k = \Phi \boldsymbol{\mu}^{\sf x}_k$ follows from $\Phi_\original = \Phi \Phi_\auxiliary$.

    (b) is proved by induction.
    \textbf{Base Case.} Select $\widetilde{\thickbar x}_k := \widetilde{\mu}^{\sf \,x}_k$. The $t=k$ case of (b) follows from (a) and relations $\thickbar x_k := \mu^{\sf x}_k$ as \eqref{Eq:nominal_model:initial} and $\thickbar{\mathbf{x}}_k := \boldsymbol{\mu}^{\sf x}_k$ as \eqref{Eq:DDModel:nominal_model:initial}.
    \textbf{Inductive Step.} Assume the $t = \tau$ case of (b) for some $\tau \in \integer_{[k,k+N-2]}$, and thus we have
\begin{align*}
    \thickbar x_{\tau+1} 
    \overset{\text{via \eqref{Eq:nominal_model:state}}}{=} A \thickbar x_\tau + B \thickbar u_\tau
    = \Phi \mathbf{A} \thickbar{\mathbf{x}}_\tau + \Phi \mathbf{B} \thickbar u_\tau
    \overset{\text{via \eqref{Eq:DDModel:nominal_model:state}}}{=} \Phi \thickbar{\mathbf{x}}_{\tau+1},
\end{align*}
    where the second equality used $B = \Phi \mathbf{B}$ in Claim \ref{CLAIM:OrigModel_AuxModel_relation} and $A \thickbar x_\tau = \Phi \mathbf{A} \thickbar{\mathbf{x}}_\tau$ through Claim \ref{CLAIM:equivalence_of_optimization_problems:corollaries} with selection $(v, \boldsymbol{v}, \tilde{v}) \gets (\thickbar{x}_\tau, \thickbar{\mathbf{x}}_\tau, \widetilde{\thickbar x}_\tau)$ given (b) of $t=\tau$.
    Moreover, we have
\begin{align*}
    \thickbar{\mathbf{x}}_{\tau+1} 
    \!\overset{\text{via \eqref{Eq:DDModel:nominal_model:state}}}{=}\! 
    \mathbf{A} \thickbar{\mathbf{x}}_\tau \!+\! \mathbf{B} \thickbar u_\tau
    = \Phi_\auxiliary \widetilde{A} \,\widetilde{\thickbar x}_\tau \!+\! \Phi_\auxiliary \widetilde{B} \thickbar u_\tau 
    = \Phi_\auxiliary \widetilde{\thickbar x}_{\tau+1}
\end{align*}
    by choosing $\widetilde{\thickbar x}_{\tau+1} := \widetilde{A} \,\widetilde{\thickbar x}_\tau + \widetilde{B} \thickbar u_\tau$,
    where the second equality used $\mathbf{B} = \Phi_\auxiliary \widetilde{B}$ in Claim \ref{CLAIM:AuxModel_TildeModel_relation} and $\mathbf{A} \thickbar{\mathbf{x}}_\tau = \Phi_\auxiliary \widetilde{A} \,\widetilde{\thickbar x}_\tau$ through Claim \ref{CLAIM:equivalence_of_optimization_problems:corollaries} with $(v, \boldsymbol{v}, \tilde{v}) \gets (\thickbar{x}_\tau, \thickbar{\mathbf{x}}_\tau, \widetilde{\thickbar x}_\tau)$ given (b) of $t=\tau$. Thus, we have the $t = \tau + 1$ case of (b). This shows (b).

    Last, we have (c) $\thickbar y_t \overset{\text{via \eqref{Eq:nominal_model:output}}}= C \thickbar x_t + D \thickbar u_t = \mathbf{C} \thickbar{\mathbf{x}}_t + D \thickbar u_t \overset{\text{via \eqref{Eq:DDModel:nominal_model:output}}}= \thickbar{\mathbf{y}}_t$ using $C \thickbar x_t = \mathbf{C} \thickbar{\mathbf{x}}_t$ through Claim \ref{CLAIM:equivalence_of_optimization_problems:corollaries} with selection $(v, \boldsymbol{v}, \tilde{v}) \gets (\thickbar{x}_t, \thickbar{\mathbf{x}}_t, \widetilde{\thickbar x}_t)$ given (b).
\end{proof}

    With Claim \ref{CLAIM:equivalence_of_optimization_problems:uy_variance_relation} and Claim \ref{CLAIM:equivalence_of_optimization_problems:nominal_state}(c), the objective functions of problems \eqref{Eq:SMPC_reduced} and \eqref{Eq:DDSMPC_reduced} are equal, and the constraint \eqref{Eq:safety_constraint_reduced} in problem \eqref{Eq:SMPC_reduced} and the constraint \eqref{Eq:DDModel:safety_constraint_reduced} in problem \eqref{Eq:DDSMPC_reduced} are equivalent.
    Thus, problems \eqref{Eq:SMPC_reduced} and \eqref{Eq:DDSMPC_reduced} have the same objective function and constraints, and the result follows.
\end{proof}

\section{Proof of Theorem \ref{PROPOSITION:equivalence_of_control_algorithms}}
\label{APPENDIX:PROOF:equivalence_of_control_algorithms}
\setcounter{proposition}{\getrefnumber{PROPOSITION:equivalence_of_control_algorithms}}
\setcounter{claim}{0}

\begin{proof}

Let $\{x^{\sf a}_t, u^{\sf a}_t, y^{\sf a}_t\}$ denote the trajectory produced by process a), and $\{x^{\sf b}_t, u^{\sf b}_t, y^{\sf b}_t\}$ the trajectory from process b). 
We have the following intermediate result.

\begin{claim} \label{CLAIM:PROPOSITION:equivalence_of_control_algorithms}
    Consider a control step $k = \kappa$ in both processes a) and b). Assume that
\begin{enumerate}[i)]
    \item the states $x^{\sf a}_\kappa, x^{\sf b}_\kappa$ in processes a) and b) are equal, and
    \item the parameters $\mu^{\sf \hat x}_\kappa, \mu^{\sf \bar x}_\kappa$ in process a) and parameters $\boldsymbol{\mu}^{\sf \hat x}_\kappa, \boldsymbol{\mu}^{\sf \bar x}_\kappa$ in process b) satisfy \eqref{Eq:initial_condition_relation} with $k = \kappa$.
\end{enumerate}
    Let $\kappa^\plus := \kappa +N_{\rm c}$. Then, for $t \in \integer_{[\kappa, \kappa^\plus]}$, we have
\begin{enumerate}[(a)]
    \item the states $x^{\sf a}_t, x^{\sf b}_t$ in processes a) and b) are equal,
    \item the variable $\hat x^\minus_t$ in process a) and variable $\hat{\mathbf{x}}^\minus_t$ in process b) satisfy $\hat x^\minus_t = \Phi \hat{\mathbf{x}}^\minus_t$ and $\hat{\mathbf{x}}^\minus_t = \Phi_\auxiliary \widetilde{\hat x^\minus_t}$ for some $\widetilde{\hat x^\minus_t} \in \real^{n_\xi}$,
\end{enumerate}
    and, for $t \in \integer_{[\kappa, \kappa^\plus)}$, we have
\begin{enumerate}[(a)]
\setcounter{enumi}{2}
    \item the inputs $u^{\sf a}_t, u^{\sf b}_t$ in processes a) and b) are equal,
    \item the outputs $y^{\sf a}_t, y^{\sf b}_t$ in processes a) and b) are equal.
\end{enumerate}
    Moreover, at the next control step $k = \kappa^\plus$, we have
\begin{enumerate}[(a)]
\setcounter{enumi}{4}
    \item the parameters $\mu^{\sf \hat x}_{\kappa^\plus}, \mu^{\sf \bar x}_{\kappa^\plus}$ in process a) and parameters $\boldsymbol{\mu}^{\sf \hat x}_{\kappa^\plus}, \boldsymbol{\mu}^{\sf \bar x}_{\kappa^\plus}$ in process b) satisfy \eqref{Eq:initial_condition_relation}  with $k = \kappa^\plus$.
\end{enumerate}
\end{claim}
\begin{proof}
\renewcommand\qedsymbol{$\blacklozenge$}
    We prove (a)-(d) by induction. 
    \textbf{Base Case:} we show (a) and (b) for $t=\kappa$. Result (a) of $t=\kappa$ is exactly as condition i).
    Through Proposition \ref{PROPOSITION:equivalence_of_optimization_problems} and the fact that both problems \eqref{Eq:SMPC_reduced} and \eqref{Eq:DDSMPC_reduced} produce unique optimal $\theta$, the values of $\theta$ are the same in processes a) and b).
    Given condition ii), $\mu^{\sf x}_\kappa$ in process a) and $\boldsymbol{\mu}^{\sf x}_\kappa$ in process b) satisfy $\mu^{\sf x}_\kappa = \Phi \boldsymbol{\mu}^{\sf x}_k$ and $\boldsymbol{\mu}^{\sf x}_\kappa = \Phi_\auxiliary \widetilde{\mu}^{\sf \,x}_\kappa$ for some $\widetilde{\mu}^{\sf \,x}_\kappa$ according to Claim \ref{CLAIM:equivalence_of_optimization_problems:nominal_state}.
    Combining these relations with $\hat x^\minus_\kappa := \mu^{\sf x}_\kappa$ as \eqref{Eq:Kalman_filter:initial} and $\hat{\mathbf{x}}^\minus_\kappa := \boldsymbol{\mu}^{\sf x}_\kappa$ as \eqref{Eq:DDModel:Kalman_filter:initial}, we obtain (b) of $t=\kappa$ by choosing $\widetilde{\hat x^\minus_\kappa} := \widetilde{\mu}^{\sf \,x}_\kappa$, as
\begin{align*}
    \hat x^\minus_\kappa &= \mu^{\sf x}_\kappa = \Phi \boldsymbol{\mu}^{\sf x}_\kappa = \Phi \hat{\mathbf{x}}^\minus_\kappa, &
    \hat{\mathbf{x}}^\minus_\kappa &= \boldsymbol{\mu}^{\sf x}_\kappa = \Phi_\auxiliary \widetilde{\mu}^{\sf \,x}_\kappa = \Phi_\auxiliary \widetilde{\hat x^\minus_\kappa}.
\end{align*}
    \textbf{Inductive Step:} we assume (a) and (b) for $t=\tau \in \integer_{[\kappa, \kappa^\plus)}$, and then prove (c), (d) for $t=\tau$ and (a), (b) for $t=\tau+1$.
    The control inputs $u^{\sf a}_\tau$, $u^{\sf b}_\tau$ are obtained through \eqref{Eq:feedback_policy} and \eqref{Eq:DDModel:feedback_policy} respectively, where the nominal inputs $\thickbar u_\tau$ are the same according to Proposition \ref{PROPOSITION:equivalence_of_optimization_problems} and the fact that both problems \eqref{Eq:SMPC_reduced}, \eqref{Eq:DDSMPC_reduced} produce a unique optimal $\thickbar u$, i.e.,
\begin{align*}
    u^{\sf a}_\tau =
    \thickbar u_\tau - K (\hat x_\tau - \thickbar x_\tau), \quad
    u^{\sf b}_\tau = \thickbar u_\tau - \mathbf{K} (\hat{\mathbf{x}}_\tau - \thickbar{\mathbf{x}}_\tau).
\end{align*}
    Thus, we have (c) $u^{\sf a}_\tau = u^{\sf b}_\tau$ of $t = \tau$, because of $K \hat x_\tau = \mathbf{K} \hat{\mathbf{x}}_\tau$ and $K \thickbar x_\tau = \mathbf{K} \thickbar{\mathbf{x}}_\tau$ through Claim \ref{CLAIM:equivalence_of_optimization_problems:corollaries} where we choose $(v, \boldsymbol{v}, \tilde{v})$ as $(\hat x^\minus_\tau, \hat{\mathbf{x}}^\minus_\tau, \widetilde{\hat x^\minus_\tau})$ and $(\thickbar{x}_\tau, \thickbar{\mathbf{x}}_\tau, \widetilde{\thickbar{x}}_\tau)$, given (b) of $t=\tau$ and Claim \ref{CLAIM:equivalence_of_optimization_problems:nominal_state}(b) of $t=\tau$.
    We then have (d) $y^{\sf a}_\tau = y^{\sf b}_\tau$ for $t = \tau$ and (a) $x^{\sf a}_{\tau+1} = x^{\sf b}_{\tau+1}$ for $t = \tau+1$, given the system model $y^{\sf z}_\tau = C x^{\sf z}_\tau + D u^{\sf z}_\tau + v_t$ as \eqref{Eq:LTI:output} and $x^{\sf z}_{\tau+1} = A x^{\sf z}_\tau + B u^{\sf z}_\tau + w_t$ as \eqref{Eq:LTI:state}, for ${\sf z} \in \{{\sf a}, {\sf b}\}$.
    Finally, we prove (b) for $t = \tau+1$ as
\begin{align*}
    & \hat x^\minus_{\tau+1} \overset{\text{via \eqref{Eq:Kalman_filter}}}{=}
    A \hat x^\minus_\tau + B u^{\sf a}_\tau + L_{\sf L}(y^{\sf a}_\tau - C \hat x^\minus_\tau)
    \\&
    = \Phi \mathbf{A} \hat{\mathbf{x}}_\tau + \Phi \mathbf{B} u^{\sf b}_\tau + \Phi \mathbf{L}_{\sf L}(y^{\sf b}_\tau - \mathbf{C} \hat{\mathbf{x}}^\minus_\tau)
    \overset{\text{via \eqref{Eq:DDModel:Kalman_filter}}}{=} \Phi \hat{\mathbf{x}}^\minus_{\tau+1} 
    \\&
    \hat{\mathbf{x}}^\minus_{\tau+1} \overset{\text{via \eqref{Eq:DDModel:Kalman_filter}}} =
    \mathbf{A} \hat{\mathbf{x}}_\tau + \mathbf{B} u^{\sf b}_\tau + \mathbf{L}_{\sf L} (y^{\sf b}_\tau - \mathbf{C} \hat{\mathbf{x}}^\minus_\tau)
    \\&
    = \Phi_\auxiliary \widetilde{A} \, \widetilde{\hat x_\tau} + \Phi_\auxiliary \widetilde{B} u^{\sf b}_\tau + \Phi_\auxiliary \widetilde{L}_{\sf L}(y^{\sf b}_\tau - \mathbf{C} \hat{\mathbf{x}}^\minus_\tau)
    = \Phi_\auxiliary \widetilde{\hat x^\minus_\tau}_{+1}
\end{align*}
    by choosing $\widetilde{\hat x^\minus_\tau}_{+1} := \widetilde{A} \, \widetilde{\hat x_\tau} + \widetilde{B} u^{\sf b}_\tau + \widetilde{L}_{\sf L}(y^{\sf b}_\tau - \mathbf{C} \hat{\mathbf{x}}^\minus_\tau)$,
    where we used $B = \Phi \mathbf{B}$ in Claim \ref{CLAIM:OrigModel_AuxModel_relation}, $\mathbf{B} = \Phi_\auxiliary \widetilde{B}$ in Claim \ref{CLAIM:AuxModel_TildeModel_relation}, $L_{\sf L} = \Phi \mathbf{L}_{\sf L}$ and $\mathbf{L}_{\sf L} = \Phi \widetilde{L}_{\sf L}$ in Claim \ref{CLAIM:equivalence_of_optimization_problems:Kalman_gain}, and $A \hat x^\minus_\tau = \Phi \mathbf{A} \hat{\mathbf{x}}_\tau$ and $\mathbf{A} \hat{\mathbf{x}}^\minus_\tau = \Phi_\auxiliary \widetilde{A} \, \widetilde{\hat x^\minus_\tau}$ by applying Claim \ref{CLAIM:equivalence_of_optimization_problems:corollaries} with $(v, \boldsymbol{v}, \tilde{v}) \gets (\hat x^\minus_\tau, \hat{\mathbf{x}}^\minus_\tau, \widetilde{\hat x^\minus_\tau})$ given (b) of $t=\tau$.
    By induction on $t$, (a) and (b) hold for $t \in \integer_{[\kappa, \kappa^\plus]}$, and (c) and (d) hold for $t \in \integer_{[\kappa, \kappa^\plus)}$.

    We finally show (e). Notice the following relations, 
\begin{subequations}
\label{Eq:CLAIM:equivalence_of_control_algorithms:relation_12}
\begin{align}
\label{Eq:CLAIM:equivalence_of_control_algorithms:relation_1}
    \hat x^\minus_{\kappa^\plus} &= \Phi \Phi_\auxiliary \widetilde{\hat x^\minus_\kappa}_{^\plus}, &
    \hat{\mathbf{x}}^\minus_{\kappa^\plus} &= \Phi_\auxiliary \widetilde{\hat x^\minus_\kappa}_{^\plus} \\
\label{Eq:CLAIM:equivalence_of_control_algorithms:relation_2}
    \thickbar{x}_{\kappa^\plus} &= \Phi \Phi_\auxiliary \widetilde{\thickbar{x}_\kappa}_{^\plus}, &
    \thickbar{\mathbf{x}}_{\kappa^\plus} &= \Phi_\auxiliary \widetilde{\thickbar{x}}_{\kappa^\plus}
\end{align}\end{subequations}
    where \eqref{Eq:CLAIM:equivalence_of_control_algorithms:relation_1} is due to (b) with $t = \kappa^\plus$, and \eqref{Eq:CLAIM:equivalence_of_control_algorithms:relation_2} follows from Claim \ref{CLAIM:equivalence_of_optimization_problems:nominal_state} with $k = \kappa$ and $t = \kappa^\plus$.
According to \eqref{Eq:initial_condition_iteration} applied in Algorithm \ref{ALGO:SMPC} and \eqref{Eq:DDModel:initial_condition_iteration} applied in Algorithm \ref{ALGO:DDSMPC}, we have
\begin{align}
\label{Eq:CLAIM:equivalence_of_control_algorithms:relation_3}
    \mu^{\sf \hat x}_{\kappa^\plus} = \hat x^\minus_{\kappa^\plus}, \;\;
    \mu^{\sf \bar x}_{\kappa^\plus} = \thickbar{x}_{\kappa^\plus}, \;\;
    \boldsymbol{\mu}^{\sf \hat x}_{\kappa^\plus} = \hat{\mathbf{x}}^\minus_{\kappa^\plus}, \;\;
    \boldsymbol{\mu}^{\sf \bar x}_{\kappa^\plus} = \thickbar{\mathbf{x}}_{\kappa^\plus}.
\end{align}
Combining \eqref{Eq:CLAIM:equivalence_of_control_algorithms:relation_12} and \eqref{Eq:CLAIM:equivalence_of_control_algorithms:relation_3}, with $\Phi_\original = \Phi \Phi_\auxiliary$ via Claim \ref{CLAIM:Phi_relation}, we obtain \eqref{Eq:initial_condition_relation} with $k = \kappa^\plus$ where we select $\widetilde{\mu}^{\sf \,\hat x}_{\kappa^\plus} := \widetilde{\hat x^\minus_\kappa}_{^\plus}$ and $\widetilde{\mu}^{\sf \,\bar x}_{\kappa^\plus} := \widetilde{\thickbar{x}}_{\kappa^\plus}$. This shows (e).
\end{proof}

We finish the proof by showing that the results (a)-(e) in Claim \ref{CLAIM:PROPOSITION:equivalence_of_control_algorithms} are true for all control steps $\kappa \in \{0, N_{\rm c}, 2N_{\rm c}, \ldots\}$, by induction on $\kappa$.
\textbf{Base Case.} For $\kappa=0$, condition i) of Claim \ref{CLAIM:PROPOSITION:equivalence_of_control_algorithms} holds given that both processes start with a common initial state $x_0$, and condition ii) of Claim \ref{CLAIM:PROPOSITION:equivalence_of_control_algorithms} holds due to Assumption \ref{ASSUMPTION:parameter_choice_of_DDSMPC_algorithm}(d) and due to the selections $(\mu^{\sf \hat x}_0, \mu^{\sf \bar x}_0) \gets (\mu^{\sf x}_\initial, \mu^{\sf x}_\initial)$ in Algorithm \ref{ALGO:SMPC} and $(\boldsymbol{\mu}^{\sf \hat x}_0, \boldsymbol{\mu}^{\sf \bar x}_0) \gets (\boldsymbol{\mu}^{\sf x}_\initial, \boldsymbol{\mu}^{\sf x}_\initial)$ in Algorithm \ref{ALGO:DDSMPC}. With both conditions i) and ii) satisfied, the results (a)-(e) of Claim \ref{CLAIM:PROPOSITION:equivalence_of_control_algorithms} are true for $\kappa = 0$.
\textbf{Inductive Step.} Assume for $\kappa = \overline \kappa$ that results (a)-(e) of Claim \ref{CLAIM:PROPOSITION:equivalence_of_control_algorithms} are true. Due to (a) and (e) of Claim \ref{CLAIM:PROPOSITION:equivalence_of_control_algorithms} for $\kappa = \overline \kappa$, the assumptions i) and ii) in Claim \ref{CLAIM:PROPOSITION:equivalence_of_control_algorithms} for $\kappa = \overline \kappa + N_{\rm c}$ are satisfied, thereby ensuring that the results (a)-(e) of Claim \ref{CLAIM:PROPOSITION:equivalence_of_control_algorithms} with $\kappa = \overline \kappa + N_{\rm c}$ are true.
By induction on $\kappa$, we have the results (a)-(e) of Claim \ref{CLAIM:PROPOSITION:equivalence_of_control_algorithms} for all control steps $\kappa \in \{0, N_{\rm c}, 2N_{\rm c}, \ldots\}$. The results (a), (c), (d) for all $\kappa$ suffices to prove the theorem.
\end{proof}

\Versions{\bibliography{References/brevalias, References/r_MPC, References/r_DDPC}}{

}


\begin{IEEEbiography}[{\includegraphics[width=1in,height=1.25in,trim={.4in 0 .4in 0},clip,keepaspectratio]{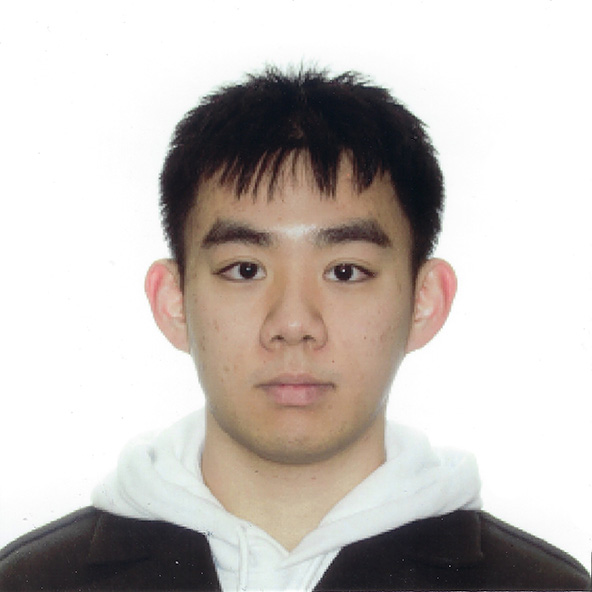}}]{Ruiqi Li} (S'22) received the B.Sc. degree in Honours Physics from the University of Waterloo, ON, Canada in 2019 and the B.Sc. degree in physics from Beijing Institute of Technology, Beijing, China in 2019. He is currently working towards the Ph.D. degree in Electrical and Computer Engineering at the University of Waterloo, ON, Canada.
His research interest includes data-driven control, model predictive control and optimization.
\end{IEEEbiography}

\begin{IEEEbiography}[{\includegraphics[width=1in,height=1.25in,clip,keepaspectratio]{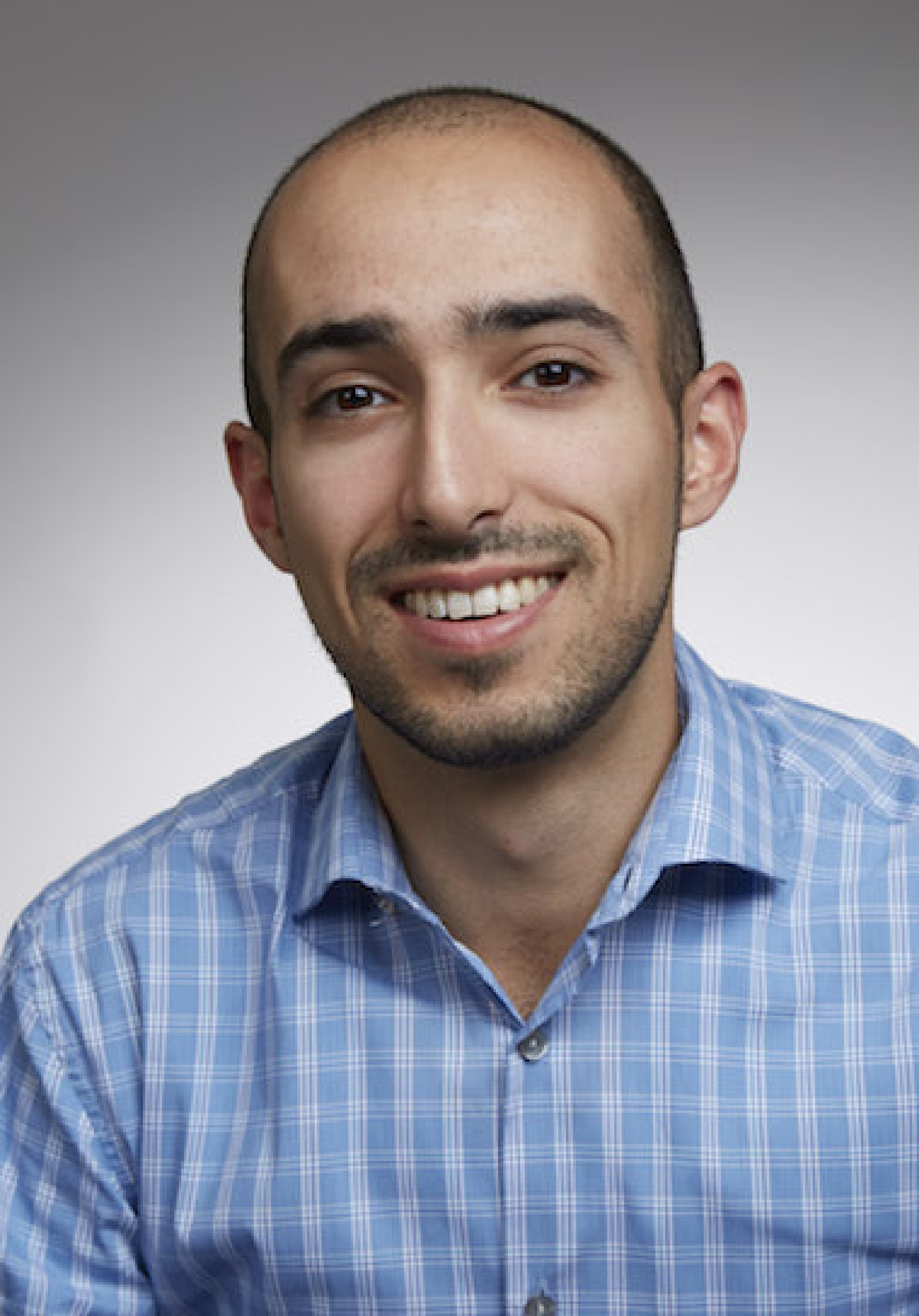}}]{John W. Simpson-Porco} (S'10--M'15--SM'22) received the B.Sc. degree in engineering physics from Queen's University, Kingston, ON, Canada in 2010, and the Ph.D. degree in mechanical engineering from the University of California at Santa Barbara, Santa Barbara, CA, USA in 2015.
He is currently an Associate Professor of Electrical and Computer Engineering at the University of Toronto, Toronto, ON, Canada. He was previously an Assistant Professor at the University of Waterloo, Waterloo, ON, Canada and a visiting scientist with the Automatic Control Laboratory at ETH Z\"{u}rich, Z\"{u}rich, Switzerland. His research focuses on feedback control theory and applications of control in modernized power grids.

Prof. Simpson-Porco is a recipient of the Automatica Paper Prize, the Center for Control, Dynamical Systems and Computation Best Thesis Award, and the IEEE PES Technical Committee Working Group Recognition Award for Outstanding Technical Report. He received the Early Researcher Award from the Province of Ontario in 2024.
\end{IEEEbiography}

\begin{IEEEbiography}[{\includegraphics[width=1in,height=1.25in,clip,keepaspectratio]{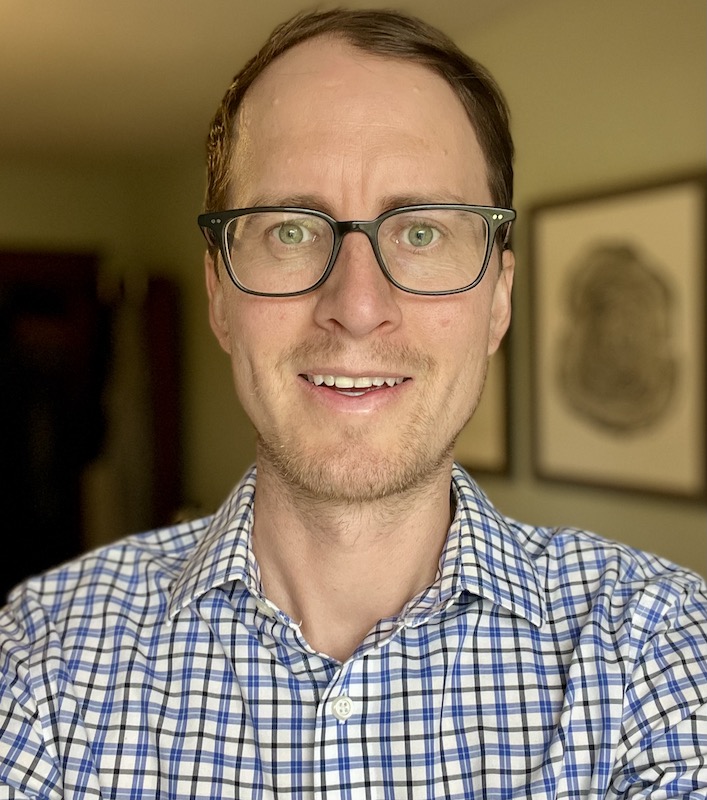}}]{Stephen L.\ Smith} (S'05--M'09--SM'15) received the B.Sc. degree in engineering physics from Queen’s University, Canada, in 2003, the
M.A.Sc. degree in electrical and computer engineering from the University of Toronto, Canada, in 2005, and the Ph.D. degree in mechanical engineering from the University of California, Santa Barbara, USA, in 2009. 
He is currently a Professor in the Department of Electrical and Computer Engineering at the University of Waterloo, Canada, where he holds a Canada Research Chair in Autonomous Systems. He is also Co-Director of the Waterloo Artificial Intelligence Institute. From 2009 to 2011 he was a Postdoctoral Associate in the Computer Science and Artificial Intelligence Laboratory at MIT. He received the Early Researcher Award from the the Province of Ontario in 2016, the NSERC Discovery Accelerator Supplement Award in 2015, and Outstanding Performance Awards from the University of Waterloo in 2016 and 2019.

He is a licensed Professional Engineer (PEng), an Associate Editor of the IEEE Transactions on Robotics and the IEEE Open Journal of Control Systems.  He was previously Associate Editor for the IEEE Transactions on Control of Network Systems (2017 - 2022), and was a General Chair of the 2021 30th IEEE International Conference on Robot and Human Interactive Communication (RO-MAN). His main research interests lie in control and optimization for autonomous systems, with a particular emphasis on robotic motion planning and coordination.
\end{IEEEbiography}

\end{document}